\documentclass{vldb}

\usepackage[full]{optional}
\usepackage{graphicx}
\usepackage{xcolor}
\usepackage{balance}
\usepackage{geometry} 
\opt{short}{\newcommand{\revise}[1]{\textcolor{blue}{#1}}} 
\opt{full}{\newcommand{\revise}[1]{#1}} 
\usepackage{soul} 
\setstcolor{red}
\definecolor{mygray}{gray}{0.75}

\usepackage{cite}

\usepackage{amsmath,amssymb,amsfonts,amsthm}

\usepackage{textcomp}

\usepackage{subfigure}
\usepackage{booktabs}
\usepackage{multirow}
\usepackage{url}
\usepackage{hyperref}
\usepackage{cleveref} 

\usepackage{breakurl}
\usepackage{mdwlist}
\usepackage{mathtools} 
\usepackage{extarrows}
\usepackage{enumitem}
\usepackage{thmtools} 
\usepackage{thm-restate}
\usepackage[hang,flushmargin,bottom]{footmisc}
\usepackage{algorithm}
\usepackage[noend]{algpseudocode}
\usepackage{eqparbox}
\usepackage{fancyhdr}

\let\oldReturn\Return
\renewcommand{\Return}{\State\oldReturn}
\crefformat{footnote}{#2\footnotemark[#1]#3}

\theoremstyle{definition}
\declaretheorem{theorem}
\newtheorem{observation}{Observation}
\newtheorem{lemma}{Lemma}
\newtheorem{definition}{Definition}
\newtheorem{example}{Example}

\DeclareMathOperator*{\argmax}{arg\,max}
\DeclareMathOperator*{\argmin}{arg\,min}
\newcommand{\para}[1]{\vskip 0.06in\noindent {\bf #1} }
\newcommand{\etal}{\textit{et al}. }
\newcommand{\notion}{CID}
\newcommand{\notionfull}{Customized Interactive Display}
\newcommand{\scenario}{SAVG}
\newcommand{\scenariofull}{Social-Aware VR Group}
\newcommand{\configuration}{\scenario\ $k$-Configuration}
\newcommand{\configurationfull}{\scenariofull\ $k$-Configuration}
\newcommand{\personalized}{personalized}
\newcommand{\group}{group}
\newcommand{\preference}{subgroup-by-preference}
\newcommand{\friendship}{subgroup-by-friendship}
\newcommand{\prob}{SVGIC}
\newcommand{\problem}{\underline{S}ocial-aware \underline{V}R \underline{G}roup-\underline{I}tem \underline{C}onfiguration}
\newcommand{\probtwo}{SVGIC-ST}
\newcommand{\problemtwo}{\underline{S}ocial-aware \underline{V}R \underline{G}roup-\underline{I}tem \underline{C}onfiguration with \underline{T}eleportation and \underline{S}ize constraint}
\newcommand{\algo}{AVG}
\newcommand{\algofull}{\underline{A}lignment-aware \underline{V}R Sub\underline{g}roup Formation}
\newcommand{\algor}{AVG}

\newcommand{\algod}{AVG-D}
\newcommand{\algodfull}{\underline{D}eterministic \underline{A}lignment-aware \underline{V}R Sub\underline{g}roup Formation}

\newcommand{\fragsol}{fractional solution}
\newcommand{\roundingitem}{focal item}
\newcommand{\roundingslot}{focal slot}
\newcommand{\roundingalpha}{grouping threshold}

\newcommand{\stageone}{config phase}

\newcommand{\stagetwo}{focus phase}

\newcommand{\stagethree}{CSF}
\newcommand{\stagethreefull}{Co-display Subgroup Formation}
\newcommand{\weight}{utility factor}

\opt{short}{
 \setlength{\textfloatsep}{5pt plus 1pt minus 2pt}
 \setlength{\topsep}{5pt plus 1pt minus 2pt} 
 \setlength{\abovedisplayskip}{5pt plus 1pt minus 2pt}
 \setlength{\belowdisplayskip}{5pt plus 1pt minus 2pt}
 \setlength{\floatsep}{5pt plus 1pt minus 2pt}
}

\vldbTitle{Optimizing Item and Subgroup Configurations for Social-Aware VR Shopping}
\vldbAuthors{Shao-Heng Ko, Hsu-Chao Lai, Hong-Han Shuai, De-Nian Yang, Wang-Chien Lee, Philip S. Yu}
\vldbDOI{}
\vldbVolume{14}
\vldbNumber{xxx}
\vldbYear{2020}

\begin{document}

\title{Optimizing Item and Subgroup Configurations for Social-Aware VR Shopping}

\numberofauthors{1} 

\DeclareRobustCommand*{\IEEEauthorrefmark}[1]{\raisebox{0pt}[0pt][0pt]{\textsuperscript{\footnotesize\ensuremath{\ifcase#1\or *\or \dagger\or \ddagger\or%
    \mathsection\or \mathparagraph\or \|\or **\or \dagger\dagger%
    \or \ddagger\ddagger \else\textsuperscript{\expandafter\romannumeral#1}\fi}}}}

\author{
\alignauthor
Shao-Heng Ko\IEEEauthorrefmark{1},
Hsu-Chao Lai\IEEEauthorrefmark{1}\IEEEauthorrefmark{2},
Hong-Han Shuai\IEEEauthorrefmark{2},\\
De-Nian Yang\IEEEauthorrefmark{1},
Wang-Chien Lee\IEEEauthorrefmark{3},
Philip S. Yu\IEEEauthorrefmark{4}\\
\affaddr{\IEEEauthorrefmark{1}\textit{Academia Sinica, Taiwan} \IEEEauthorrefmark{2}\textit{National Chiao Tung University, Taiwan}\\
\IEEEauthorrefmark{3}\textit{The Pennsylvania State University, USA} \IEEEauthorrefmark{4}\textit{University of Illinois at Chicago, USA}}
\email{\IEEEauthorrefmark{1}\{arsenefrog, hclai0806, dnyang\}@iis.sinica.edu.tw \IEEEauthorrefmark{2}hhshuai@g2.nctu.edu.tw\\
\IEEEauthorrefmark{3}wlee@cse.psu.edu \IEEEauthorrefmark{4}psyu@cs.uic.edu}
}

\maketitle

\begin{abstract}
Shopping in VR malls has been regarded as a paradigm shift for E-commerce, but most of the conventional VR shopping platforms are designed for a single user. In this paper, we envisage a scenario of VR group shopping, which brings major advantages over conventional group shopping in brick-and-mortar stores and Web shopping: 1) configure flexible display of items and partitioning of subgroups to address individual interests in the group, and 2) support social interactions in the subgroups to boost sales. Accordingly, we formulate the \problem\ (\prob) problem to configure a set of displayed items for flexibly partitioned subgroups of users in VR group shopping. We prove \prob\ is $\mathsf{APX}$-hard \revise{and also $\mathsf{NP}$-hard to approximate within $\frac{32}{31} - \epsilon$}. We design an approximation algorithm based on the idea of \stagethreefull\ (\stagethree) to configure proper items for display to different subgroups of friends. Experimental results on real VR datasets and a user study with hTC VIVE manifest that our algorithms outperform baseline approaches by at least 30.1\% of solution quality.
\end{abstract}

\section{Introduction} \label{sec:intro}

Virtual Reality (VR) has emerged as a disruptive technology for social \cite{FBsocial}, travel \cite{VRtourism}, and E-commerce applications. \revise{Particularly, a marketing report about future retails from Oracle \cite{Oracle} manifests that 78\% of online retailers already have implemented or are planning to implement VR and AI. Recently, International Data Corporation (IDC) forecasts the worldwide spending on VR/AR to reach 18.8 billion USD in 2020 \cite{IDC}, including \$1.5 billion in retails. It also foresees the VR/AR market to continue an annual growth rate of 77\% through at least 2023. Moreover, shopping in VR malls is regarded as a paradigm shift for E-commerce stores, evident by emerging VR stores such as Amazon's VR kiosks \cite{Amazon}, eBay and Myer's VR department store \cite{eBay}, Alibaba Buy+ \cite{Alibaba}, and IKEA VR Store \cite{IKEA}.} Although these VR shopping platforms look promising, most of them are designed only for a single user instead of a group of friends, who often appear in brick-and-mortar stores. As a result, existing approaches for configuring the displayed items in VR shopping malls are based on personal preference (similar to online web shopping) without considering potential social discussions amongst friends on items of interests, which is reported as beneficial for boosting sales in the marketing literature \cite{MY14, ZX14, XZ18}. In this paper, with the support of \textit{\notionfull} (\notion), we envisage the scenario of group shopping with families and friends in the next-generation VR malls, where item placement is customized flexibly in accordance with both user preferences and potential social interactions during shopping.

The \notion\ technology \cite{RL14,LH18} naturally enables VR group shopping systems with two unique strengths:  1) \textit{Customization}. \revise{IKEA (see video \cite{IKEA} at 0:40) and Lowe's \cite{Lowes} respectively launch VR store applications where the displayed furniture may adapt to the preferences of their users, which attracts more than 73\% of the users to recommend IKEA VR store \cite{IKEAsteam}. Alibaba's \cite{AlibabaForbes} and eBay's VR stores \cite{ebaymacy} also devote themselves to provide personalized product offers. According to a marketing survey, 79\% of surveyed US, UK, and China consumers are likely to visit a VR store displaying customized products \cite{forbessurvey}.} Similar to group traveling and gaming in VR, the virtual environments (VEs) for individual users in VR group shopping need not be identical. While it is desirable to have consistent high-level layout and user interface for all users, the displayed items for each user can be customized based on her preferences. As CID allows different users to view different items at the same display slot, personalized recommendation is supported.

2) \textit{Social Interaction}. While a specific slot no longer needs to display the same items to all users, users viewing a common item may engage a discussion on the item together, potentially driving up the engagement and purchase conversion \cite{ZX14, XZ18}. As a result, the displayed items could be tailored to maximize potential discussions during group shopping. \revise{A survey of 1,400 consumers in 2017 shows that 65\% of the consumers are excited about VR shopping, while 54\% of them acknowledge social shopping as their ways to purchase products~\cite{wssurvey}. The L.E.K. consulting survey \cite{leksurvey} shows that 70\% of 1,000 consumers who had already experienced VR technology are strongly interested in virtual shopping with friends who are not physically present. Embracing the trend, Shopify \cite{shopify} and its technical partner Qbit \cite{Qbit} build a social VR store supporting attendance of multiple users with customized display \cite{Qbitstore}.} In summary, compared with brick-and-mortar shopping, VR group shopping can better address the preferences of individuals in the group due to the new-found flexibility in item placement, which can be configured not only for the group as a whole but also for individuals and subgroups. On the other hand, compared with conventional E-commerce shopping on the Web, VR group shopping can boost sales by facilitating social interactions and providing an immersive experience.

\begin{figure}[tp]
	\centering
	\subfigure[b][Comparison of different approaches.] {\
		\centering \includegraphics[width = 0.95 \columnwidth]{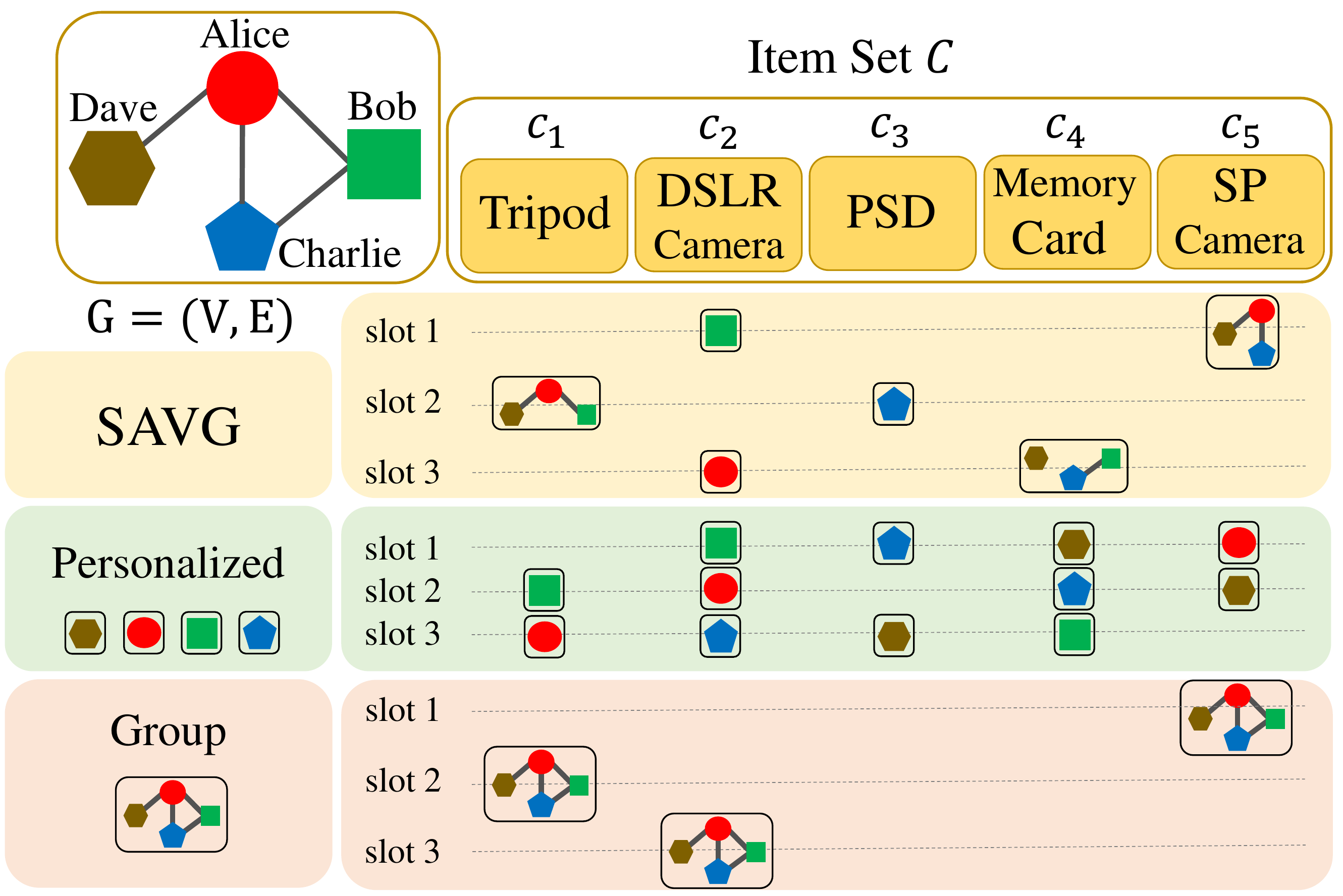}
		\label{fig:example} }
	\subfigure[b][\revise{Item assignments of different approaches.}] {\
		\centering \includegraphics[width = 0.95 \columnwidth]{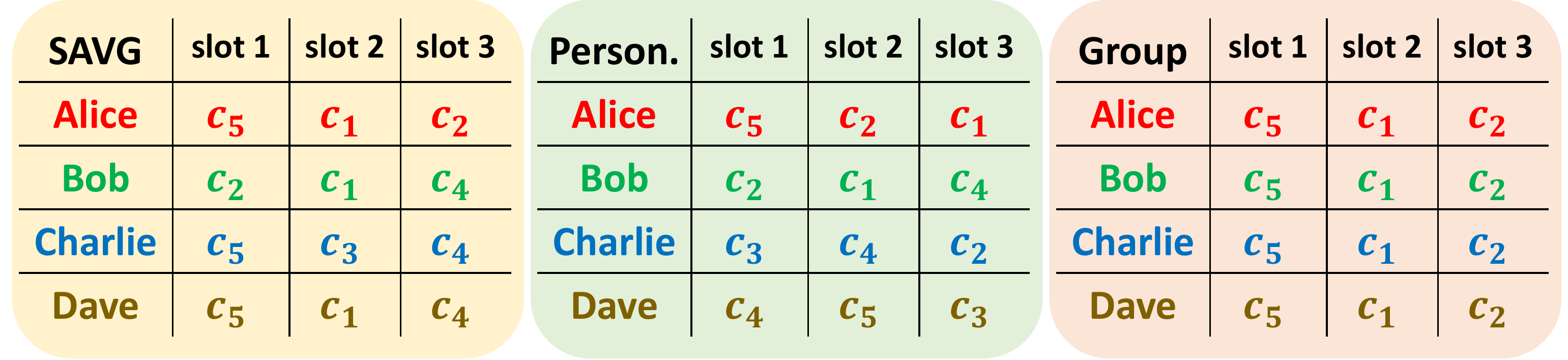}
		\label{fig:example_add} }
	\subfigure[b][\centering Alice's view at slot 1.] {\
		\centering \includegraphics[width = 0.45 \columnwidth]{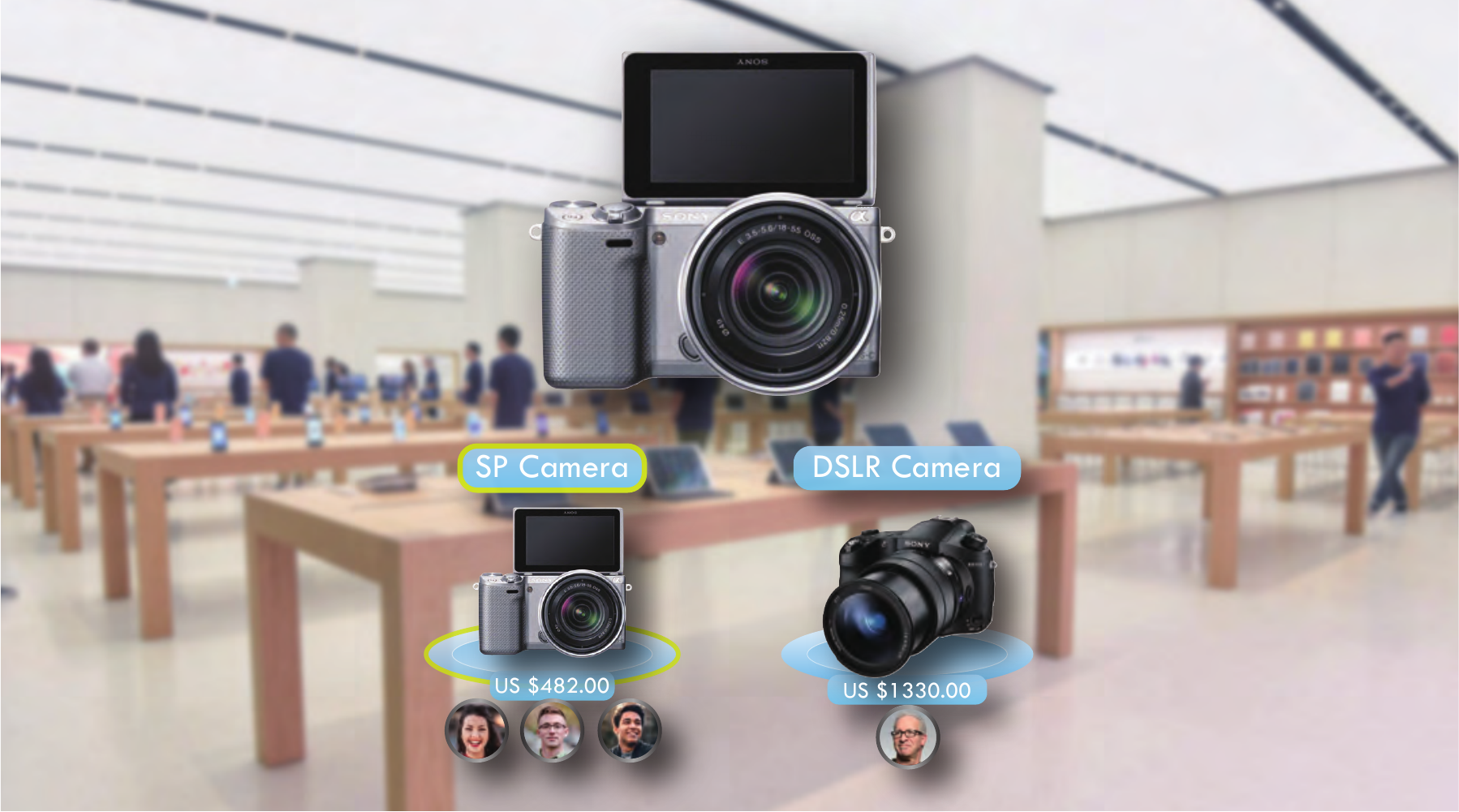}
		\label{fig:example_slot1} } 
		\hfill
	\subfigure[b][\centering Alice's view at slot 2.] {\
		\centering \includegraphics[width = 0.45 \columnwidth]{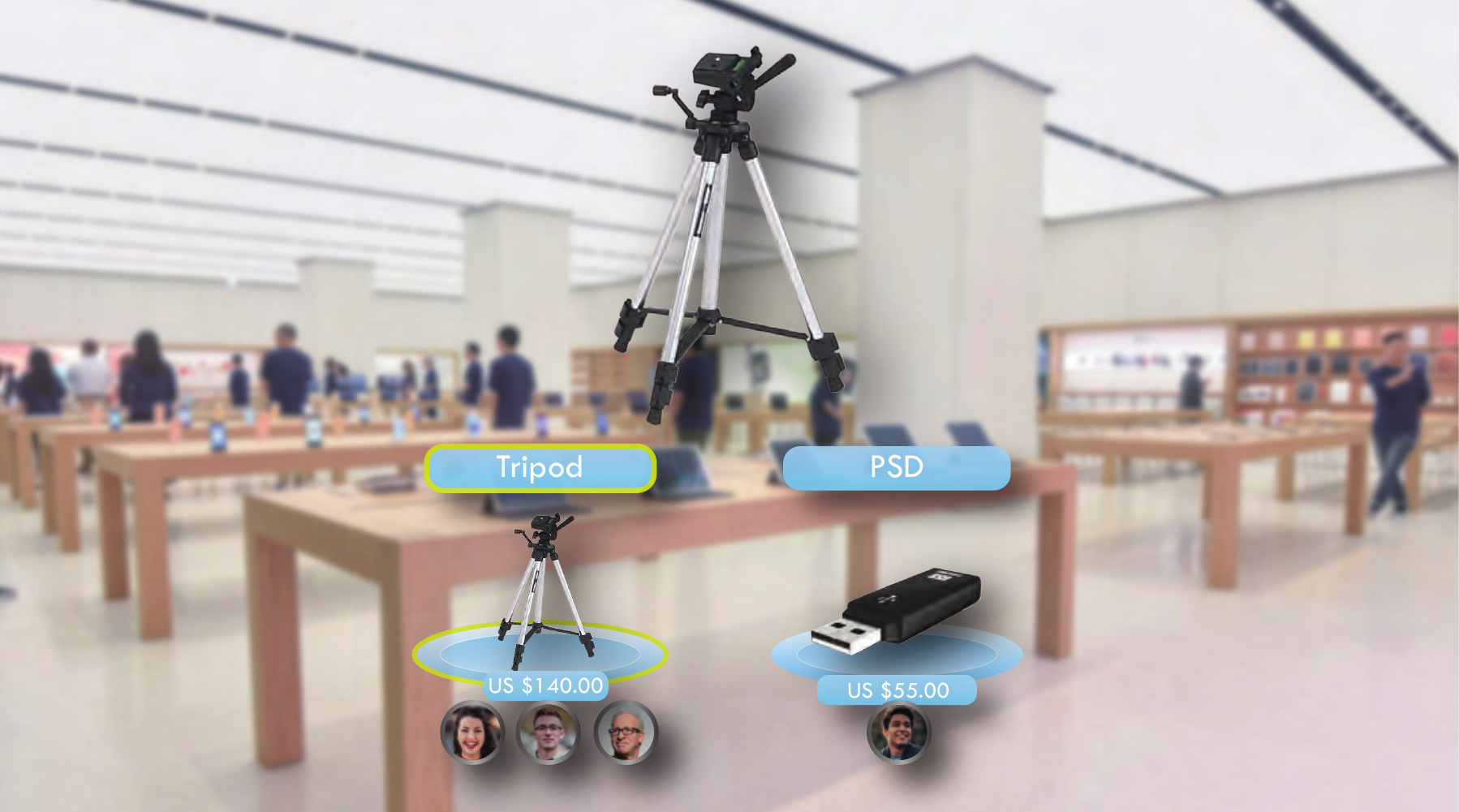}
		\label{fig:example_slot2} }
	\caption{Illustrative example.}
	\label{fig:comp_example}
\end{figure}

\revise{Encouraged by the above evidence, we make the first attempt to formulate the problem of configuring displayed items for {\em \scenariofull} (\scenario) shopping. } Our strategy is to meet the item preferences of individual users via customization while enhancing potential discussions by displaying common items to a shopping group (and its subgroups). For customization, one possible approach is to use existing personalized recommendation techniques \cite{HR16AAAI, XH17, JC17SIGIR} to infer individual user preferences, and then retrieve the top-$k$ favorite items for each user. However, this \personalized\ approach fails to promote items of common interests that may trigger social interactions. 
To encourage social discussions, 
conventional group recommendation systems \cite{LH14,SQ16,SD17, CYS18} may be leveraged to retrieve a bundled itemset for all users. Nevertheless, by presenting the same configuration for the whole group of all users, this approach may sacrifice the diverse individual interests. 
In the following example, we illustrate the difference amongst the aforementioned approaches, in contrast to the desirable \scenario\ configuration we target on.

\begin{example} [Illustrative Example] \label{example:illustrative}
Figure \ref{fig:example} depicts a scenario of group shopping 
for a VR store of digital photography. At the upper left is a social network $G = (V, E)$ of four VR users, Alice, Bob, Charlie, and Dave (indicated by red circles, green squares, blue pentagons, and brown hexagons, respectively). On top is an item set $C$ consisting of five items: tripod, DSLR camera, portable storage device (PSD), memory card, and self-portrait (SP) camera. Given three display slots, the shaded areas, corresponding to different configuration approaches, illustrate how items are displayed to individuals or subgroups (in black rectangles), respectively. For instance, in \scenario, the SP camera is displayed at slot 1 to Alice, Charlie, and Dave to stimulate their discussion. \revise{Figure \ref{fig:example_add} shows the same configuration as Figure \ref{fig:example} by depicting the individual item assignments for each user with different approaches.} 

A configuration based on the \personalized\ approach is shown in the light-green shaded area. It displays the top-3 items of interests to individual users based on their preferences (\revise{which is consistent with the numerical example shown in Table \ref{table:sat} in Section \ref{sec:problem}}). This configuration, aiming to increase the exposure of items of interests to customers, does not have users seeing the same item at the same slot. Next, the configuration shaded in light-orange, based on the \group\ approach, displays exactly the same items to every user in the group. While encouraging discussions in the whole group, this configuration may sacrifice some individual interests and opportunities to sell some items, e.g., Dave may not find his favorite item (the memory card). 
Aiming to strike a balance between the factors of individual preferences and social interactions, the \scenario\ configuration forms subgroups flexibly across the displayed slots, having some users seeing the same items at the same slots (to encourage discussions), yet finding items of individual interests at the remaining slots. For example, the tripod is displayed at slot 2 to all users except for Charlie, who sees the PSD on his own. The DSLR camera is displayed to Bob at slot 1 and to Alice at slot 3, respectively, satisfying their individual interests. Figures \ref{fig:example_slot1} and \ref{fig:example_slot2} show Alice's view at slot 1 and slot 2, respectively, in this configuration. As shown, Alice is co-displayed the SP camera with Charlie and Dave at slot 1 (informed by the user icons below the primary view of the item), then co-displayed the Tripod with Bob and Dave at slot 2. As shown, \scenario\ shopping displays items of interests to individuals or different subgroups at each slot, and thereby is more flexible than other approaches.\hfill \qedsymbol
\end{example}

As illustrated above, in addition to identifying which items to be displayed in which slots, properly partitioning subgroups by balancing both factors of personal preferences and social discussions is critical for the above-depicted \scenario\ shopping. In this work, we define the notion of \textit{\configuration}, which \textit{specifies the partitioned subgroups (or individuals) and corresponding items for each of the allocated $k$ slots.} We also introduce the notion of \textit{co-display} that represents users sharing views on common items. Formally, we formulate a new optimization problem, namely \textit{\problem} (\prob), to find the optimal \configuration\ that maximizes the overall 1) \textit{preference utility} from users viewing their allocated items and 2) \textit{social utility} from all pairs of friends having potential discussions on co-displayed items, where the basic preference utility of an individual user on a particular item and the basic social utility of two given users on a co-displayed item are provided as inputs. Meanwhile, we ensure that no duplicated items are displayed at different slots to a user. 
The problem is very challenging due to a complex trade-off between prioritizing personalized item configuration and facilitating social interactions. 
Indeed, we prove that \revise{\prob\ is $\mathsf{NP}$-hard to approximate \prob\ within a ratio of $\frac{32}{31}-\epsilon$.}


To solve \prob, we first present an Integer Program (IP) as a baseline to find the exact solution which requires superpolynomial time. To address the efficiency issue while ensuring good solution quality, we then propose a novel approximation algorithm, namely \textit{\algofull} (\algo), to approach the optimal solution in polynomial time. Specifically, \algo\ first relaxes the IP to allow fragmented \configuration s, and obtains the optimal (fractional) solution of the relaxed linear program. It then assigns the fractional decision variables derived from the optimal solution as the \textit{\weight s} for various potential allocations of user-item-slot in the solution. Items of high \weight s are thus desirable as they are preferred to individuals or encouraging social discussions. Moreover, by leveraging an idea of \textit{dependent rounding}, \algo\ introduces the notion of \textit{\stagethreefull} (\stagethree) to strike a balance between personal preferences and social interactions in forming subgroups. \stagethree\ forms a \textit{target subgroup} of socially connected users with similar interests to display an item, according to a randomized \textit{\roundingalpha} on \weight s to determine the membership of the target subgroup. With \stagethree, \algo\ finds the subgroups (for all slots) and selects appropriate items simultaneously, and thereby is more effective than other approaches that complete these tasks separately in predetermined steps. Theoretically, we prove that \algo\ achieves 4-approximation in expectation. We then show that \algo\ can be derandomized into a deterministic 4-approximation algorithm for \prob. \revise{We design further LP transformation and sampling techniques to improve the efficiency of \algo.}


Next, we enrich \prob\ by taking into account some practical VR operations and constraints. We define the notion of \textit{indirect co-display} to capture the potential social utility obtained from friends displayed a common item at two different slots in their VEs, where discussion is facilitated via the \textit{teleportation} \cite{BE16} function in VR. We also consider a \textit{subgroup size constraint} on the partitioned subgroups at each display slot due to practical limits in VR applications. Accordingly, we formulate the \problemtwo\ (\probtwo) problem, and prove that \probtwo\ is unlikely to be approximated within a constant ratio in polynomial time. Nevertheless, we extend \algo\ to support \probtwo\ and guarantee feasibility of the solution. \revise{Moreover, we have extended \algo\ to support a series of practical scenarios. 1) \textit{Commodity values}. Each item is associated with a commodity value to maximize the total profit.  2) \textit{Layout slot significance}. Each slot location is associated with a different significance (e.g., center is better) according to retailing researches \cite{slot, shelf}. 3) \textit{Multi-View Display}, where a user can be displayed multiple items in a slot, including one default personally preferred item in the \textit{primary view}, and multiple items to view with friends in \textit{group views}, whereas the primary and group views can be freely switched. 4) \textit{Generalized social benefits}, where social utility can be measured from not only the pairwise (each pair of friends) but also group-wise (any group of friends) perspectives. 5) \textit{Subgroup change}, where the fluctuations (i.e., change of members) between the partitioned subgroups at consecutive slots are limited to ensure smooth social interactions as the elapse of time. 6) \textit{Dynamic scenario}, where users dynamic join and leave the system with different moving speeds. Finally, we have also identified \textit{Social Event Organization} (SEO) as another important application of the targeted problem. \opt{short}{Due to space limit, the details of \probtwo\ and the above extensions are presented in the full version \cite{Online} of this paper.}} 

The contributions of this work are summarized as follows:
\begin{itemize}
    \sloppy \item We coin the notion of \configuration\ under the context of VR group shopping and formulate the \prob\ problem, aiming to return an \configuration\ that facilitates social interactions while not sacrificing the members' individual preferences. We prove \prob\ is \revise{$\mathsf{APX}$-hard and also $\mathsf{NP}$-hard to approximate within $\frac{32}{31} - \epsilon$.}
    \item We systematically tackle \prob\ by introducing an IP model and designing an approximation algorithm, \algo, based on the idea of \stagethreefull\ (\stagethree) that leverages the flexibility of \notion\ to partition subgroups for each slot and display common items to subgroup members. 
    \opt{full}{\item We define the \probtwo\ problem to incorporate teleportation and subgroup size constraint. We prove \probtwo\ admits no constant-ratio polynomial-time approximation algorithms unless ETH fails. We extend \algo\ to obtain feasible solutions of \probtwo.}
    \item A comprehensive evaluation on real VR datasets and a user study implemented in Unity and hTC VIVE manifest that our algorithms outperform the state-of-the-art recommendation schemes by at least 30.1\% in terms of solution quality. \opt{short}{The code and VR implementation can be downloaded at \cite{Online}.}
\end{itemize}

This paper is organized as follows. Section \ref{sec:related} reviews the related work. \opt{short}{Section \ref{sec:problem} formally defines the notion of \configuration, then formulates the \prob\ problem and an IP model.}\opt{full}{Section \ref{sec:problem} formally defines the notion of \configuration, then formulates the \prob\ and \probtwo\ problems and IP models for them.} Then, Section \ref{sec:algo} details the proposed \algo\ algorithm and the theoretical results. \revise{Section \ref{sec:extension} discusses the directions for extensions.} Section \ref{sec:exp} reports the experimental results, and Section \ref{sec:conclusion} concludes this paper.

\section{Related Work} \label{sec:related}
\para{Group Recommendation.}
\sloppy Various schemes for estimating group preference have been proposed by aggregating features from different users in a group \cite{SQ16,SD17}. Cao \etal \cite{DC18} propose an attention network for finding the group consensus. However, sacrificing personal preferences, the above group recommenders assign a unified set of items for the entire group based only on the aggregate preference without considering social topologies. For advanced approaches, recommendation-aware group formation \cite{SBR15} forms subgroups according to item preferences. SDSSel \cite{CYS18} finds dense subgroups and diverse itemsets. Shuai \etal \cite{HHS18WWW} customize the sequence of shops without considering \notion. However, the above recommenders find \textit{static} subgroups, i.e., a universal partition, and still assign a fixed configuration of items to each subgroup, where every subgroup member sees the same item in the same slot. In contrast, the \scenario\ approach considered in this paper allows the partitioned subgroups to vary across all display slots and thereby is more flexible than the above works.

\para{Personalized Recommendation.}
Personalized recommendation, a cornerstone of E-commerce, has been widely investigated. A recent line of study integrates deep learning with Collaborative Filtering (CF) \cite{LZ15,JC17SIGIR} on heterogeneous applications \cite{RG17}, while Bayesian Personalized Ranking \cite{SR09, HC16CIKM} is proposed to learn the ranking of recommended items. However, the above works fail to consider social interactions. While social relations have been leveraged to infer preferences of products \cite{WXZ14KDD,WXZ16TKDE}, POIs \cite{HL16KDD}, and social events \cite{YL18TOIS}, they do not take into account the social interactions among users in recommendation of items. Thus, they fail to consider the trade-off between social and personal factors. In this paper, we exploit the preferences obtained from such studies to serve as inputs for the tackled problems.

\para{Social Group Search and Formation.}
Research on finding various groups from online social networks (OSNs) for different purposes has drawn increasing attention in recent years. 
Community search finds similar communities containing a given set of query nodes \cite{JL17, LC18}. 
Group formation organizes a group of experts with low communication costs and specific skill combinations \cite{MK13, RSS13}. In addition, organizing a group in location-based social networks based on certain spatial factors has also gained more attention \cite{CYS16TKDE, CYS16TKDD}. 
The problems tackled in the above studies are fundamentally different from the scenario in this paper, since they focus on retrieving only \textit{parts} of the social network according to some criteria, whereas item selection across multiple slots is not addressed. Instead, \scenario\ group shopping aims to configure item display for all VR shopping users, whereas the subgroups partition the entire social network.

\revise{
\para{Related Combinatorial Optimization Problems.}
\prob\ is related to the \textit{Multi-Labeling} (ML) \cite{multilabeling} problem and its variations, including {Multiway Partition} \cite{multilabeling, multiwaypartition}, \textit{Maximum Happy Vertices/Edges} (MHV/MHE) \cite{homophyly}, and \textit{Multiway Cut} \cite{multiwaycut} in graphs. 
\opt{short}{In Section \ref{subsec:hardness_ip}, we revisit the challenging combinatorial nature of the proposed \prob\ problem and its relation with the related problems, whereas detailed introduction of each problem, as well as a summary table, are presented in the full version \cite{Online}.}\opt{full}{In Section \ref{subsec:combinatorial}, we revisit the challenging combinatorial nature of the proposed \prob\ problem and its relation with the related problems with detailed introduction of all problems and a summary table.}
}

\section{Problem Formulation and Hardness Results} \label{sec:problem}

\newcommand{\specialcell}[2][c]{\begin{tabular}[#1]{@{}c@{}}#2\end{tabular}}
\small
\begin{table*}[t]
\centering
\caption{Preference and social utility values in Example \ref{example:running}.}
\label{table:sat}
\begin{tabular}{c||c|c|c|c||c|c|c|c|c|c|c|c}
      & $\text{p}(u_A,\cdot)$ & $\text{p}(u_B,\cdot)$ & $\text{p}(u_C,\cdot)$ & $\text{p}(u_D,\cdot)$ & \specialcell{$\tau(u_A,$\\$u_B,\cdot)$} & \specialcell{$\tau(u_A,$\\$u_C,\cdot)$} & \specialcell{$\tau(u_A,$\\$u_D,\cdot)$} & \specialcell{$\tau(u_B,$\\$u_A,\cdot)$} & \specialcell{$\tau(u_B,$\\$u_C,\cdot)$} & \specialcell{$\tau(u_C,$\\$u_A,\cdot)$} & \specialcell{$\tau(u_C,$\\$u_B,\cdot)$} & \specialcell{$\tau(u_D,$\\$u_A,\cdot)$} \\ \hline
$c_1$ & 0.8 & 0.7 & 0 & 0.1               & 0.2 & 0 & 0.2 & 0.2 & 0 & 0 & 0.1 & 0.3\\
$c_2$ & 0.85 & 1.0 & 0.15 & 0          & 0.05 & 0.05 & 0.05 & 0.05 & 0.05 & 0.05 & 0.05 & 0.05\\
$c_3$ & 0.1 & 0.15 & 0.7 & 0.3            & 0.1 & 0.1 & 0.1 & 0.1 & 0.1 & 0.1 & 0.1 & 0.05\\
$c_4$ & 0.05 & 0.2 & 0.6 & 1.0            & 0  & 0 & 0.05 & 0.05 & 0.2 & 0.05 & 0.2 & 0\\
$c_5$ & 1.0 & 0.1 & 0.1 & 0.95            & 0.05 & 0.3 & 0.2 & 0.05 & 0 & 0.3 & 0.05 & 0.25\\
\end{tabular}
\end{table*}
\normalsize

\opt{short}{In this section, we first define the notion of \configuration\ and then formally introduce the \prob\ problem. We also prove its hardness of approximation, and introduce an integer program for \prob\ as a cornerstone for the \algo\ algorithm.}
\opt{full}{In this section, we first define the notion of \configuration\ and then formally introduce the \prob\ and \probtwo\ problems. We also prove their hardness of approximation, and introduce integer programs for \prob\ and \probtwo\ as cornerstones for the \algo\ algorithm.}

\subsection{Problem Formulation} \label{subsec:problem}

Given a collection $\mathcal{C}$ of $m$ items (called the \textit{Universal Item Set}), a directed social network $G=(V,E)$ with a vertex set $V$ of $n$ users (i.e., shoppers to visit a VR store) and an edge set $E$ specifying their social relationships. In the following, we use the terms \textit{user set} and \textit{shopping group} interchangeably to refer to $V$, and define \configuration\ to represent the partitioned subgroups and the corresponding items displayed at their allocated slots.

\begin{definition}{\textit{\configurationfull} \newline (\configuration).}
Given $k$ display slots for displaying items to users in a VR shopping group, an \configuration\ is a function $\mathbf{A}(\cdot, \cdot): (V \times [k]) \rightarrow \mathcal{C}$ mapping a tuple $(u,s)$ of a user $u$ and a slot $s$ to an item $c$. $\mathbf{A}(u,s) = c$ means that the configuration has item $c$ displayed at slot $s$ to user $u$, and $\mathbf{A}(u,:) = \langle \mathbf{A}(u,1), \mathbf{A}(u,2), \ldots, \mathbf{A}(u,k) \rangle$ are the $k$ items displayed to $u$. Furthermore, the function is regulated by a \textit{no-duplication constraint} which ensures the $k$ items displayed to $u$ are distinct, i.e., $\mathbf{A}(u,s) \neq \mathbf{A}(u,s'), \forall s \neq s'$.
\end{definition}

For shopping with families and friends, previous research \cite{MY14, ZX14, XZ18} demonstrates that discussions and interactions are inclined to trigger more purchases of an item. \revise{Specifically, the VR technology supporting social interactions has already been employed in existing VR social platforms, e.g., VRChat \cite{VRChat}, AltspaceVR \cite{AltspaceVR}, and Facebook Horizon \cite{fbHorizon}, where users, represented by customized avatars, interact with friends in various ways: 1) \textit{Movement and gesture.} Commercial VR devices, e.g., Oculus Quest and HTC Vive Pro, support six-degree-of-freedom headsets and hand-held controllers to detect the user movements and gestures in real-time. These movements are immediately updated to the server and then perceived by other users. For example, a high-five gesture mechanism is used in Rec Room \cite{RecRoom} to let users befriend others. 2) \textit{Emotes and Emojis.} Similar to emoticons in social network platforms, VR users can present their emotions through basic expressions (happy face, sad face, etc.) or advanced motions (e.g., dance, backflip, and die in VRChat \cite{VRChatGuide}). 3) \textit{Real sound.} This allows users to communicate with friends verbally through their microphones. As such, current VR applications support immersive and seamless social interactions, rendering social interactions realistic in shopping systems.}

\revise{With non-customized user environments in brick-and-mortar stores, discussing a specific item near its location is the default behavior. On the other hand, as users are displayed completely customized items in personalized e-commerce (e.g., the 20 items appearing on the main page of a shopping site can be completely disjoint for two friends), it often requires a ``calibration'' step (e.g., sharing the exact hyperlink of the product page) to share the view on the target item before e-commerce users can engage in discussions on some items. Since shopping users are used to having discussions on the common item they are viewing at the identical slot in their own Virtual Environments (VEs), we devote to enhance the possibility of this intuitive setting in our envisioned VR shopping VEs even though other settings are also possible in VR.}

\revise{Therefore, for a pair of friends displayed a common item, it is most convenient to display the item at the same slot in the two users' respective VEs such that they can locate the exact position easily (e.g., ``come to the second shelf; you should see this.''). Accordingly, upon viewing any item, it is expected that the user interface indicates a list of friends being co-displayed the specific item with the current user, so that she is aware of the potential candidates to start a discussion with.} To display an item at the same slot to a pair of friends, we define the notion of \textit{co-display} as follows.

\begin{definition}{\textit{Co-display} ($u \xleftrightarrow[s]{c} v$).} \label{def:co-display} Let $u \xleftrightarrow[s]{c} v$ represent that users $u$ and $v$ are \textit{co-displayed} 
an identical item $c$ at slot $s$, i.e., $\mathbf{A}(u,s) = \mathbf{A}(v,s) = c$. Let $u \xleftrightarrow{c} v$ denote that there exists at least one $s$ such that $u \xleftrightarrow[s]{c} v$.
\end{definition}

\noindent Naturally, the original group of users are partitioned into subgroups in correspondence with the displayed items. \revise{That is, for each slot $s \in [k]$, the \configuration\ implicitly partitions the user set $V$ into a collection of disjoint subsets $V^s = \{ V^s_1, V^s_2, \ldots, V^s_{\text{N}_{\text{p}}(s)} \}$, where $\text{N}_{\text{p}}(s)$ is the number of partitioned subgroups at slot $s$, such that $u \xleftrightarrow[s]{c} v$ if and only if $u,v \in V^s_i, i = 1, \ldots, \text{N}_{\text{p}}(s)$.}

\revise{Under the scenario of brick-and-mortar group shopping, previous research \cite{XL05, LQS13, MY14, LB16, XZ18} indicates that the satisfaction of a group shopping user is affected by two factors: \textit{personal preference} and \textit{social interaction}. Furthermore, empirical research on dynamic group interaction in retails finds that shopping while discussing with others consistently boosts sales \cite{MY14}, while intra-group talk frequency has a significant impact on users' shopping tendencies and purchase frequencies \cite{ZX14, XZ18}. Accordingly, we capture the overall satisfaction by combining 1) the aggregated personal preferences and 2) the retailing benefits from facilitating social interactions and discussions. We note that, from an optimization perspective, these two goals form a trade-off because close friends/families may have diverse personal preferences over items. One possible way to simultaneously consider both factors is to use an end-to-end machine learning-based approach, which generates the \textit{user embeddings} and \textit{item embeddings} \cite{DC18} and designs a neural network aggregator to derive the total user satisfaction. However, this approach requires an algorithm to generate candidate configurations for ranking and relies on a huge amount of data to train the parameters in the aggregator. On the other hand, previous research \cite{HHS13, XW18WWW} has demonstrated that a \textit{weighted combination} of the preferential and social factors is effective in measuring user satisfaction, where the weights can be directly set by a user or implicitly learned from existing models \cite{ZT14,YL18TOIS}. Another possible objective is a weighted sum of the previous objective metric (the aggregated preference) with the aggregated social satisfaction as the constraint, or vice versa, which is equivalent to our problem after applying Lagrangian relaxation to the constraint. Therefore, we follow \cite{HHS13, XW18WWW} to define the \scenario\ utility as a combination of \textit{aggregated preference utility} and \textit{aggregated social utility}, weighted by a parameter $\lambda$.} 


Specifically, for a given pair of friends $u$ and $v$, i.e., $(u,v) \in E$, and an item $c \in \mathcal{C}$, let $\text{p}(u,c) \geq 0$ denote the \textit{preference utility} of user $u$ for item $c$, and let $\tau(u,v,c) \geq 0$ denote the \textit{social utility} of user $u$ from viewing item $c$ together with user $v$, 
where $\tau(u,v,c)$ can be different from $\tau(v,u,c)$.
Following \cite{HHS13, XW18WWW}, the \textit{\scenario\ utility} is defined as follows.

\begin{definition}{\textit{\scenario\ utility} \label{def:savg_utility} ($\mathbf{w}_{\textbf{A}}(u,c)$).}
Given an \configuration\ $\textbf{A}$, the \scenario\ utility of user $u$ on item $c$ in $\textbf{A}$ represents a combination of the preference and social utilities, where $\lambda \in [0,1]$ represents their relative weight.

\opt{short}{\vskip -10pt}

\begin{align*}
    \text{w}_{\textbf{A}}(u,c) = (1-\lambda)\cdot\text{p}(u,c) + \lambda\cdot\sum\limits_{v | u \xleftrightarrow{c} v} \tau(u,v,c)
\end{align*}

\opt{short}{\vskip -30pt}

\end{definition}

\noindent \revise{The preference and social utility values can be directly given by the users or obtained from social-aware personalized recommendation learning models \cite{ZT14,YL18TOIS}. Those models, learning from purchase history, are able to infer $(u,c)$ and $(u,v,c)$ tuples to relieve users from filling those utilities manually.} The weight $\lambda$ can also be directly set by a user or implicitly learned from existing models \cite{ZT14,YL18TOIS}. 
\revise{To validate the effectiveness of this objective model, we have conducted a user study in real \notion\ VR shopping system which shows high correlations between user satisfaction and the proposed problem formulation (Please see Section \ref{sec:user_study} for the results).} 

\begin{example} \label{example:running}
Revisit Example \ref{example:illustrative} where $\mathcal{C} = \{c_1, c_2, \ldots, c_5\}$ and let $V = \{u_A,u_B,u_C,u_D\}$ denote Alice, Bob, Charlie, and Dave, respectively. Table \ref{table:sat} summarizes the given preference and social utility values. For example, the preference utility of Alice to the tripod is $\text{p}(u_A, c_1) = 0.8$. In the \scenario\ 3-configuration described at the top of Figure \ref{fig:comp_example}, $\mathbf{A}(u_A,\cdot)=\langle c_5, c_1, c_2 \rangle$, meaning that the SP camera, the tripod, and the DSLR camera are displayed to Alice at slots 1, 2, and 3, respectively. Let $\lambda = 0.4$. At slot 2, Alice is \textit{co-displayed} the tripod with Bob and Dave. Therefore, $\mathbf{w}_{\textbf{A}}(u_A, c_1) = 0.6 \cdot 0.8 + 0.4 \cdot (0.2 + 0.2) = 0.64$. \hfill \qedsymbol
\end{example}

\opt{full}{
\begin{table}[t!]
\caption{Notations used in \prob.}
\begin{center}
\begin{tabular}{|c|l|}\hline
Symbol & Description \\\hline\hline
$n$ & number of users \\\hline
$m$ & number of items \\\hline
$k$ & number of display slots \\\hline
$G = (V,E)$ & social network \\\hline
$V$ & user set/shopping group \\\hline
$E$ & (social) edge set \\\hline
$\mathcal{C}$ & Universal Item Set \\\hline
$\mathbf{A}(\cdot, \cdot)$ & \configuration\ \\\hline
$\mathbf{A}^{\ast}$ & optimal \configuration\ \\\hline
$\text{p}(u,c)$ & preference utility \\\hline
$\tau(u,v,c)$ & social utility \\\hline
$\text{w}_{\textbf{A}}(u,c)$ & \scenario\ utility \\\hline
$\lambda$ & weight of social utility\\\hline
$\text{p'}(u,c)$ & scaled preference utility \\\hline
$u \xleftrightarrow[s]{c} v$ & Co-display (at slot $s$) \\\hline
$u \xleftrightarrow{c} v$ & Co-display (at some slot) \\\hline
$\text{N}_{\text{p}}(s)$ & number of user subgroups at slot $s$ \\\hline
$V^s$ & collection of subgroups at slot $s$ \\\hline
\multirow{2}{*}{$x^c_{u,s}$} & decision variable of \\ & $u$ displayed $c$ at slot $s$\\\hline
\multirow{2}{*}{$x^c_u$} & decision variable of \\ & $u$ displayed $c$\\\hline
\multirow{2}{*}{$y^c_{e,s}$} & decision variable of \\ & $u,v$ co-displayed $c$ at slot $s$\\\hline
\multirow{2}{*}{$y^c_{e}$} & decision variable of \\ & $u,v$ co-displayed $c$\\\hline
\end{tabular}
\label{tab:symbols_prob1}
\end{center}
\end{table}
}

\noindent We then formally introduce the \prob\ problem as follows. \opt{full}{All used notations in \prob\ are summarized in Table \ref{tab:symbols_prob1}.}

\vskip 0.06in
\hrule
\vskip 0.03in

\noindent \textbf{Problem:} Social-aware VR Group-Item Configuration.

\noindent \textbf{Given:} A social network $G=(V,E)$, a universal item set $\mathcal{C}$, preference utility $\text{p}(u,c)$ for all $u \in V$ and $c \in \mathcal{C}$, social utility $\tau(u,v,c)$ for all $(u,v) \in E$ and $c \in \mathcal{C}$, the weight $\lambda$, and the number of slots $k$.

\noindent \textbf{Find:} An \configuration\ $\mathbf{A}^{\ast}$ to maximize the total \scenario\ utility 
\opt{short}{\vskip -0.3in}
$$\sum_{u \in V} \sum_{c \in \mathbf{A}^{\ast}(u,\cdot)} \text{w}_{\mathbf{A}^{\ast}}(u,c),$$
\noindent where $\text{w}_{\mathbf{A}^{\ast}}(u,c)$ is the \scenario\ utility defined in Definition \ref{def:savg_utility}, and the no-duplication constraint is ensured.

\vskip 0.03in
\hrule
\vskip 0.06in

\opt{short}{Finally, we remark here that the above additive objective function in \prob\ can be viewed as a generalization of more limited variations of objectives. For instance, some objectives in personal/group recommendation systems without considering social discussions, e.g., the AV semantics in \cite{SBR15}, 
are special cases of \prob\ where $\lambda = 0$. We detail this in Section \ref{subsec:problem} in the full version \cite{Online}.}

\opt{full}{
\begin{table}[t!]
\caption{Notations used in \probtwo.}
\begin{center}
\begin{tabular}{|c|l|}\hline
Symbol & Description \\\hline\hline
$\text{w'}_{\textbf{A}}(u,c)$ & \scenario\ utility with indirect co-display \\\hline
$u \xleftrightarrow[s,s']{c} v$ & Indirect Co-display (at slots $s,s'$) \\\hline
$u \xleftrightarrow[\text{ind}]{c} v$ & Indirect Co-display (at some slots) \\\hline
$\text{d}_{\text{tel}}$ & discount factor \\\hline
$M$ & subgroup size constraint \\\hline
$z^c_e$ & decision variable of $u,v$ co-displayed $c$\\\hline
\end{tabular}
\label{tab:symbols_prob2}
\end{center}
\end{table}
}



\opt{full}{
Based on \prob, we have the following observation on several special cases of this problem. The \personalized\ approach, introduced in Section \ref{sec:intro}, is a special case of \prob\ with $\lambda$ set to $0$. In this special case, it is sufficient to leverage personalized (or social-aware personalized) recommendation systems to infer the personal preferences, because the problem reduces into trivially recommending the top-$k$ favorite items for each user. On the other hand, the \group\ approach is also a special case of \prob\ with $\text{N}_{\text{p}}(s) = 1$ for each $s$, i.e., all users view the same item at the same slot. However, the combinatorial complexity of the solution space for \prob\ is $\Theta(m^{nk})$ due to the flexible partitioning of subgroups across different display slots in \scenario. Formally, $\textbf{A}(u,s) = \textbf{A}(v,s), \forall u,v \in V, s = 1,\ldots, k$, which corresponds to a group recommendation scenario.  

Due to the complicated interplay and trade-off of preference and social utility, the following theorem manifests that the optimal objective value of \prob\ can be significantly larger than the optimal objective values of the special cases above, corresponding to \personalized\ and \group\ approaches. Specifically, for each problem instance $\mathbf{I}$, let $\text{OPT}(\mathbf{I})$ denote the optimal objective value of \prob. Let $\text{OPT}_{\text{P}}(\mathbf{I})$ and $\text{OPT}_{\text{G}}(\mathbf{I})$ respectively denote the optimal objective values of the special cases of \personalized\ and \group\ approaches in the above discussion, i.e., \prob\ with $\lambda = 1$, and \prob\ with $\text{N}_{\text{p}}(s) = 1$ for each slot $s$, respectively.

\begin{restatable}{theorem}{gap} \label{thm:problem_gap}
Given any $n$, there exists 1) a \prob\ instance $\mathbf{I}_\text{G}$ with $n$ users such that $\frac{\text{OPT}(\mathbf{I}_\text{G})}{\text{OPT}_{\text{G}}(\mathbf{I}_\text{G})} = n$. 2) a \prob\ instance $\mathbf{I}_\text{P}$ such that $\frac{\text{OPT}(\mathbf{I}_\text{P})}{\text{OPT}_{\text{P}}(\mathbf{I}_\text{P})} \geq \frac{\lambda}{1-\lambda} \cdot \frac{n-1}{2} = O(n)$.
\end{restatable}

\begin{proof}
We construct $\mathbf{I}_\text{G}$ and $\mathbf{I}_\text{P}$ as follows. For $\mathbf{I}_\text{G}$, each user $u_i$ in $\mathbf{I}_\text{G}$ prefers exactly $k$ items $\mathcal{C}_i = \{ c_i, c_{n+i}, \ldots, c_{(k-1)n+i} \}$, i.e., $\text{p}(u_i, c_j) = 1$ for $c_j \in \mathcal{C}_i$, and 0 for $c_j \notin \mathcal{C}_i$. Let $E = \emptyset$, and $\tau(u,v,c)$ thereby is 0 for all $(u,v,c)$. In the optimal solution of the \group\ approach, each slot $s$ contributes at most $1-\lambda$ to the total \scenario\ utility because every item is preferred by exactly one user. By contrast, the optimal solution in \prob\ recommends $u_i$ all the $k$ items in $\mathcal{C}_i$. Each $u_i$ then contributes $k(1-\lambda)$ to the total \scenario\ utility, and thus $\frac{\text{OPT}(\mathbf{I}_G)}{\text{OPT}_{\text{G}}(\mathbf{I}_G)} = \frac{nk(1-\lambda)}{k(1-\lambda)} = n$. 

\sloppy For $\mathbf{I}_\text{P}$, each user $u_i$ in $\mathbf{I}_\text{P}$ prefers the items in $\mathcal{C}_i$ only \textit{slightly more than} all the other items. Specifically, let $\text{p}(u_i, c_j) = 1$ for $c_j \in \mathcal{C}_i$, and $\text{p}(u_i, c_j) = 1 - \epsilon$ for $c_j \notin \mathcal{C}_i$. Let $G$ represent a complete graph here, and $\tau(u,v,c) = 1$ for all $u,v$ and $c$. The optimal solution of the \personalized\ approach selects all the $k$ items in $\mathcal{C}_i$ for each $u_i$ and contributes $(1-\lambda) \cdot n \cdot k$ to the total \scenario\ utility. Note that co-display is not facilitated because no common item is chosen. By contrast, one possible solution in \prob\ recommends to all users the same $k$ items in $\mathcal{C}_1$, where each item contributes $(1-\lambda) \cdot ((1-\epsilon)n + \epsilon)$ to the total preference and $\lambda \cdot \frac{n(n-1)}{2}$ to the social utility. Thus, as $\epsilon$ approaches 0, $\frac{\text{OPT}(\mathbf{I}_P)}{\text{OPT}_{\text{P}}(\mathbf{I}_P)} \geq \frac{k[(1-\lambda) \cdot ((1-\epsilon)n + \epsilon) + \lambda \cdot \frac{n(n-1)}{2}]}{(1-\lambda) \cdot n \cdot k} \approx 1 + \frac{\lambda}{1-\lambda} \cdot \frac{n-1}{2} = O(n)$, regarding $\lambda$ as a constant.
\end{proof}
}

\subsection{Indirect Co-display and \probtwo}

\revise{In \prob, we characterize the merit of social discussions by the social utilities to a pair of users when they see the same item at the \textit{same slot}. On the other hand, if a common item is displayed at different slots in two users' VEs, a prompt discussion is more difficult to initiate since the users need to locate and move to the exact position of the item first. \textit{Teleportation} \cite{BE16}, widely used in VR tourism and gaming applications, allows VR users to directly transport between different positions in the VE. Thus, for a pair of users Alice and Bob co-displayed an item at different slots, as long as they are aware of where the item is displayed, one (or both) of them can teleport to the respective display slot of the item to trigger a discussion. To model this event that requires more efforts from users, we further propose a generalized notion of \textit{indirect co-display} to consider the above-described scenario.}

\revise{\begin{definition}{\textit{Indirect Co-display.} ($u \xleftrightarrow[s,s']{c} v$).} Let $u \xleftrightarrow[s,s']{c} v$ denote that users $u$ and $v$ are \textit{indirectly} co-displayed an item $c$ at slots $s$ and $s'$ respectively in their VEs, i.e., $\mathbf{A}(u,s) = \mathbf{A}(v,s') = c$. We use $u \xleftrightarrow[\text{ind}]{c} v$ to denote that there exist different slots $s \neq s'$ such that $u \xleftrightarrow[s,s']{c} v$.
\end{definition}}

\opt{full}{\noindent Note that we use $u \xleftrightarrow{c} v$ to denote \textit{direct co-display}, i.e., there exists at least one $s$ such that $u \xleftrightarrow[s]{c} v$, while using $u \xleftrightarrow[\text{ind}]{c} v$ to denote \textit{indirect co-display}. Also note that these two events are mutually exclusive due to the no-duplication constraint. Also note that $u \xleftrightarrow[s,s]{c} v$ is equivalent to direct co-display: $u \xleftrightarrow[s]{c} v$.}

\revise{As social discussions on indirectly co-displayed items are less immersive and require intentional user effort, we introduce a discount factor $\text{d}_{\text{tel}} < 1$ on teleportation to downgrade the social utility obtained via indirect co-display. Therefore, the total \scenario\ utility incorporating indirect co-display is as follows.}

\opt{short}{
\revise{In the full version of this paper \cite{Online}, we explicitly formulate a more generalized \probtwo\ problem with the above notions of indirect display and a modified \scenario\ utility, where the problem is associated with an additional \textit{subgroup size constraint} to accommodate practical limits in existing VR applications. Tables of all used notations in both \prob\ and \probtwo\ are also provided in the full version \cite{Online}.}
}
\opt{full}{
\begin{definition}{\textit{\scenario\ utility with indirect co-display} \label{def:savg_utility_unaligned} \newline ($\text{w'}_{\textbf{A}}(u,c)$).}
Given an \configuration\ $\textbf{A}$, the \textit{\scenario\ utility with indirect co-display} of user $u$ on item $c$ in $\textbf{A}$ is defined as
\begin{align*}
     \text{w'}_{\textbf{A}}(u,c) &= (1-\lambda)\cdot\text{p}(u,c)\\ &+ \lambda\cdot \big( \sum\limits_{v | u \xleftrightarrow{c} v} \tau(u,v,c) + \sum\limits_{v | u \xleftrightarrow[\text{ind}]{c} v} \text{d}_{\text{tel}} \cdot \tau(u,v,c) \big).
\end{align*}
\end{definition}

Finally, we note that existing VR applications often do not accommodate an unlimited number of users in one shared VE, e.g., VRChat \cite{VRChat} allows at most 16 users, and IrisVR \cite{IrisVR} allows up to 12 users within the same VE. Therefore, we further incorporate a \textit{subgroup size constraint} $M$ as an upper bound on the number of users directly co-displayed the same item at the same location. More specifically, for all slot $s \in [k]$, all partitioned subgroups $V^s_i \in V^s$ contain no more than $M$ users; or equivalently, given a fixed slot $s$ and a fixed item $c$, $\mathbf{A}^{\ast}(u,s) = c$ for at most $M$ different users $u$. In the following, we introduce the \probtwo\ problem that incorporates indirect co-display and the subgroup size constraint.

\vskip 0.06in
\hrule
\vskip 0.03in

\noindent \textbf{Problem: \problemtwo} (\probtwo).

\noindent \textbf{Given:} A social network $G=(V,E)$, a universal item set $\mathcal{C}$, preference utility $\text{p}(u,c)$ for all $u$ and $c$, social utility $\tau(u,v,c)$ for all $u$,$v$, and $c$, the weight $\lambda$, the number of slots $k$, the teleportation discount $\text{d}_{\text{tel}}$, and the subgroup size upper bound $M$.

\noindent \textbf{Find:} An \configuration\ $\mathbf{A}^{\ast}$ to maximize the total \scenario\ utility (with indirect co-display)
$$\sum_{u \in V} \sum_{c \in \mathbf{A}^{\ast}(u,\cdot)} \text{w'}_{\mathbf{A}^{\ast}}(u,c),$$
\noindent \sloppy such that the subgroup size constraint is satisfied, and $\text{w'}_{\mathbf{A}^{\ast}}(u,c)$ is defined as in Definition \ref{def:savg_utility_unaligned}.

\vskip 0.03in
\hrule
\vskip 0.06in

Notations in \probtwo\ are summarized in Table \ref{tab:symbols_prob2}.
}

\subsection{Hardness Result and Integer Program} \label{subsec:hardness_ip}
\opt{short}{In the following, we state theoretical hardness results for \prob\ and briefly discuss its relation with other combinatorial problems.}
\opt{full}{In the following, we prove that \prob\ is $\mathsf{APX}$-hard and also $\mathsf{NP}$-hard to approximate within a ratio of $\frac{32}{31} - \epsilon$. We then prove that \probtwo\ does not admit any constant-factor polynomial-time approximation algorithms, unless the Exponential Time Hypothesis (ETH) fails.}

\begin{theorem} \label{thm:new_hardness}
\revise{\prob\ is $\mathsf{APX}$-hard. Moreover, for any $\epsilon > 0$, it is $\mathsf{NP}$-hard to approximate \prob\ within a ratio of $\frac{32}{31} - \epsilon$.}
\end{theorem}

\opt{short}{
\vskip -15pt
\begin{proof}
\textcolor{blue}{We prove the theorem with a gap-preserving reduction from the $\mathsf{MAX-E3SAT}$ problem~\cite{MAX3SAT}. Due to space constraint, please see the full version \cite{Online} for the proof.}
\end{proof}
\vskip -15pt

\label{para:comparison_combinatorial} \textcolor{blue}{\para{Comparison with Related Combinatorial Problems.} Among the related \textit{Multi-Labeling} (ML)-type problems detailed in Section \ref{sec:related}, \prob\ is particularly related to the graph coloring problem MHE~\cite{homophyly}, where, different from traditional graph coloring, vertices are encouraged to share the same color with neighbors. The optimization goal of MHE is to maximize the number of \textit{happy edges}, i.e., edges with same-color vertices. Regarding the displayed items in \prob\ as colors (labels) in MHE, the social utility achieved in \prob\ is closely related to the number of preserved happy edges in MHE, and \prob\ thereby encourages partitioning all users into dense subgroups to preserve the most social utility. However, \prob\ is more difficult than the ML-type problems due to the following reasons. 1) The ML-type problems find a strict partition that maps each entity/vertex to only one color (label), while \prob\ assigns $k$ items to each user, implying that any direct reduction from the above problems can only map to the $k=1$ special case in \prob. 2) Most of the ML-type problems do not discriminate among different labels, i.e., switching the assigned labels of two different subgroups does not change their objective functions. This corresponds to the special case of \prob\ where all preference utility $\text{p}(
u, c)$ and social utility $\tau(u, v, c)$ do not depend on the item $c$. 3) Most of the ML-type problems admit a partial labeling (pre-labeling) in the input such that some entities have predefined fixed labels (otherwise labeling every entity with the same label is optimal, rendering the problem trivial), while \prob\ does not specify any item to be displayed to specific users. However, \prob\ requires $k$ items for each user; moreover, even in the $k=1$ special case, simply displaying the same item to all users in \prob\ is not optimal due to the item-dependent preference and social utility.
}}
\opt{full}{
In the following, we prove the hardness result with a gap-preserving reduction from the $\mathsf{MAX-E3SAT}$ problem~\cite{MAX3SAT}. Given a conjunctive normal form (CNF) formula $\phi$ with exactly 3 literals in each clause, $\mathsf{MAX-E3SAT}$ asks for a truth assignment of all variables such that the number of satisfied clauses in $\phi$ is maximized. The following lemma (from~\cite{MAX3SAT}) states the hardness of approximation of $\mathsf{MAX-E3SAT}$.

\begin{figure}
    \centering
    \includegraphics[width=\columnwidth]{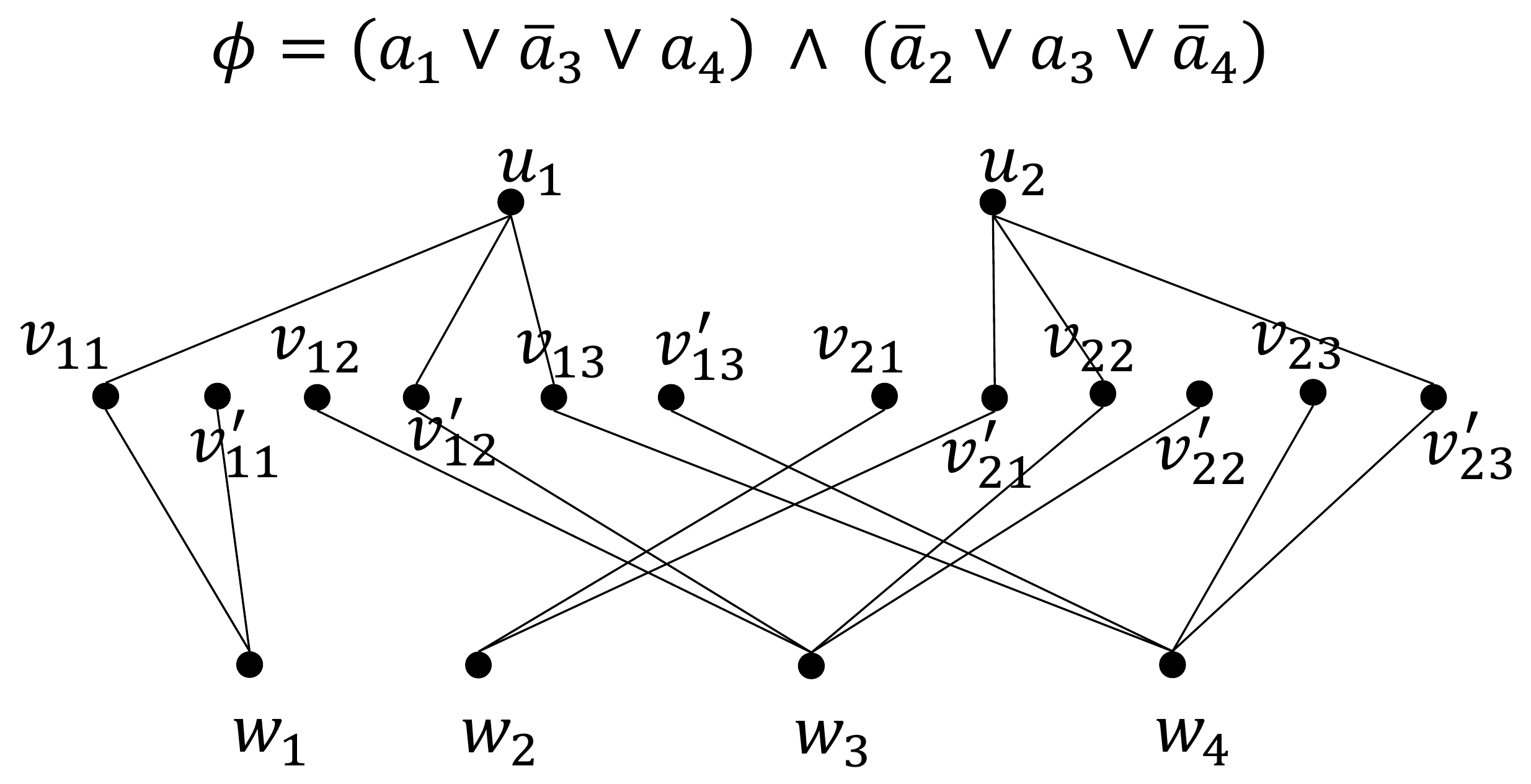}
    \caption{An illustration instance built for $\mathsf{MAX-E3SAT}$}
    \label{fig:reduction}
\end{figure}

\begin{lemma} \label{lemma:hardnesse3sat}
$\mathsf{MAX-E3SAT}$ is $\mathsf{NP}$-hard to approximate within a ratio of $\frac{8}{7}-\epsilon$.
\end{lemma}

Our reduction builds on the following lemma.

\begin{lemma} \label{lemma:gapreduction}
There exists a gap-preserving reduction from $\mathsf{MAX-E3SAT}$ to $\mathsf{\prob}$ which translates from a formula $\phi$ to an $\mathsf{\prob}$ instance $\mathbf{I}$ such that
\begin{itemize}
    \item If at least $\chi$ clauses can be satisfied in $\phi$, then the optimal value of $\mathbf{I}$ is at least $8 \cdot \chi$,
    \item If less than $\frac{7}{8} \cdot \chi$ clauses can be satisfied in $\phi$, then the optimal value of $\mathbf{I}$ is less than $\frac{31}{4} \cdot \chi$.
\end{itemize}
\end{lemma}

\begin{proof}
Given a $\mathsf{MAX-E3SAT}$ instance $\phi$ with $n_{\text{var}}$ Boolean variables $\mathtt{a}_1, \ldots, \mathtt{a}_{n_{\text{var}}}$ and $m_{\text{cla}}$ clauses $\mathtt{C}_1, \ldots, \mathtt{C}_{m_{\text{cla}}}$, we construct an \prob\ instance as follows. 1) Each clause $\mathtt{C}_j = (\mathtt{l}_{j_1} \lor \mathtt{l}_{j_2} \lor \mathtt{l}_{j_3})$ corresponds to i) a vertex $u_j \in V_1$, and ii) six vertices $v_{j_{1}}, v'_{j_{1}}, v_{j_{2}}, v'_{j_{2}}, v_{j_{3}}, v'_{j_{3}} \in V_2$, where $v_{j_{t}}$ corresponds to $\mathtt{a}_{i}$, and $v'_{j_{t}}$ corresponds to $\overline{\mathtt{a}_{i}}$ (the negation of $\mathtt{a}_{i}$), if and only if $\mathtt{l}_{j_t}$ is $\mathtt{a}_{i}$ or $\overline{\mathtt{a}_{i}}$. 2) Each Boolean variable $\mathtt{a}_i$ corresponds to a vertex $w_i \in V_3$. (Thus, $V = (V_1 \cup V_2 \cup V_3)$ has $(n_{\text{var}}+7 \cdot m_{\text{cla}})$ vertices.) 3) For each clause $\mathtt{C}_j$, we add an edge between $u_j$ and each of the three vertices in $V_2$ that correspond to $\mathsf{TRUE}$ assignments of the literals in $\mathtt{C}_j$. 4) For each literal $\mathtt{l}_{j_t}$, if $\mathtt{l}_{j_t}$ is $\mathtt{a}_{i}$ or $\overline{\mathtt{a}_{i}}$, we add an edge between $w_i$ and $v_{j_{t}}$ and also an edge between $w_i$ and $v'_{j_{t}}$. (Therefore, $G$ has $(9 \cdot m_{\text{cla}})$ edges.) 5) For each edge between $V_1$ and $V_2$ with the form $(u_j, v_{j_{t}})$, we add an item $c_{j_t}$; similarly, for each edge between $V_1$ and $V_2$ with the form $(u_j, v'_{j_{t}})$, we add an item $c'_{j_t}$. 6) For each Boolean variable $\mathtt{a}_i$, we add two items $c_{i}$ and $c'_{i}$. 7) For the preference utility values, let $\text{p}(v,c) = 0$ for all $v \in V, c \in \mathcal{C}$. 8) For the social utility values, let $\tau(u_j,v_{j_{t}},c_{j_t}) = \tau(v_{j_{t}},u_j,c_{j_t}) = 1$ for every $(u_j, v_{j_{t}}) \in E$; similarly, let $\tau(u_j,v'_{j_{t}},c'_{j_t}) = \tau(v'_{j_{t}},u_j,c'_{j_t}) = 1$ for every $(u_j, v'_{j_{t}}) \in E$. Moreover, let $\tau(w_i, v_{j_{t}}, c_{i}) = \tau(v_{j_{t}}, w_i, c_{i}) = 1$ if and only if $\mathtt{l}_{j_t} = \mathtt{a}_{i}$; similarly, let $\tau(w_i, v'_{j_{t}}, c'_{i}) = \tau(v'_{j_{t}}, w_i, c'_{i}) = 1$ if and only if $\mathtt{l}_{j_t} = \overline{\mathtt{a}_{i}}$. All other $\tau$ values are zero. An example is illustrated in Figure \ref{fig:reduction}. Finally, let $k=\lambda=1$.

We first prove the sufficient condition. Assume that there exists a truth assignment in $\phi$ satisfying at least $\mathtt{m}$ clauses. We construct a feasible solution for \prob\ as follows. 1) For each satisfied clause $\mathtt{C}_j$, let ${\underline{t}_j}$ be the smallest $t$ such that $\mathtt{l}_{j_t}$ is assigned $\mathsf{TRUE}$. $u_j$ is displayed item $c_{j_{\underline{t}_j}}$ if $\mathtt{l}_{j_{\underline{t}_j}}$ is of the form $\mathtt{a}_{i}$; otherwise $u_j$ is displayed item $c'_{j_{\underline{t}_j}}$; 2) For each literal $\mathtt{l}_{j_t}$ with the form $\mathtt{a}_{i}$ and assigned $\mathsf{TRUE}$, $v_{j_{t}}$ is displayed item $c_{j_t}$; for each literal $\mathtt{l}_{j_t}$ with the form $\overline{\mathtt{a}_{i}}$ and assigned $\mathsf{TRUE}$, $v'_{j_{t}}$ is displayed item $c'_{j_t}$; 3) For each literal $\mathtt{l}_{j_t}$ with the form $\mathtt{a}_{i}$ assigned $\mathsf{FALSE}$, $v_{j_{t}}$ is displayed item $c_{i}$; for each literal $\mathtt{l}_{j_t}$ with the form $\overline{\mathtt{a}_{i}}$ assigned $\mathsf{FALSE}$, $v'_{j_{t}}$ is displayed item $c'_{i}$; 4) For each Boolean variable $\mathtt{a}_i$, $w_i$ is displayed item $c_i$ if and only if $\mathtt{a}_i$ is $\mathsf{FALSE}$ in the assignment; otherwise $w_i$ is displayed item $c'_i$; 5) For an unsatisfied clause $\mathtt{C}_j$, $u_j$ can be displayed any item.

In the following, we prove that this solution achieves an objective of at least $(2 \cdot \mathtt{m}+6 \cdot m_{\text{cla}})$ in \prob. We first observe that for each satisfied clause $\mathtt{C}_j$, either $c_{j_{\underline{t}_j}}$ is co-displayed to $u_j$ and $v_{j_{\underline{t}_j}}$, or $c'_{j_{\underline{t}_j}}$ is co-displayed to $u_j$ and $v'_{j_{\underline{t}_j}}$; each of the above cases achieves an objective value of 2. 
As there are at least $\mathtt{m}$ satisfied clauses, the above two cases in total contribute at least $2 \cdot \mathtt{m}$ to the objective of \prob. Next, for each Boolean variable $\mathtt{a}_i$ assigned $\mathsf{TRUE}$, $w_i$ is co-displayed $c'_i$ with every $v'_{j_{t}}$ such that $\mathtt{l}_{j_t}$ is $\mathtt{a}_{i}$ or $\overline{\mathtt{a}_{i}}$; similarly, for each Boolean variable $\mathtt{a}_i$ assigned $\mathsf{FALSE}$, $w_i$ is co-displayed $c_i$ with every $v_{j_{t}}$ such that $\mathtt{l}_{j_t}$ is $\mathtt{a}_{i}$ or $\overline{\mathtt{a}_{i}}$. 
The above two cases achieve an objective value of 2 for each co-display pairs. As there are exactly $(3 \cdot m_{\text{cla}})$ such pairs of $(w_i, v_{j_{t}})$ or $(w_i, v'_{j_{t}})$, they in total contributes $(6 \cdot m_{\text{cla}})$ to the objective of \prob. Combining the above two parts, the objective value is at least $(2 \cdot \mathtt{m} + 6 \cdot m_{\text{cla}})$.

We then prove the necessary condition. Assume that all possible truth assignments in $\phi$ satisfy fewer than $\mathtt{m}'$ clauses. We then prove that the optimal solution in the \prob\ instance is strictly less than $(2 \cdot \mathtt{m}'+6 \cdot m_{\text{cla}})$. We prove it by contradiction. If there exists a feasible solution of \prob\ with an objective at least $(2 \cdot \mathtt{m}'+6 \cdot m_{\text{cla}})$, we prove that there exists at least one feasible solution of \prob\ such that \textit{the total social utility from co-displaying items to one vertex in $V_2$ and another vertex in $V_3$}, denoted by $\tau(V_2, V_3)$, is exactly $6 \cdot m_{\text{cla}}$. To prove the above argument, we first observe that the total number of edges between vertices in $V_2$ and vertices in $V_3$ is $6 \cdot m_{\text{cla}}$. Furthermore, these edges form exactly $3 \cdot m_{\text{cla}}$ copies of $\text{P}_3$ (where every $\text{P}_3$ is formed by a $w_i$ and two vertices $v_{j_{t}}$ and $v'_{j_{t}}$ such that $\mathtt{l}_{j_t}$ is $\mathtt{a}_{i}$ or is $\overline{\mathtt{a}_{i}}$). Since the two edges in each $\text{P}_3$ can only contribute to the final objective through different co-displayed items, these edges in total contribute at most $6 \cdot m_{\text{cla}}$ to the final objective. Therefore, $\tau(V_2, V_3) \leq 6 \cdot m_{\text{cla}}$. Next, suppose a solution with $\tau(V_2, V_3) < 6 \cdot m_{\text{cla}}$ is given. For an arbitrary literal $\mathtt{l}_{j_t}$, let $\mathtt{a}_i$ be the corresponding variable. We examine the following cases: 1) $w_i$ is co-displayed $c_i$ with $v_{j_{t}}$ or $c'_i$ with $v'_{j_{t}}$. In this case, a social utility of 2 is achieved. 2) $w_i$ is not co-displayed any item with $v_{j_{t}}$ or $v'_{j_{t}}$. As $w_i$ does not contribute to the objective if it is displayed neither $c_i$ nor $c'_i$, it is feasible to display $c_i$ to $w_i$ without decreasing the total objective. Next, it is also feasible to co-display $c_i$ to \textit{both} $v_{j_{t}}$ and $v'_{j_{t}}$ while not decreasing the total objective, since the most social utility that $v_{j_{t}}$ or $v'_{j_{t}}$ could contribute previously from being co-displayed $c_{j_{t}}$ (or $c'_{j_{t}}$) is also at most 2. By repeatedly applying the two cases, we conclude that there exists a solution where each literal $\mathtt{l}_{j_t}$ contributes exactly 2 to $\tau(V_2, V_3)$, which implies $\tau(V_2, V_3) = 6 \cdot m_{\text{cla}}$.

Given a feasible solution with $\tau(V_2, V_3) = 6 \cdot m_{\text{cla}}$, we construct a truth assignment in $\phi$ according to this solution as follows. Let every variable $\mathtt{a}_{i}$ be $\mathsf{TRUE}$ if and only if $w_{i}$ is displayed $c'_i$; otherwise $\mathtt{a}_{i}$ is $\mathsf{FALSE}$. Note that this is consistent to the construction in the sufficient condition. For each pair of vertices $(u_j, v_{j_{t}})$ co-displayed $c_{j_{t}}$ in the solution, since $v_{j_{t}}$ is not co-displayed an item with any vertex in $V_3$, it follows that $v'_{j_{t}}$ must be co-displayed $c'_i$ with some $w_i \in V_3$ (otherwise $\tau(V_2, V_3) = 6 \cdot m_{\text{cla}}$ does not hold), which in turn implies $a_i$ (appearing in clause $\mathtt{C}_j$) is assigned $\mathsf{TRUE}$ in this truth assignment, and $\mathtt{C}_j$ is therefore satisfied. Analogously, for $(u_j, v'_{j_{t}})$ co-displayed $c'_{j_{t}}$ in the solution, $\mathtt{C}_j$ is also satisfied by this truth assignment. As the total objective is at least $(2 \cdot \mathtt{m}'+6 \cdot m_{\text{cla}})$ and $\tau(V_2, V_3) = 6 \cdot m_{\text{cla}}$, \textit{the total social utility from co-displaying items to one vertex in $V_1$ and another vertex in $V_2$}, denoted by $\tau(V_1, V_2)$, is at least $(2 \cdot \mathtt{m}')$, meaning that at least $\mathtt{m}'$ clauses are satisfied in $\phi$ by the constructed truth assignment, leading to a contradiction. Therefore, we conclude that the optimal solution in the \prob\ instance is strictly less than $(2 \cdot \mathtt{m}'+6 \cdot m_{\text{cla}})$.

Finally, for any instance of $\mathsf{MAX-E3SAT}$, a folklore random assignment of variables satisfies $\frac{7}{8} \cdot m_{\text{cla}}$ clauses in expectation (see, for example, Theorem 2.15 in \cite{MAX3SAT}). Therefore, for any instance of $\mathsf{MAX-E3SAT}$ with $m_{\text{cla}}$ clauses, the optimal objective is always at least $\frac{7}{8} \cdot m_{\text{cla}}$. Therefore, let $\mathsf{OPT}_{\mathsf{E3SAT}}$ denote the optimal objective in $\mathsf{MAX-E3SAT}$, we have
$$ \frac{7}{8} \cdot m_{\text{cla}} \leq \mathsf{OPT}_{\mathsf{E3SAT}} \leq m_{\text{cla}}.$$
Let $\chi$ be a constant. Based on this lower bound, we have:

\begin{align*}
    \mathsf{OPT}_{\mathsf{E3SAT}} \geq \chi \implies \mathsf{OPT}_{\mathsf{SVGIC}} &\geq 2\chi + 6m_{\text{cla}}\\
    &\geq 8\chi\\
    \mathsf{OPT}_{\mathsf{E3SAT}} < \frac{7}{8} \cdot \chi \implies \mathsf{OPT}_{\mathsf{SVGIC}} &< 2 \cdot \frac{7}{8} \cdot \chi + 6m_{\text{cla}}\\
    &< \frac{31}{4} \cdot \chi,
\end{align*}

\noindent where the second inequality holds since $ \frac{7}{8} \cdot m_{\text{cla}} \leq \mathsf{OPT}_{\mathsf{E3SAT}} < \frac{7}{8} \cdot \chi$ implies $ m_{\text{cla}} < \chi$. The lemma follows.
\end{proof}

Therefore, combining Lemma \ref{lemma:hardnesse3sat} and \ref{lemma:gapreduction}, if there exists an algorithm $\mathcal{ALG}$ that approximates \prob\ within a ratio of $\frac{8}{\frac{31}{4}} = \frac{32}{31}$ (or equivalently, solves the $\frac{32}{31}$-gap version of \prob), then it also solves the $\frac{8}{7}$-gap version of $\mathsf{MAX-E3SAT}$, which is known to be $\mathsf{NP}$-hard. This completes the proof of Theorem \ref{thm:new_hardness}.

We then prove the $\mathsf{APX}$-hardness through an exact reduction from the maximum $\mathcal{K}_3$-packing problem (Max-K3P) \cite{packing} as follows.

\begin{proof}
Given a graph $\hat{G} = (\hat{V}, \hat{E})$, Max-K3P aims to find a subgraph $\hat{H} \subseteq \hat{G}$ with the largest number of edges such that $\hat{H}$ is a union of vertex-disjoint edges and triangles. Given a Max-K3P instance $\hat{G} = (\hat{V}, \hat{E})$, we construct an \prob\ instance as follows. Let $G = \hat{G}$. For each edge $e = (u,v) \in \hat{E}$, we construct an item $c_{e}$. Let social utility $\tau(u,v,c_e) = \tau(v,u,c_e) = 0.5$. For each triangle $\Delta$ consisting of vertices $u,v$, and $w$ in $\hat{G}$, we construct an item $c_\Delta$. Let social utility $\tau(u,v,c_\Delta) = \tau(v,u,c_\Delta) = \tau(u,w,c_\Delta) = \tau(w,u,c_\Delta) = \tau(v,w,c_\Delta) = \tau(w,v,c_\Delta) = 0.5$. Let all other social utility values be 0. Also, let preference utility $p(u,c) = 0$ for all users $u$ and items $c$. Finally, let $k= \lambda = 1$.

We first prove the sufficient condition. Assume there exists a feasible $\hat{H} \subseteq \hat{G}$ with $x$ edges. For each disjoint edge $e = (u,v) \in \hat{H}$, let $\mathbf{A}(u,1) = \mathbf{A}(v,1) = c_{e}$ in \prob. Similarly, for each disjoint triangle $\Delta$ consisting of vertices $u,v$, and $w$ in $\hat{H}$, let $\mathbf{A}(u,1) = \mathbf{A}(v,1) = \mathbf{A}(w,1) = c_{\Delta}$. Each edge in $\hat{H}$ achieves an \scenario\ utility of 1. Therefore, this configuration achieves a total \scenario\ utility of $x$. We then prove the necessary condition. If the optimal objective in the \prob\ instance is $x$, we let $\hat{H}$ include all edges $e = (u,v)$ such that $\mathbf{A}(u,1) = \mathbf{A}(v,1) = c$ and $\tau(u,v,c) = 0.5$ (Note that $c$ could be $c_e$ or some $c_{\Delta}$). It is straightforward to verify that any connected component in $\hat{H}$ is either an edge or a triangle. As each edge in $\hat{H}$ contributes 1 to the objective in \prob, there are exactly $x$ edges in $\hat{H}$. Finally, since Max-K3P is $\mathsf{APX}$-hard \cite{packing}, \prob\ is also $\mathsf{APX}$-hard. 
\end{proof}
}

\opt{full}{
We then prove the hardness of approximation for the advanced \probtwo\ problem.

\begin{theorem} \label{thm:veryhard}
There exists no polynomial-time algorithm that approximates \probtwo\ within any constant factor, unless the exponential time hypothesis (ETH) fails.
\end{theorem}

\begin{proof}
We prove the theorem with a gap-preserving reduction from the Densest k-Subgraph problem (DkS) \cite{eth}. Given a graph $\hat{G} = (\hat{V}, \hat{E})$, DkS aims to find a subgraph $\hat{H} \subseteq \hat{G}$ with $\hat{k}$ vertices with the largest number of induced edges. (Note that for a fixed $\hat{k}$, maximizing the density is equivalent to maximizing the number of edges. We use $\hat{k}$ instead of the conventional $k$ to avoid confusion with the number of slots in \probtwo.) Given a DkS instance $\hat{G} = (\hat{V}, \hat{E})$ and $\hat{k}$, we construct an \prob\ instance as follows. Let $V = \hat{V} \cup S$, where $S = \{w_1, w_2, \ldots, w_{|S|} \}$ consists of $|S| = \hat{k} - (|\hat{V}|$ mod $\hat{k})$ additional vertices if $\hat{k} \nmid |\hat{V}|$, and $S = \emptyset$ if $\hat{k} \mid |\hat{V}|$. Let $E = \hat{E}$. Therefore, all vertices in $S$ are singletons. Let $\mathcal{C} = \{c_1, c_2, \ldots, c_m\}$, where $m = \frac{|V|}{\hat{k}}$. Let preference utility $p(u,c) = 0$ for all $u \in V$ and $c \in \mathcal{C}$. For each edge $e = (u,v) \in \hat{E}$, let $\tau(u,v,c_1) = \tau(v,u,c_1) = 0.5$. Let $\tau(u,v,c) = 0$ for all $(u,v) \notin \hat{E}$ or $c \neq c_1$. Finally, let $k=1$, $\lambda = 1$, and $M = \hat{k}$.

We first prove the sufficient condition. Assume that there exists a subgraph $\hat{H} \subseteq \hat{G}$ with $\hat{k}$ vertices and $x$ edges. We construct an \probtwo\ solution with objective $x$ by co-displaying $c_1$ to all vertices in $\hat{H}$ (and thereby obtaining social utility of 1 on each edge) and then randomly partitioning the vertices not in $\hat{H}$ to $m-1$ sets of cardinality $\hat{k}$. We then prove the necessary condition. Since there are $m = \frac{|V|}{\hat{k}}$ items, and the subgroup size constraint is $M = \hat{k}$, it follows that each item is co-displayed to exactly $M = \hat{k}$ users in a feasible solution. Assume the optimal objective in the \probtwo\ instance is $x$, then the $M = \hat{k}$ users co-displayed item $c_1$ collectively form a induced subgraph of exactly $x$ edges. Let $\hat{H}$ be the induced subgraph formed by these vertices in $\hat{G}$, then $\hat{H}$ has exactly $x$ edges. Finally, according to \cite{eth}, assuming the exponential time hypothesis (ETH) holds, there is no polynomial-time algorithm that approximates DkS to within a constant factor; in fact, a stronger inapproximability of $n^{(1/ \log \log n)^c}$ holds. The result for \probtwo\ directly follows.
\end{proof}

Here we further point out that both the hardness results of \prob\ and \probtwo\ are from a simple case where $k=1$, i.e., the considered problems are already very hard when only one display slot is in consideration. Therefore, the hardness of \prob\ and \probtwo\ with general values of $k$ may be even harder. We revisit this issue in Section \ref{subsec:combinatorial}.
}

Next, we propose an Integer Programming (IP) model for \prob\ \opt{full}{and \probtwo\ }to serve as the cornerstone for the approximation algorithm proposed later in Section \ref{sec:algo}.\opt{full}{ We begin with the IP for \prob.} Let binary variable $x^c_{u,s}$ denote whether user $u$ is displayed item $c$ at slot $s$, i.e., $x^c_{u,s} = 1$ if and only if $\mathbf{A}(u,s) = c$. Let $x^c_u$ indicate whether $u$ is displayed $c$ at any slot in the \configuration. Moreover, for each pair of friends $e = (u,v) \in E$, let binary variable $y^c_{e,s}$ denote whether $u$ and $v$ are co-displayed item $c$ at slot $s$, i.e., $y^c_{e,s} = 1$ if and only if $u \xleftrightarrow[s]{c} v$. Similarly, variable $y^c_{e} = 1$ if and only if $u \xleftrightarrow{c} v$. The objective of \prob\ is specified as follows.

$$ \text{max} \sum\limits_{u \in V} \sum\limits_{c \in \mathcal{C}} [ (1-\lambda) \cdot \text{p}(u, c) \cdot x^c_u + \lambda \cdot \sum\limits_{e = (u,v) \in E} ( \tau(u,v,c) \cdot y^c_e) ] $$

\noindent subject to the following constraints,
\small
\begin{align}
    \sum_{s=1}^{k} x^c_{u,s} \leq 1, \quad & \forall u \in V, c \in \mathcal{C} \label{ilp:no_replicated_item}\\
    \sum_{c \in \mathcal{C}} x^c_{u,s} = 1, \quad & \forall u \in V, s \in [k] \label{ilp:one_slot_one_item}\\
    x^c_u = \sum_{s=1}^{k} x^c_{u,s}, \quad & \forall u \in V, c \in \mathcal{C} \label{ilp:user_recommended}\\
    y^c_{e} = \sum_{s=1}^{k} y^c_{e,s}, \quad & \forall e = (u,v) \in E, c \in \mathcal{C} \label{ilp:edge_recommended}\\
    y^c_{e,s} \leq x^c_{u,s}, \quad & \forall e = (u,v) \in E, s \in [k], c \in \mathcal{C} \label{ilp:direct_relation_1} \\
    y^c_{e,s} \leq x^c_{v,s}, \quad & \forall e = (u,v) \in E, s \in [k], c \in \mathcal{C} \label{ilp:direct_relation_2} \\
    x^c_{u,s}, x^c_u, y^c_{e,s}, y^c_e \in \{ 0,1 \}, \quad & \forall u \in V, e \in E, s \in [k], c \in \mathcal{C} \label{ilp:integrality}.
\end{align}

\normalsize
Constraint (\ref{ilp:no_replicated_item}) states that each item $c$ can only be displayed at most once to a user $u$ (i.e., the no-duplication constraint). Constraint (\ref{ilp:one_slot_one_item}) guarantees that each user $u$ is displayed exactly one item at each slot $s$. Constraint (\ref{ilp:user_recommended}) ensures that $x^c_u = 1$ ($u$ is displayed $c$ in the configuration) if and only if there is exactly one slot $s$ with $x^c_{u,s} = 1$. Similarly, constraint (\ref{ilp:edge_recommended}) ensures $y^c_{e} = 1$ if and only if $y^c_{e,s} =1 $ for exactly one $s$. Constraints (\ref{ilp:direct_relation_1}) and (\ref{ilp:direct_relation_2}) specify the co-display, where $y^c_{e,s}$ is allowed to be $1$ only if $c$ is displayed to both $u$ and $v$ at slot $s$, i.e., $x^c_{u,s} = x^c_{v,s} = 1$. Finally, constraint (\ref{ilp:integrality}) ensures all decision variables are binary. \revise{Note that the $x$ variables in the IP model are sufficient to represent the solution of \prob\ (i.e., $x^c_{u,s}$ denotes whether an item $c$ is displayed at slot $s$ for user $u$), whereas the $y$ variables are auxiliary to enable \prob\ to be formulated as an IP. Without incorporating the $y$ variables, the formulation will become nonlinear.}



\opt{full}{
Next, for each $e = (u,v) \in E$, let binary variable $z^c_e$ denote whether $u$ and $v$ are co-displayed (both directed \textit{or} indirected) item $c$. Note that $y^c_{e} = 1$ implies $z^c_{e} = 1$. To avoid repetitive calculation of the social utility, the coefficient before $y^c_{e}$ in the objective is modified to be $(1-\text{d}_{\text{tel}})$, such that a social utility of $\tau(u,v,c)$ is obtained when $y^c_{e} = z^c_{e} = 1$. The objective of \probtwo\ is thus

\begin{align*} \text{max} \sum\limits_{u \in V} \sum\limits_{c \in \mathcal{C}} &[ (1-\lambda) \cdot \text{p}(u, c) \cdot x^c_u + \lambda \cdot \sum\limits_{e = (u,v) \in E} \tau(u,v,c) \\ &\cdot \big( (1-\text{d}_{\text{tel}}) \cdot y^c_e + \text{d}_{\text{tel}} \cdot z^c_e \big) ] 
\end{align*}

\noindent subject to constraints (\ref{ilp:no_replicated_item}) to (\ref{ilp:direct_relation_2}) above, as well as the following constraints:
\small
\begin{align}
    z^c_{e} \leq x^c_u, \quad & \forall e = (u,v) \in E, c \in \mathcal{C} \label{ilp2:indirect_relation_1} \\
    z^c_{e} \leq x^c_v, \quad & \forall e = (u,v) \in E, c \in \mathcal{C} \label{ilp2:indirect_relation_2} \\
    x^c_{u,s}, x^c_u, y^c_{e,s}, y^c_e, z^c_e \in \{ 0,1 \}, \quad & \forall u \in V, e \in E, s \in [k], c \in \mathcal{C} \label{ilp2:integrality}.
\end{align}

\normalsize
Constraints \ref{ilp2:indirect_relation_1} and \ref{ilp2:indirect_relation_2} specify the co-display, where $z^c_{e}$ is allowed to be $1$ only if $c$ is displayed to $u$ and $v$ in some (not necessarily the same) slots, i.e., $x^c_u = x^c_v = 1$. Therefore, $y^c_{e} = z^c_{e} = 1$ for the case of direct co-display. Constraint \ref{ilp2:integrality} is the integrality constraint. The other constraints follow from the basic \prob\ problem. Note that in the optimal solution, whenever $y^c_{e} = 1$ for some edge $e$, the corresponding $z^c_{e}$ will also naturally equal $1$, as it is always better to have a larger $z^c_{e}$. Therefore, no additional constraint needs to be introduced to enforce this relation. Note that, similar to that in the IP for \prob, the $y$ and $z$ variables are auxiliary variables to enable linearity in this IP.
}

\opt{full}{
\subsection{Related Combinatorial Problems} \label{subsec:combinatorial}

\begin{table*}[t]
\begin{small}
\caption{Related Combinatorial Problems (approximation ratio for MAX problems inversed).}
\begin{center}
\begin{tabular}{|l|l|l|l|l|l|l|l|l|}
\hline
\begin{tabular}[c]{@{}l@{}}Problem\\ Name\end{tabular} & Abbrev. & Prelabeling & Type & Objective & \begin{tabular}[c]{@{}l@{}}Hardness\\ Result\end{tabular} & \begin{tabular}[c]{@{}l@{}}Approx.\\ Result\end{tabular} & \begin{tabular}[c]{@{}l@{}}Is a special\\ case of\end{tabular} & \begin{tabular}[c]{@{}l@{}}Is the comple-\\ ment of\end{tabular} \\ \hline \hline
\begin{tabular}[c]{@{}l@{}}Submodular\\ Multi-Labeling \cite{multilabeling} \end{tabular} & Sub-ML & General & MIN & \begin{tabular}[c]{@{}l@{}}submodular\\ set function\end{tabular} & \begin{tabular}[c]{@{}l@{}}$2-\frac{2}{m}-\epsilon$\\ (UGC) \cite{multilabeling} \end{tabular}  & $2-\frac{2}{m}$ \cite{multilabeling} &  & Sup-ML \\ \hline
\begin{tabular}[c]{@{}l@{}}Supermodular\\ Multi-Labeling \cite{multilabeling} \end{tabular} & Sup-ML & General & MAX & \begin{tabular}[c]{@{}l@{}}supermodular\\ set function\end{tabular} & $\mathsf{NP}$-hard & $\frac{m}{2}$ \cite{multilabeling}  &  & Sub-ML \\ \hline
\begin{tabular}[c]{@{}l@{}}Submodular\\ Multiway Partition \cite{multiwaypartition} \end{tabular} & Sub-MP & $m$ terminals & MIN & \begin{tabular}[c]{@{}l@{}}submodular\\ set function\end{tabular} &  \begin{tabular}[c]{@{}l@{}}$2-\frac{2}{m}-\epsilon$\\ ($\mathsf{NP}=\mathsf{RP}$) \cite{mp13soda} \end{tabular} & $2-\frac{2}{m}$ \cite{mp13soda} & Sub-ML & Sup-MP \\ \hline
\begin{tabular}[c]{@{}l@{}}Supermodular\\ Multiway Partition \cite{multilabeling} \end{tabular} & Sup-MP & $m$ terminals & MAX & \begin{tabular}[c]{@{}l@{}}supermodular\\ set function\end{tabular} & $\mathsf{NP}$-hard &  & Sup-ML & Sub-MP \\ \hline
\begin{tabular}[c]{@{}l@{}}Maximum\\ Happy Vertices \cite{homophyly} \end{tabular} & MHV & General & MAX & \begin{tabular}[c]{@{}l@{}}\# happy\\ vertices\end{tabular} & $\mathsf{NP}$-hard & $\frac{m}{2}$ \cite{multilabeling} & Sup-ML & MUHV \\ \hline
\begin{tabular}[c]{@{}l@{}}Minimum\\ Unhappy Vertices \cite{multilabeling} \end{tabular} & MUHV & General & MIN & \begin{tabular}[c]{@{}l@{}}\# unhappy\\ vertices\end{tabular} & \begin{tabular}[c]{@{}l@{}}$2-\frac{2}{m}-\epsilon$\\ (UGC) \cite{multilabeling} \end{tabular}  & $2-\frac{2}{m}$ \cite{multilabeling} & Sub-ML & MHV \\ \hline
\begin{tabular}[c]{@{}l@{}}Maximum\\ Happy Edges \cite{homophyly} \end{tabular} & MHE & General & MAX & \begin{tabular}[c]{@{}l@{}}\# happy\\ edges\end{tabular} & $\mathsf{NP}$-hard & 1.1716 \cite{mhe} & Sup-ML &  \\ \hline
\multicolumn{2}{|l|}{Multiway Cut \cite{multiwaycut}} & $m$ terminals & MIN & \begin{tabular}[c]{@{}l@{}}\# edges\\ removed\end{tabular} & $\frac{8}{7}$ \cite{sdp} & 1.2965 \cite{multiwaycutapprox} & UML & \begin{tabular}[c]{@{}l@{}}Multiway\\ Uncut\end{tabular} \\ \hline
\multicolumn{2}{|l|}{Multiway Uncut \cite{multiwayuncut}} & $m$ terminals & MAX & \begin{tabular}[c]{@{}l@{}}\# edges\\ preserved\end{tabular} &$\mathsf{NP}$-hard  & 1.1716 \cite{multiwayuncut} & MHE & \begin{tabular}[c]{@{}l@{}}Multiway\\ Cut\end{tabular} \\ \hline \hline
\begin{tabular}[c]{@{}l@{}}Uniform Metric\\ Labeling \cite{JK02} \end{tabular} & UML & $m$ terminals & MIN & \begin{tabular}[c]{@{}l@{}}\# total cost of\\edges/vertices\end{tabular} & \begin{tabular}[c]{@{}l@{}}$2-\frac{2}{m}$\\ (UGC) \cite{sdp} \end{tabular} & $2$ \cite{JK02} & \begin{tabular}[c]{@{}l@{}}Metric\\ Labeling\end{tabular} & \\ \hline
\end{tabular}
\label{tab:theory}
\end{center}
\end{small}
\end{table*}

\prob\ is related to the \textit{Multi-Labeling} (ML) \cite{multilabeling} problem and its variations, including {Multiway Partition} \cite{multilabeling, multiwaypartition}, \textit{Maximum Happy Vertices/Edges} (MHV/MHE) \cite{homophyly}, and \textit{Multiway Cut} \cite{multiwaycut} in graphs. In the Multi-Labeling (ML) problem, the general inputs are a ground set of entities $V$, a set of labels $L = \{1, 2, \ldots, m\}$, and a \textit{partial} labeling function $\ell : V \mapsto L$ that pre-assigns each label $i$ to a non-empty subset $T_i \subseteq V$, and a set function $f: 2^V \mapsto \mathcal{R}$. The goal of ML is to partition the ground set into $m$ subsets to optimize the aggregated set function on the partitioned subsets, i.e., $\sum\limits_{i=1}^{m} f(T_i)$. The special cases with set function $f$ being submodular or supermodular is called \textit{Submodular Multi-Labeling} (Sub-ML) or \textit{Supermodular Multi-Labeling} (Sup-ML) \cite{multilabeling}. Another work \cite{multiwaypartition} studies the case where the partial labeling function is required to assign a label $i$ to only one entity $v_i$. The above problems become \textit{Submodular Multiway Partition} (Sub-MP) and \textit{Supermodular Multiway Partition} (Sup-MP) in this context. Some special cases of Sub-ML and Sup-ML are described as graph problems, where the ground set $V$ is the vertex set on a graph $G = (V,E)$. In this context, the above assumption on $f$ in Sub-MP and Sup-MP can be interpreted as $m$ specified terminals on the graph. Two representative special cases are the \textit{Maximum Happy Vertices} (MHV) problem and the \textit{Maximum Happy Edges} (MHE) problem  \cite{homophyly}. Originally proposed to study homophyly effects in networks, MHV and MHE are described as coloring problems where, different from traditional graph coloring, vertices are encouraged to share the same color with neighbors. \textit{Edges with same-color vertices} and \textit{vertices sharing the same color with all neighbors} are characterized as \textit{happy edges} and \textit{happy vertices}, respectively. The optimization goals of MHV and MHE are thus to maximize the number of happy vertices and edges. The complement problem of MHV, \textit{Minimum Unhappy Vertices} (MUHV) \cite{multilabeling}, minimizes the number of unhappy vertices. When the above optimization goal of MHE is combined with the assumption of $m$ terminals, the resultant problem becomes the \textit{Multiway Uncut} problem \cite{multiwayuncut}, which is the complement of the classical and extensively studied \textit{Multiway Cut} problem \cite{multiwaycut}. The Multiway Cut problem asks for the minimum number of edges removed to partition a graph into $m$ subgraph, each containing one terminal. The Metric Labeling problem \cite{JK02} is a generalization of Multiway Cut considering edge weights and also a cost function on labeling the vertices. We summarize the problem characteristics, hardness results, currently best algorithmic results, and their interdependency in Table \ref{tab:theory}. All the problems are $\mathsf{NP}$-hard, while more advanced inapproximability results, relying on other complexity conjectures, are stated for some problems. Note that the approximation ratios for maximization problems are shown as inverses (i.e., values larger than 1) for better comparison and consistency with this paper.

Among the ML-type problems, \prob\ is particularly correlated to the MHE problem. Regarding the displayed items in \prob\ as colors(labels) in MHE, the social utility achieved in \prob\ is closely related to preserved number of happy edges in MHE, and \prob\ thereby encourages partitioning all users into dense subgroups to preserve the most social utility. However, \prob\ is more difficult than the above ML-type problems due to the following reasons. 1) The ML-type problems find a strict partition that maps each entity/vertex to only one color(label), while \prob\ assigns $k$ items to each user, implying that any direct reduction from the above problems can only map to the $k=1$ special case in \prob, where \algo\ can achieve a 2-approximation. 
2) The above problems do not discriminate between different labels, i.e., switching the assigned labels of two different subgroups will not change their objective functions. This corresponds to the special case of \prob\ where all preference utility $\text{p}(u, c)$ and social utility $\tau(u, v, c)$ do not depend on the item $c$.\opt{fullrevise}{ In contrast, \prob\ captures the item preferences as both utility measures in \prob\ depend on the item identity.}
3) The above problems admit a partial labeling (pre-labeling) in the input such that some entities have predefined fixed labels (otherwise labeling every entity with the same label is optimal and renders the problem trivial), while \prob\ does not specify any items to be displayed to specific users. However, \prob\ requires $k$ items for each user; moreover, even in the $k=1$ special case, it is still not optimal to simply display the same item to all users in \prob\ due to the item-dependent preference and social utility.

Finally, we note that while most of the minimization problems admit stronger hardness results than conventional $\mathsf{NP}$-hardness, no such type of results are known for all maximization problems above. As \prob\ is also a maximization problem, its hardness does not directly follow from any complementary minimization problems. Interestingly, the algorithmic aspect of all these problems share an obvious trend: almost all approximation results here, including our algorithm for \prob, utilizes some kind of (dependent) randomized rounding on some formulation of relaxed problem (not limited to linear relaxation). In this regard, this paper extends the scheme beyond labeling problems to the \prob\ problem that assigns a $k$-itemset to each user, which has an even more complicated combinatorial nature.
}

\section{Algorithm Design}\label{sec:algo}

\opt{full}{
\begin{table}[t!]
\caption{Notations used in \algo.}
\begin{center}
\begin{tabular}{|c|l|}\hline
Symbol & Description \\\hline\hline
$X^\ast$ & optimal \fragsol\ \\\hline
${x^\ast}^c_{u,s}$ & optimal decision variable \\\hline
$(c,s,\alpha)$ & set of focal parameters \\\hline
$c/\hat{c}$ & \roundingitem\ \\\hline
$s/\hat{s}$ & \roundingslot\ \\\hline
$u/\hat{u}$ & current user/eligible user\\\hline
$\alpha$ & \roundingalpha\ \\\hline
$U$ & target subgroup \\\hline
$\text{OPT}$ & optimal \scenario\ utility \\\hline
$\mathcal{R}$ & achieved utility \\\hline
$\mathcal{R}_{\text{per}}$ & achieved preference utility \\\hline
$\mathcal{R}_{\text{soc}}$ & achieved social utility \\\hline
$S_{\text{cur}}$ & available display units \\\hline
$S_{\text{tar}}(c,s,{x^\ast}^c_{u,s})$ & display units to be filled \\\hline
$S_{\text{fut}}(c,s,{x^\ast}^c_{u,s})$ & display units not to be filled \\\hline
${\text{ALG}}(S_{\text{tar}}(c,s,{x^\ast}^c_{u,s}))$ & \scenario\ utility gained by \stagethree\ \\\hline
${\text{OPT}}_{\text{LP}}(S_{\text{fut}}(c,s,{x^\ast}^c_{u,s}))$ & expected future \scenario\ utility \\\hline
$\text{w}^c_e$ & $\tau(\hat{u},\hat{v},c)+\tau(\hat{v},\hat{u},c)$\\\hline
$r$ & balancing ratio \\\hline
\end{tabular}
\label{tab:symbols_algo}
\end{center}
\end{table}
}

In this section, we introduce the \algofull\ (\algo) algorithm to tackle \prob. As shown in Example \ref{example:illustrative}, the \personalized\ and \group\ approaches do not solve \prob\ effectively, as the former misses on social utility from co-display while the latter fails to leverage the flexibility of \notion\ to preserve personal preference. An alternative idea, called the \textit{subgroup} approach, is to first pre-partition the shopping group (i.e., the whole user set) into some smaller social-aware subgroups (e.g., using traditional community detection techniques), and then determine the displayed items based on preferences of the subgroups. While this idea is effective for social event organization \cite{XW18WWW} where each user is assigned to exactly one social activity, it renders the partitioning of subgroups static across all display slots in \prob, i.e., a user is always co-displayed common items only with other users in the same subgroup. Therefore, this approach does not fully exploit the \notion\ flexibility, leaving some room for better results.

Instead of using a universal partition of subgroups as in the aforementioned subgroup approach, we aim to devise a more sophisticated approach that allows varied co-display subgroups across the display slots to maximize the user experience. Accordingly, we leverage Linear Programming (LP) relaxation strategies that build on the solution of the Integer Program formulated in Section \ref{subsec:hardness_ip} because it naturally facilitates different subgroup partitions across all slots while allocating proper items for those subgroups with CID. In other words, our framework partitions the subgroups (for each slot) and selects the items simultaneously, thus avoiding any possible shortcomings of two-phased approaches that finish these two tasks sequentially. By relaxing all the integrality constraints in the IP, we obtain a relaxed linear program whose \textit{fractional} optimal solution can be explicitly solved in polynomial time. For an item $c$ to be displayed to a user $u$ at a certain slot $s$, the fractional decision variable ${x^\ast}^c_{u,s}$ obtained from the optimal solution of the LP relaxation problem can be assigned as its \emph{\weight}. Items with larger \weight s\ are thus inclined to contribute more \scenario\ utility (i.e., the objective value), since they are preferred by the users or more capable of triggering social interactions.

Next, it is important to design an effective rounding procedure to construct a promising \configuration\ according to the \weight s. We observe that simple \textit{independent} rounding schemes may perform egregiously in \prob\ because they do not facilitate effective co-displaying of common items, thereby losing a huge amount of potential social utility upon constructing the \configuration, especially in the cases where the item preferences are not diverse. Indeed, we prove that independent rounding schemes may achieve an expected total objective of only $O(1/m)$ of the optimal amount. Motivated by the incompetence of independent rounding, our idea is to leverage \textit{dependent rounding} schemes that encourage co-display of items of common interests, i.e., with high \weight s to multiple users in the optimal LP solution. 

Based on the idea of dependent rounding schemes in \cite{JK02}, we introduce the idea of \textit{\stagethreefull} (\stagethree) that co-displays a \textit{\roundingitem} $c$ at a specific \textit{\roundingslot} $s$ to every user $u$ with a \weight\ $x^c_{u,s}$ greater than a \textit{\roundingalpha} $\alpha$. In other words, \stagethree\ clusters the users with high  \weight s to a \roundingitem\ $c$ to form a \textit{target subgroup} in order to co-display $c$ to the subgroup at a specific display slot $s$. Depending on the randomly chosen set of \textit{focal parameters}, including the \roundingitem, the \roundingslot, and the \roundingalpha, the size of the created target subgroups can span a wide spectrum, i.e., as small as a single user and as large as the whole user set $V$, to effectively avoid the pitfalls of \personalized\ and \group\ approaches. The randomness naturally makes the algorithm less vulnerable to extreme-case inputs, therefore resulting in a good approximation guarantee. Moreover, \stagethree\ allows the partitions of subgroups to vary across all slots in the returned \configuration, exploiting the flexibility provided by \notion. However, different from the dependent rounding schemes in \cite{JK02}, the construction of \configuration s in \prob\ faces an additional challenge -- it is necessary to carefully choose the displayed items at multiple slots to ensure the no-duplication constraint. 

We prove that \algo\ is a 4-approximation algorithm for \prob\ and also show that \algo\ can be fully derandomized into a deterministic approximation algorithm. \opt{full}{We also tailor \stagethree\ to consider the size constraint so as to extend \algo\ for the more complicated \probtwo. }
In the following, we first deal with the case with $\lambda = \frac{1}{2}$. We observe that all other cases with $\lambda \neq 0$ can be reduced to this case by proper scaling of the inputs, i.e., $\text{p'}(u,c) = \frac{1-\lambda}{\lambda}\text{p}(u,c)$, whereas $\lambda = 0$ makes the problem become trivial. \opt{short}{We explicitly prove this property in Section \ref{subsec:enhancements} in the full version \cite{Online}.}\opt{full}{We explicitly prove this property in Section \ref{subsec:enhancements}.} Moreover, for brevity, the total \scenario\ utility is scaled up by 2 so that it is a direct sum of the preference and social utility. \opt{short}{A table of all notations used in \algo\ is also provided in \cite{Online}.}\opt{full}{All used notations in \algo\ are summarized in Table \ref{tab:symbols_algo}.}

\subsection{LP Relaxation and an Independent Rounding Scheme} \label{subsec:independent}

\revise{Following the standard linear relaxation technique \cite{approxalg}, the LP relaxation of \prob\ is formulated by replacing the integrality constraint (constraint (\ref{ilp:integrality})) in the IP model, i.e., $ x^c_{u,s}, x^c_u, y^c_{e,s}, y^c_e \in \{ 0,1 \}$, with linear upper/lower bound constraints, i.e., $0 \leq  x^c_{u,s}, x^c_u, y^c_{e,s}, y^c_e \leq 1$. The optimal \fragsol\ of the relaxed problem can be acquired in polynomial time with commercial solvers, e.g., CPLEX \cite{CPLEX} or Gurobi \cite{Gurobi}. Recall that the $x$-variables are sufficient to represent the solution of \prob\ (i.e., $x^c_{u,s}$ denotes whether an item $c$ is displayed at slot $s$ for user $u$), whereas the $y$-variables are auxiliary. Therefore, the optimal solution can be fully represented by $X^\ast$ (the set of optimal $x$ variables).} The fractional decision variable ${x^\ast}^c_{u,s}$ in $X^\ast$ is then taken as the \textit{\weight} of item $c$ at slot $s$ for user $u$. Note that the optimal objective in the relaxed LP is an upper bound of the optimal total \scenario\ utility in \prob, because the optimal solution in \prob\ is guaranteed to be a feasible solution of the LP relaxation problem.

\begin{example} \label{example:lp_relax}
Table \ref{table:fragsol} shows the \weight s in Example \ref{example:running}, where the \fragsol\ is identical for all slots 1-3 (thereby only slot 1 is shown). For example, the \weight\ of $c_1$ (the tripod) to Alice at each slot is ${x^\ast}^{c_1}_{u,1} = {x^\ast}^{c_1}_{u,2} = {x^\ast}^{c_1}_{u,3} = 0.33$. 
\end{example}

\small
\begin{table}[t]
	\centering
	\caption{Optimal \fragsol\ for slot 1 in Example \ref{example:running}.}
	\label{table:fragsol}
	\begin{tabular}{l|ccccc}
  &  ${x^\ast}^{c_1}_{\cdot,1}$ & ${x^\ast}^{c_2}_{\cdot,1}$ & ${x^\ast}^{c_3}_{\cdot,1}$ & ${x^\ast}^{c_4}_{\cdot,1}$ & ${x^\ast}^{c_5}_{\cdot,1}$ \\ \hline 
 $u_A$ &  $0.33$  & $0.33$  &  $0$    &  $0$    &  $0.33$ \\
 $u_B$ &  $0.33$  & $0.33$  &  $0$    &  $0.33$ &  $0$    \\
 $u_C$ &  $0$     & $0$     & $0.33$  &  $0.33$ &  $0.33$ \\
 $u_D$ &  $0.33$  & $0$     & $0$     &  $0.33$ &  $0.33$ \\
	\end{tabular}
\end{table}
\normalsize

Note that three items ($c_1, c_2$, and $c_5$) have nonzero \weight s to Alice at slot 1 in Example \ref{example:lp_relax}, which manifests that the optimal LP solution may not construct a valid \configuration\ because each user is allowed to display exactly one item at each display slot in \prob. Therefore, a rounding scheme is needed to construct appropriate \configuration s from the \weight s. Given $X^\ast$, a simple rounding scheme is to randomly (and independently) assign item $c$ to user $u$ at slot $s$ with probability ${x^\ast}^c_{u,s}$, i.e., the \weight\ of $c$ to $u$ at $s$, so that more favorable items are more inclined to be actually displayed to the users. \opt{full}{This rounding scheme is summarized in Algorithm \ref{alg:trivial_rounding}. 

\begin{algorithm}[h]  
    \caption{Trivial Rounding Scheme}
    \label{alg:trivial_rounding}
    \begin{algorithmic}[1]
        \Require $X^\ast$
        \Ensure An \configuration\ $\mathbf{A}$
        \For {$u \in V$}
            \For {$s \in \{1,2,\ldots,k\}$}
                \State Display item $c$ to user $u$ at slot $s$ independently with probability ${x^\ast}^c_{u,s}$
            \EndFor
        \EndFor
    \end{algorithmic}
\end{algorithm}
}

However, as this strategy selects the displayed items independently, for a pair of friends $u$ and $v$, the chance that the algorithm obtains high social utility by facilitating co-display is small, since it requires the randomized rounding process to hit on the same item for both $u$ and $v$ simultaneously. Furthermore, this strategy could not ensure the final \configuration\ to follow the no-duplication constraint, as an item $c$ can be displayed to a user $u$ at any slot with a nonzero \weight. The following lemma demonstrates the ineffectiveness of this rounding scheme.

\begin{lemma} \label{lemma:trivial_rounding}
There exists an \prob\ instance $I$ on which the above rounding scheme achieves only a total \scenario\ utility of $O(\frac{1}{m})$ of the optimal value in expectation.
\end{lemma}

\opt{short}{
\begin{proof}
Due to space constraints, please see Section \ref{subsec:independent} in the full version \cite{Online} for the proof.
\end{proof}
}
\opt{full}{
\begin{proof}
Assume that for all users $u,v \in V$ and $c \in \mathcal{C}$, we have $\text{p}(u,c) = 0$ and $\tau(u,v,c) = \tau$ for a constant $\tau > 0$. Intuitively, every user is indifferent among all items. In this case, a trivial optimal solution for the relaxed LP can be found by setting ${x^\ast}^c_{u,s} = \frac{1}{m}$ for all $c,u,s$. As the trivial rounding scheme determines the displayed items independently, for any pair of users $u,v$ and any slot $s$, the probability that $u$ and $v$ are co-displayed any item at slot $s$ is only $\frac{1}{m}$. Therefore, the expected total \scenario\ utility achieved is $\frac{n(n-1)}{m} \cdot \tau \cdot k$. On the other hand, co-displaying an arbitrary item to all users achieves a social utility of $n(n-1) \cdot \tau$, and repeating this with distinct items for all slots yields a total \scenario\ utility of $n(n-1) \cdot \tau \cdot k$. This independent rounding scheme thus achieves only $\frac{(n-1)\tau}{m(n-1)\tau} = O(\frac{1}{m})$ of the optimal value. Moreover, the resulting \configuration\ is highly unlikely to satisfy the no-duplication constraint.
\end{proof}
}

\subsection{Alignment-aware Algorithm} \label{subsec:details}

To address the above issues, 
we devise the \textit{\stagethreefull} (\stagethree) rounding scheme, inspired by the dependent rounding scheme for labeling problems \cite{JK02}, as the cornerstone of \algo\ to find a target subgroup $U$ according to a set of focal parameters for co-display of the \roundingitem\ to all users in $U$. Given the optimal \fragsol\ $X^\ast$ to the LP relaxation problem, \algo\ iteratively 1) samples a set of focal parameters $(c,s,\alpha)$ with $c \in \mathcal{C}, s \in \{1, 2, \ldots, k\}$, and $\alpha \in [0,1]$ uniformly at random; it then 2) conducts \stagethree\ according to the selected set of parameters $(c,s,\alpha)$ until a complete \configuration\ is constructed. It is summarized in Algorithm \ref{alg:alignment_aware_rounding}.


\begin{algorithm}[h]  
    \caption{\algofull\ (\algo)}
    \label{alg:alignment_aware_rounding}
    \begin{algorithmic}[1]
        \Require $X^\ast$
        \Ensure An \configuration\ $\mathbf{A}$
        \State $\mathbf{A}(\hat{u},\hat{s}) \gets \text{NULL}$ for all $\hat{u},\hat{s}$
        \State $X^\ast \gets X^\ast_{\text{LP}}$ 
        \While {some entry in $\mathbf{A}$ is $\text{NULL}$} 
                \State Sample $c \in \mathcal{C}$, $s \in [k]$, $\alpha \in [0,1]$ randomly
            \For{$\hat{u} \in V$}
                \If{$\mathbf{A}(\hat{u},s) = \text{NULL}$ and $\mathbf{A}(\hat{u},t) \neq c \, \forall t \neq s$} \\ \hfill \Comment{$\hat{u}$ eligible for $(c,s)$}
                    \If{${x^\ast}^c_{\hat{u},s} \geq \alpha$} 
                        \State $\mathbf{A}(\hat{u},s) \gets c$
                    \EndIf
                \EndIf
            \EndFor
        \EndWhile
        \Return $\mathbf{A}$
    \end{algorithmic}
\end{algorithm}
\normalsize


\para{\stagethreefull.} Given the randomly sampled set of parameters $(c,s,\alpha)$, \stagethree\ finds the target subgroup as follows. With the \roundingitem\ $c$ and the \roundingslot\ $s$, a user $\hat{u}$ is \textit{eligible} for $(c,s)$ in \stagethree\ if and only if 1) $\hat{u}$ has not been displayed any item at slot $s$, and 2) $c$ has not been displayed to $\hat{u}$ at any slot. Users not eligible for $(c,s)$ are not displayed any item in \stagethree\ to ensure the no-duplication constraint. For each eligible user $\hat{u}$, \stagethree\ selects $c$ for $\hat{u}$ at slot $s$ (i.e., $\textbf{A}(\hat{u},s) \leftarrow c$) if and only if ${x^\ast}^c_{\hat{u},s}$ is no smaller than the \roundingalpha\ $\alpha$. In other words, given $(c,s,\alpha)$, \stagethree\ co-displays $c$ to a target subgroup $U$ that consists of every eligible user $\hat{u}$ with ${x^\ast}^c_{\hat{u},s} \geq \alpha$. Therefore, the \roundingalpha\ $\alpha$ plays a key role to the performance bound in the formation of subgroups. Later we prove that with the above strategy, for any pair of users $u,v$ and any item $c$, $\text{Pr}(u \xleftrightarrow{c} v) \geq \frac{{y^\ast}^c_e}{4}$; or equivalently, the expected social utility of $u$ from viewing $c$ with $v$ obtained in the final \configuration\ is at least a constant factor within that in the optimal LP solution.

\algo\ repeats the process of parameter sampling and \stagethree\ until a \textit{feasible} \configuration\ is fully constructed, i.e., each user is displayed exactly one item at each slot, and the no-duplication constraint is satisfied.

\opt{full}{
\para{Revisiting independence vs. dependence.} Recall the troublesome input instance in the proof of Lemma \ref{lemma:trivial_rounding} where independent rounding performs poorly. By exploiting the dependent rounding scheme in \stagethree, since ${x^\ast}^c_{u,s} = \frac{1}{m}$ for all $c,u,s$, upon the first time a \roundingalpha\ $\alpha \leq \frac{1}{m}$ is sampled, \stagethree\ co-displays the \roundingitem\ $c$ to \textit{every user} in the shopping group, which is the optimal solution. On the other hand, independent rounding could not facilitate co-displaying an item to all users as effectively.
}

\begin{example} \label{example:random}
\sloppy For Example \ref{example:running} with the \weight s shown in Table \ref{table:fragsol}, assume that the set of focal parameters are sampled as $(c,s,\alpha) = (c_1,3,0.06).$ Since ${x^\ast}^{c_1}_{u_A,3} = {x^\ast}^{c_1}_{u_B,3} = {x^\ast}^{c_1}_{u_D,3} = 0.33 > 0.06 > {x^\ast}^{c_1}_{u_C,3} = 0$, \stagethree\ co-displays the tripod to the subgroup \{Alice, Bob, Dave\} at slot 3. \opt{full}{Note that this solution is not the \scenario\ configuration in Figure \ref{fig:comp_example}. }
Next, for the second set of parameters $(c,s,\alpha) = (c_4,2,0.22)$, \{Bob, Charlie, Dave\} is formed and co-displayed the memory card at slot 2, since ${x^\ast}^{c_4}_{u_B,2} = {x^\ast}^{c_4}_{u_C,2} = {x^\ast}^{c_4}_{u_D,2} = 0.33 > 0.22 > {x^\ast}^{c_4}_{u_A,2} = 0$. \opt{short}{The subsequent sets of parameters $(c,s,\alpha)$ are respectively $(c_3,1,0.04)$, $(c_5,3,0.2)$, $(c_5,1,0.31)$, $(c_2,1,0.01)$, and $(c_2,2,0.19)$ in the next five iterations, achieving a total \scenario\ utility of 9.75. Please see Section \ref{subsec:details} in the full version \cite{Online} for more details.}
\opt{full}{In the third iteration, RFS selects $(c,s,\alpha) = (c_3,1,0.04)$. As only ${x^\ast}^{c_3}_{u_C,1} = 0.33$ is nonzero among the \weight s for $c_3$ at slot 1, \stagethree\ assigns PSD to \{Charlie\} alone at slot 1. Next, in iteration 4, $(c,s,\alpha) = (c_5,3,0.2)$. At this moment, only Charlie has not been assigned an item at slot 3 since the others are co-displayed the tripod earlier. Because \weight\ ${x^\ast}^{c_3}_{u_C,5} = 0.33 > 0.2$, the SP camera is displayed to \{Charlie\} at slot 3. Iteration 5 generates $(c,s,\alpha) = (c_5,1,0.31)$. For the users without displayed items at slot 1, only Alice and Dave (but not Bob) have their \weight s of the SP camera larger than 0.31. Thus \{Alice, Dave\} are co-displayed the item at slot 1. The final two iterations with ${x^\ast}^{c_3}_{u_C,5}$ as $(c_2,1,0.01)$ and $(c_2,2,0.19)$ displays the DSLR camera to \{Bob\} at slot 1 and \{Alice\} at slot 2. This finalizes the construction of a \configuration\ as represented in Table \ref{table:algr_result}, achieving a total \scenario\ utility of 9.75.\footnote{It is worth noting that those iterations with focal parameters not leading to any item display are omitted from this example. Indeed, it suffices for RFS to sample from the combinations of focal parameters that does result in item displays, which also help improve the practical efficiency of \algor. We revisit this issue in the proof of Theorem \ref{thm:randomized_bound_long} and also in the enhancement strategies in Section \ref{subsec:enhancements}.}}
\hfill \qedsymbol
\end{example}

\opt{short}{The theoretical guarantee of the \algo\ algorithm is given in the following results, which we explicitly prove in Section \ref{subsec:details} in the full version \cite{Online} due to the space constraint.} \opt{full}{We then show that \algo\ is a 4-approximation algorithm for \prob\ in expectation.}

\begin{restatable}{theorem}{randombound} \label{thm:randomized_bound_long}
\sloppy Given the optimal \fragsol\ $X^\ast$, \algo\ returns an expected 4-approximate \configuration\ in $O(n^2 \cdot k)$-time.
\end{restatable}

\opt{full}{
Let $\text{OPT}$ be the optimal total \scenario\ utility in \prob. An \textit{iteration} of \algor\ includes 1) \stagetwo\ sampling a $(c,s,\alpha)$ and 2) \stagethree\ with $(c,s,\alpha)$. Let $\mathcal{R}$ denote the total \scenario\ utility achieved by \algor. Moreover, let $\mathcal{R}_{\text{pre}}$ and $\mathcal{R}_{\text{soc}}$ denote the total preference and social utilities achieved by \algor, respectively. Let $\text{OPT}$ be the optimal total \scenario\ utility of \prob. Let $X$ and $Y$ be the solutions found by \algor, i.e., feasible solutions to $x^c_{u,s}$, $x^c_u$, $y^c_{e,s}$, and $y^c_e$ in Section \ref{subsec:hardness_ip}. Let $\text{w}^c_e = \tau(u,v,c)+\tau(v,u,c)$ for all $e = (u,v) \in E$. Therefore, $ \mathbb{E}[\mathcal{R}] = \mathbb{E}[\mathcal{R}_{\text{pre}}] + \mathbb{E}[\mathcal{R}_{\text{soc}}],$ where 

\begin{align*}
    \mathbb{E}[\mathcal{R}_{\text{pre}}] &= \mathbb{E}[\sum_{c \in \mathcal{C}} \sum_{u \in V} \text{p}(u, c) \cdot x^c_u]\\
    \mathbb{E}[\mathcal{R}_{\text{soc}}] &= \mathbb{E}[\sum_{c \in \mathcal{C}} \sum_{e \in E} \text{w}^c_e \cdot y^c_e]
\end{align*}

We first prove that 
$ \mathbb{E}[\mathcal{R}_{\text{per}}] \geq \sum\limits_{c \in \mathcal{C}} \sum\limits_{u \in V} \big( \text{p}(u, c) \cdot \frac{{x^\ast}^c_u}{2} \big).$ Based on the definition of \algor, we have the following observation.

\begin{lemma} \label{observation:x}
In any iteration $t$, if $u$ is eligible for $(c,s)$, the probability $P^u_{\text{rec}}$ that $\mathbf{A}(u,s) \gets c$ is $\frac{{x^\ast}^c_{u,s}}{k \cdot m}$ (where $k$ and $m$ are respectively the numbers of slots and items) since $c$ and $s$ are selected randomly, 
and $\alpha$ is uniformly chosen from $[0,1]$.
\end{lemma}

\noindent Note that the above observation only gives the conditional probability of $u$ being assigned $c$ at slot $s$ when $u$ is still eligible for $(c,s)$ at the beginning of iteration $t$. Thus, we also need to derive the probability that $u$ is eligible for $(c,s)$ for each iteration $t$. 

\begin{lemma} \label{lemma:noteligible}
In any iteration $t$, for any $c$, $s$ and a user $u$ eligible for $(c,s)$, the probability $P^u_{\text{ne}}$ that $u$ is not eligible for $(c,s)$ in iteration $(t+1)$ is at most $\frac{2}{k \cdot m}$.
\end{lemma}

\begin{proof}
User $u$ is not eligible for $(c,s)$ in iteration $(t+1)$ when one of the following cases occurs in iteration $t$: 1) $u$ is displayed $c$ in some slot $\hat{s}$, or 2) $u$ is displayed some item $\hat{c}$ at slot $s$. From Lemma \ref{observation:x}, the probabilities for the above two cases are at most $\sum\limits_{\hat{s}=1}^k \frac{{x^\ast}^c_{u,\hat{s}}}{k \cdot m}$ and $\sum\limits_{\hat{c} \in \mathcal{C}} \frac{{x^\ast}^{\hat{c}}_{u,s}}{k \cdot m}$, respectively. Recall that for any $u$, $\sum\limits_{\hat{s}=1}^{k} {x^\ast}^c_{u,\hat{s}} \leq 1$ and $\sum\limits_{\hat{c} \in \mathcal{C}} {x^\ast}^{\hat{c}}_{u,s} = 1$ in LP relaxation. Therefore, the total probability of the above cases is at most 
$\frac{1}{k \cdot m} (\sum\limits_{\hat{s}=1}^k {x^\ast}^c_{u,\hat{s}} + \sum\limits_{\hat{c} \in \mathcal{C}} {x^\ast}^{\hat{c}}_{u,s}) \leq \frac{2}{k \cdot m}.$
\end{proof}

In the following, we first consider the case that $u$ is eligible for $(c,s)$ in the beginning of iteration $t$. According to Lemma \ref{observation:x}, the probability $P_{\text{rec}}$ that $u$ is displayed $c$ in slot $s$ in this iteration is $P_{\text{rec}} = \frac{{x^\ast}^c_{u,s}}{k \cdot m}$. Moreover, according to above, let $P_{\text{ne}}$ denote the probability of $u$ losing eligibility for $(c,s)$ in this iteration, then $P_{\text{ne}} \leq \frac{2}{k \cdot m}$. Therefore, we have

\begin{align*}
    & \text{Pr}(\mathbf{A}(u,s) = c) \\ &= \sum\limits^\infty_{t=1} \text{Pr($\mathbf{A}(u,s) \leftarrow c$ in iteration $t$)} \\
    &= \sum^\infty_{t=1} P^u_{\text{rec}} \cdot \text{Pr[$u$ is eligible for $(c,s)$ in the $t$-th iteration]}\\
    &= \sum^\infty_{t=1} P^u_{\text{rec}} \cdot (1-P^u_{\text{ne}})^{t-1} = \frac{P^u_{\text{rec}}}{P^u_{\text{ne}}} \geq \frac{{x^\ast}^c_{u,s}}{2},
\end{align*}

\noindent where $t = \infty$ is allowed in the analysis (but not in the algorithm design) since an empty target group can be randomly generated here (explained later). Thus, $\mathbb{E}[\mathcal{R}_{\text{per}}] \geq \sum_{c \in \mathcal{C}} \sum_{u \in V} \big( \text{p}(u, c) \cdot \frac{{x^\ast}^c_u}{2} \big).$ Next, we aim to use a similar approach to prove that $\mathbb{E}[\mathcal{R}_{\text{soc}}] \geq \frac{1}{4} \sum\limits_{c \in \mathcal{C}} \sum\limits_{e \in E} \text{w}^c_e \cdot {y^\ast}^c_e$. To prove this for the more complicated social utility, instead of directly analyzing $\mathbb{E}[\mathcal{R}_{\text{soc}}]$, 
we first consider the case that the social utility $\tau(u,v,c)$ is generated when both $u$ and $v$ are co-displayed $c$ in the same iteration, and let $\mathbb{E}[\mathcal{R}'_{\text{soc}}]$ denote the expected total social utility in this case. Clearly, $\mathbb{E}[\mathcal{R}_{\text{soc}}] \geq \mathbb{E}[\mathcal{R}'_{\text{soc}}]$. Similarly, we have the following observations.

\begin{lemma} \label{observation:y}
In any iteration $t$, for any pair of users $e = (u,v)$ with both $u$ and $v$ eligible for $(c,s)$, the probability that $\mathbf{A}(u,s) \gets c $ \textit{or} $\mathbf{A}(v,s) \gets c $ is $\frac{\max \{ {x^\ast}^c_{u,s}, {x^\ast}^c_{v,s} \}}{k \cdot m}$, and the probability $P^e_{\text{rec}}$ that $\mathbf{A}(u,s) \gets c $ \textit{and} $\mathbf{A}(v,s) \gets c $ is $\frac{\min \{ {x^\ast}^c_{u,s}, {x^\ast}^c_{v,s} \}}{k \cdot m}$.
\end{lemma}

The reason of Observation \ref{observation:y} is as follows. If $\alpha \leq \min\limits_{u \in U} {x^\ast}^c_{u,s}$, \stagethree\ assigns $c$ at $s$ to all users in $U$ since $x^c_{u,s} \geq \alpha$ for all $u \in U$. Similarly, if $\alpha \leq \max\limits_{u \in U} {x^\ast}^c_{u,s}$, \stagethree\ at least assigns $c$ to the user $u$ with the largest $x^c_{u,s}$. For each iteration $t$, the following lemma then bounds the probability that at least one user in a group $U$ \textit{loses eligibility} for $(c,s)$ in iteration $t$, either due to the no-duplication constraint or due to the assignment of some other item at slot $s$. 

\begin{lemma} \label{lemma:ratio}
In any iteration $t$, for any pair of users $e = (u,v)$ with $u$ and $v$ eligible for $(c,s)$, the probability $P^e_{\text{ne}}$ that \textit{at least one} of $u$ and $v$ is not eligible for $(c,s)$ in iteration $(t+1)$ is at most $\frac{4}{k \cdot m}$.
\end{lemma}

\begin{proof}
At least one of $u$ and $v$ is not eligible for $(c,s)$ in iteration $(t+1)$ when one of the following cases occurs in iteration $t$: 1) $u$ or $v$ is displayed $c$ in some slot $\hat{s}$, or 2) $u$ or $v$ is displayed some item $\hat{c}$ in slot $s$. From Lemma \ref{observation:y}, the probabilities for the above two cases are at most $\sum\limits_{\hat{s}=1}^k \frac{\max \{ {x^\ast}^c_{u,\hat{s}}, {x^\ast}^c_{v,\hat{s}} \}}{k \cdot m}$ and $\sum\limits_{\hat{c} \in \mathcal{C}} \frac{\max \{ {x^\ast}^{\hat{c}}_{u,s}, {x^\ast}^{\hat{c}}_{v,s} \}}{k \cdot m}$, respectively. Recall that for any $u$, $\sum\limits_{\hat{s}=1}^{k} {x^\ast}^c_{u,\hat{s}} \leq 1$ and $\sum\limits_{\hat{c} \in \mathcal{C}} {x^\ast}^{\hat{c}}_{u,s} = 1$ in LP relaxation. Therefore, the total probability of the above cases is at most 

\begin{align*}
    &\frac{1}{k \cdot m} (\sum_{\hat{s}=1}^k \max \{ {x^\ast}^c_{u,\hat{s}}, {x^\ast}^c_{v,\hat{s}} \} + \sum_{\hat{c} \in \mathcal{C}} \max \{ {x^\ast}^{\hat{c}}_{u,s}, {x^\ast}^{\hat{c}}_{v,s} \}) \\
    &\leq \frac{1}{k \cdot m} (\sum_{\hat{s}=1}^k ( {x^\ast}^c_{u,\hat{s}} + {x^\ast}^c_{v,\hat{s}} ) + \sum_{\hat{c} \in \mathcal{C}} ( {x^\ast}^{\hat{c}}_{u,s} + {x^\ast}^{\hat{c}}_{v,s} ) )\\
    &= \frac{1}{k \cdot m} (\sum_{\hat{s}=1}^k {x^\ast}^c_{u,\hat{s}} + \sum_{\hat{s}=1}^k {x^\ast}^c_{v,\hat{s}}  + \sum_{\hat{c} \in \mathcal{C}} {x^\ast}^{\hat{c}}_{u,s} + \sum_{\hat{c} \in \mathcal{C}} {x^\ast}^{\hat{c}}_{v,s}) \\
    &\leq \frac{4}{k \cdot m}.
\end{align*}
\end{proof}

Similarly, consider the case that in the beginning of iteration $t$, both $u,v$ are eligible for $(c,s)$. 
Let $u \xleftrightarrow[s]{c} v |_{t}$ denote that $u,v$ are co-displayed $c$ at slot $s$ in iteration $t$. Therefore, we have
\begin{align*}
    &\text{Pr}[u \xleftrightarrow[s]{c} v] = \sum\limits^\infty_{t=1} \text{Pr}[u \xleftrightarrow[s]{c} v |_{t}]\\
    &= \sum^\infty_{t=1} P^e_{\text{rec}} \cdot \text{Pr[$u,v$ are both eligible for $(c,s)$ in $t$-th iteration]}\\
    &=\sum\limits^\infty_{t=1} P^e_{\text{rec}} \cdot \text{Pr[$u,v$ both eligible for $(c,s)$ in $t$-th iteration]}\\
    &= \sum\limits^\infty_{t=1} P^e_{\text{rec}} \cdot (1-P^e_{\text{ne}})^{t-1} = \frac{P^e_{\text{rec}}}{P^e_{\text{ne}}}\\
    &\geq \frac{\frac{\min\{{x^\ast}^c_{u,s}, {x^\ast}^c_{v,s}\}}{k \cdot m}}{\frac{4}{k \cdot m}}\\
    &\geq \frac{\min\{{x^\ast}^c_{u,s}, {x^\ast}^c_{v,s}\}}{4} = \frac{{y^\ast}^c_{e,s}}{4}.
\end{align*}

Finally, because of the no-duplicate constraint, for all $u$, $v$ and $c$, the events $u \xleftrightarrow[s]{c} v |_{t}$ and $u \xleftrightarrow[s']{c} v |_{t}$ for different slots $s$ and $s'$ are mutually exclusive. Similarly, the events $u \xleftrightarrow[s]{c} v |_{t}$ and $u \xleftrightarrow[s]{c} v |_{t'}$ for $t \neq t'$ are also mutually exclusive because every user sees exactly one item at each slot. Therefore, 

\begin{align*}
    \text{Pr}(u \xleftrightarrow{c} v) &\geq \sum_{s=1}^{k}\sum_{t=1}^{\infty} \text{Pr}(u \xleftrightarrow[s]{c} v |_{t}),\\
    \mathbb{E}[\mathcal{R}_{\text{soc}}]  &\geq \mathbb{E}[\mathcal{R}'_{\text{soc}}]\\
    &= \sum_{c \in \mathcal{C}} \sum_{e \in E} \text{w}^c_e \cdot \sum^\infty_{t=1} \text{Pr}(u \xleftrightarrow[s]{c} v |_{t})\\
    &\geq \sum_{c \in \mathcal{C}} \sum_{e \in E} \text{w}^c_e \cdot \frac{{y^\ast}^c_e}{4}.
\end{align*}    
    
\noindent Therefore,

\begin{align*}
\mathbb{E}[\mathcal{R}] &= \mathbb{E}[\mathcal{R}_{\text{per}}] + \mathbb{E}[\mathcal{R}_{\text{soc}}]\\
&\geq \sum_{c \in \mathcal{C}} \sum_{u \in V} \text{p}(u, c) \cdot \frac{{x^\ast}^c_u}{2} +\sum_{c \in \mathcal{C}} \sum_{e \in E} \text{w}^c_e \cdot \frac{{y^\ast}^c_e}{4}\\
&\geq \frac{\text{OPT}}{4},
\end{align*}  

\noindent which proves the approximation ratio.  \hfill \qedsymbol

g
In the above derivation, 
a large $\alpha$ could lead to an empty target group if it exceeds ${x^\ast}^c_{u,s}$ for every $u$ eligible for $(c,s)$. Therefore, the total number of iterations could approach $\infty$. To address the above issue, instead of sampling $(c,s,\alpha)$ uniformly from all possible combinations, \algor\ samples $(c,s,\alpha)$ uniformly from only the combinations generating nonempty target groups (i.e., 
an enormous $\alpha$ is no longer chosen). Because the setting of $\alpha$ is independent for each iteration, and the probability of selecting each combination to generate a nonempty target group remains equal, the expected solution quality is also identical for \algor. Therefore, the number of iterations for \algo\ is effectively reduced to $O(nk)$, and \stagethree\ in each iteration requires $O(n)$-time. The total time complexity of \algor, including the \stageone, is thus $O(\text{LP}) + O(n^2 \cdot k)$, where $O(\text{LP})$ is the complexity\footnote{The current best time complexity for solving linear program equals the complexity of matrix multiplication, or roughly $O(N^{2.5})$ for $N$ variables \cite{LP}. However, practical computation generally takes much less time.} of solving $X^\ast$.
Some immediate corollaries are directly obtained from Theorem \ref{thm:randomized_bound_long} as follows.
}

\begin{restatable}{corollary}{repeatbound} \label{thm:randomized_corollary}
Repeating \algor\ and selecting the best output returns a $(4+\epsilon)$-approximate \configuration\ in $O(n^2 \cdot k \cdot \log_{\epsilon} n)$-time with high probability, i.e., with a probability $1-\frac{1}{n^{O(1)}}$. 
\end{restatable}

\begin{restatable}{corollary}{compositebound} \label{thm:composite_corollary}
\opt{full}{\sloppy}Given a (non-optimal) \fragsol\ $\Tilde{X}^\ast$ as a $\beta$-approximation of the LP relaxation problem, \algo\ returns an expected $(4 \cdot \beta)$-approximate \configuration. 
\end{restatable}

\opt{full}{\begin{restatable}{corollary}{kisone} \label{thm:kisone_corollary}
For $k=1$, given the optimal \fragsol\ $X^\ast$, \algo\ returns an expected 2-approximate \configuration\ in $O(n^2 \cdot k)$-time.
\end{restatable}
}

\opt{full}{
\begin{proof}
For the first corollary, from Theorem \ref{thm:randomized_bound_long}, \algor\ achieves an expected total \scenario\ utility $\mathbb{E}[\mathcal{R}] \geq \frac{OPT}{4}$. Let $\mathcal{R}' = OPT - \mathcal{R}$ denote the gap between $\mathcal{R}$ (the objective achieved by \algor) and $OPT$ (the optimal objective). Clearly, $\mathcal{R}'$ is non-negative, and we have $\mathbb{E}[\mathcal{R}'] = \mathbb{E}[OPT - \mathcal{R}] = OPT - \mathbb{E}[\mathcal{R}] \leq \frac{3 \cdot OPT}{4}$. Therefore, by the Markov inequality, the probability that a single invocation of \algor\ failing to return a $(4+\epsilon)$-approximate \configuration\ is $\mathcal{R}' > (1 - \frac{1}{4+\epsilon}) \cdot OPT$ is 
\begin{align*}
    \text{Pr}[\mathcal{R}' \geq OPT \cdot (1 - \frac{1}{4+\epsilon})] &\leq \frac{\mathbb{E}[\mathcal{R}']}{(1 - \frac{1}{4+\epsilon}) \cdot OPT}\\
    &\leq \frac{\frac{3 \cdot OPT}{4}}{(1 - \frac{1}{4+\epsilon}) \cdot OPT}\\
    &= \frac{12+3\epsilon}{12+4\epsilon}.
\end{align*}

Therefore, if \algor\ is repeated $n_{\text{repeat}}$ times, the probability that the best solution is not a $(4+\epsilon)$-approximate \configuration\ is at most $ (\frac{12+3\epsilon}{12+4\epsilon})^{n_{\text{repeat}}}$. By setting $n_{\text{repeat}} = \log_{\frac{12+4\epsilon}{12+3\epsilon}} n$, the success rate will be at least $1-(\frac{12+3\epsilon}{12+4\epsilon})^{n_{\text{repeat}}} = 1 - \frac{1}{n}.$

A single run of \algor\ is in $O(n^2 \cdot k)$-time. Therefore, repeating it for $\log_{\frac{12+4\epsilon}{12+3\epsilon}} n$ times can be done in $O(n^2 \cdot k \cdot \log_{\frac{12+4\epsilon}{12+3\epsilon}} n) = O(n^2 \cdot k \cdot \log_{\epsilon} n)$-time.

For the second corollary, note that the total objective achieved by the non-optimal approximate LP solution $\Tilde{X}^\ast$ is at least $\beta \cdot OPT$, since it is a $\beta$-approximation of the LP relaxation problem, which has an optimal objective of at least $OPT$. Therefore, from Theorem \ref{thm:randomized_bound_long}, running \algo\ with $\Tilde{X}^\ast$ achieves an expected \scenario\ utility of at least $\frac{OPT}{4 \cdot \beta}$.

Finally, for the third corollary, simply note that for $k=1$, the final sum in Lemma \ref{lemma:ratio} changes from $\frac{4}{k \cdot m}$ to $\frac{2}{k \cdot m}$, since the first case in the proof of Lemma \ref{lemma:ratio} is impossible to happen in $k=1$. Following the original proof but plugging in this change implies the 2-approximation of \algo\ when $k=1$.
\end{proof}
}

The second corollary is particularly useful in practice for improving the efficiency of \algo\ since state-of-the-art LP solvers often reach a close-to-optimal solution in a short time but need a relatively long time to return the optimal solution, especially for large inputs. Therefore, it allows for a quality-efficiency trade-off in solving \prob.

\subsection{Derandomizing \algo} \label{subsec:derandomize}
\opt{full}{\sloppy}From the investigation of \algo, we observe that the \roundingalpha\ $\alpha$ plays a key role in forming effective target subgroups in \stagethree. If $\alpha$ is close to $0$, \stagethree\ easily forms a large subgroup consisting of all users and co-displays the \roundingitem\ to them, similar to the ineffective \group\ approach. On the other hand, large $\alpha$ values lead to small subgroups, not good for exploiting social interactions. Due to the randomness involved in \algo, these caveats cannot be completely avoided. To address these issues, we aim to further strengthen the performance guarantee of \algo\ by derandomizing the selection of focal parameters to obtain a stronger version of the algorithm, namely \algodfull\ (\algod), which is a \textit{deterministic} 4-approximation algorithm. First, we observe that $\alpha$ can be assigned in a discrete manner.

\begin{observation} \label{observation:poly_many}
    There are $O(knm)$ distinct possible outcomes in \stagethree, each corresponding to a \roundingalpha\ $\alpha = {x^\ast}^c_{u,s}$, i.e., the \weight\ of an item $c$ to a user $u$ at a slot $s$.
\end{observation}

The above observation can be verified as follows. Given $c$ and $s$, the outcome of \stagethree\ with \roundingalpha\ $\alpha = x$, for any $x \in [0,1]$, is equivalent to that with $\alpha$ set to the smallest ${x^\ast}^c_{u,s} \geq x$. It enables us to derandomize \algo\ effectively. Instead of randomly sampling $(c,s,\alpha)$, we carefully evaluate the outcomes (of \stagethree) from setting ${\alpha}$ to every possible ${x^\ast}^c_{u,s}$ in the optimal \fragsol. Intuitively, it is desirable to select an ${\alpha}$ that results in the largest increment in the total \scenario\ utility. However, this short-sighted approach ignores the potentially significant increase in total \scenario\ utility in the future from the remaining users and slots that have not been processed. In fact, it always selects an outcome with $\alpha = 0$ to maximize the current utility increment. Therefore, it is necessary for \algod\ to carefully evaluate all potential future allocations of items.

\opt{short}{\revise{Specificially, \algod\ selects the set of focal parameters $(c,s,\alpha)$ to maximize a weighted sum of 1) the increment of \scenario\ utility in the current iteration via \stagethree\ with $(c,s,\alpha)$, and 2) the expected \scenario\ utility in the future from item assignments at slots that are \textit{left unfilled} in \stagethree\ with $(c,s,\alpha)$. Due to the space constraint, please see Section \ref{subsec:derandomize} in the full version \cite{Online} for the details of \algod, and the proof of the following theoretical result.}}

\opt{full}{
Specifically, let $S_{\text{cur}} = \{ (\hat{u}, \hat{s}) | \mathbf{A}(\hat{u}, \hat{s}) = \text{NULL} \}$ denote the set of \textit{display units}, i.e., a slot $\hat{s}$ for a user $\hat{u}$, that have not been assigned an item before the current iteration. Let $S_{\text{tar}}(c,s,{x^\ast}^c_{u,s}) = \{ (\hat{u}, s) \in S_{\text{cur}} | {x^\ast}^{c}_{\hat{u},s} \geq {x^\ast}^c_{u,s}\}$ denote the set of display units to be assigned $c$ in \stagethree\ with ${\alpha} = {x^\ast}^c_{u,s}$. Let $S_{\text{fut}}(c,s,{x^\ast}^c_{u,s}) = S_{\text{cur}} \setminus S_{\text{tar}}(c,s,{x^\ast}^c_{u,s})$ denote the set of remaining display units to be processed in the future. Moreover, let ${\text{ALG}}(S_{\text{tar}}(c,s,{x^\ast}^c_{u,s}))$ be the \scenario\ utility gained by co-displaying $c$ to the target subgroup. Let ${\text{OPT}}_{\text{LP}}(S_{\text{fut}}(c,s,{x^\ast}^c_{u,s}))$ represent the expected \scenario\ utility acquired in the future from the display units in $S_{\text{fut}}(c,s,{x^\ast}^c_{u,s})$. To strike a balance between the increment of \scenario\ utility in the current iteration and in the future, \algod\ examines every ${\alpha} = {x^\ast}^c_{u,s}$ to maximize 
\begin{align*}
    f(c,s,{x^\ast}^c_{u,s}) &= {\text{ALG}}(S_{\text{tar}}(c,s,{x^\ast}^c_{u,s})) \\
    &+ r \cdot {\text{OPT}}_{\text{LP}}(S_{\text{fut}}(c,s,{x^\ast}^c_{u,s})),
\end{align*}
\vspace{-15pt}

\noindent where 
\vspace{-15pt}

\begin{small}
\begin{align*}
    {\text{ALG}}(S_{\text{tar}}(c,s,{x^\ast}^c_{u,s})) &= \sum\limits_{(\hat{u},s) \in S_{\text{tar}}} \text{p}(\hat{u}, c) + \sum\limits_{\substack{e = (\hat{u},\hat{v}) \\ (\hat{u},s), (\hat{v},s) \in S_{\text{tar}}}} \text{w}^c_e, \\
    {\text{OPT}}_{\text{LP}}(S_{\text{fut}}(c,s,{x^\ast}^c_{u,s})) &= \sum\limits_{c \in \mathcal{C}} \sum\limits_{s=1}^{k} [ \sum\limits_{(\hat{u},\hat{s}) \in S_{\text{fut}}} \text{p}(\hat{u}, c) {x^\ast}^c_{\hat{u},\hat{s}} \\ &+ \sum\limits_{\substack{e = (\hat{u},\hat{v}) \\ (\hat{u},\hat{s}), (\hat{v},\hat{s}) \in S_{\text{fut}}}} \text{w}^c_e {y^\ast}^c_{e,\hat{s}} ],
\end{align*}
\end{small}

\noindent $\text{w}^c_e = \tau(\hat{u},\hat{v},c)+\tau(\hat{v},\hat{u},c)$ for all $e = (\hat{u},\hat{v}) \in E$, and $r$ is the balancing ratio. In other words, \algod\ simultaneously optimizes the current and potential future \scenario\ utility in each iteration. It is summarized in Algorithm \ref{alg:deterministic_alignment_aware_rounding}.

\begin{algorithm}[h]  
    \caption{\algodfull}
    \label{alg:deterministic_alignment_aware_rounding}
    \begin{algorithmic}[1]
        \Require $X^\ast$
        \Ensure An \configuration\ $\mathbf{A}$
        \State $\mathbf{A}(\hat{u},\hat{s}) \gets \text{NULL}$ for all $\hat{u},\hat{s}$
        \State $X^\ast \gets X^\ast_{\text{LP}}$ 
        \While {some entry in $\mathbf{A}$ is $\text{NULL}$} 
                \State $c,s,u \gets \argmax\limits_{c,s,u} f(c,s,{x^\ast}^c_{u,s})$
                \State $\alpha \gets {x^\ast}^c_{u,s}$
            \For{$\hat{u} \in V$}
                \If{$\mathbf{A}(\hat{u},s) = \text{NULL}$ and $\mathbf{A}(\hat{u},t) \neq c \, \forall t \neq s$} 
                    \If{${x^\ast}^c_{\hat{u},s} \geq \alpha$} 
                        \State $\mathbf{A}(\hat{u},s) \gets c$
                    \EndIf
                \EndIf
            \EndFor
        \EndWhile
        \Return $\mathbf{A}$
    \end{algorithmic}
\end{algorithm}
\normalsize


\opt{full}{
\begin{table}[t]
	\centering
	\caption{\configuration\ returned by \algor\ for Example \ref{example:random}.}
	\label{table:algr_result}
	\begin{tabular}{l|ccc}
  &  Slot 1 & Slot 2 & Slot 3\\ \hline 
 $u_A$ &  $c_5$  & $c_2$  &  $c_1$    \\
 $u_B$ &  $c_2$  & $c_4$  &  $c_1$    \\
 $u_C$ &  $c_3$  & $c_4$  &  $c_5$  \\
 $u_D$ &  $c_5$  & $c_4$  &  $c_1$     \\
	\end{tabular}
\end{table}
\begin{table}[t]
	\centering
	\caption{\configuration\ returned by \algod\ for Example \ref{example:random}.}
	\label{table:algd_result}
	\begin{tabular}{l|ccc}
  &  Slot 1 & Slot 2 & Slot 3\\ \hline 
 $u_A$ &  $c_5$  & $c_1$  &  $c_2$    \\
 $u_B$ &  $c_5$  & $c_1$  &  $c_2$    \\
 $u_C$ &  $c_5$  & $c_3$  &  $c_2$  \\
 $u_D$ &  $c_5$  & $c_1$  &  $c_4$     \\
	\end{tabular}
\end{table}
\begin{table}[t]
	\centering
	\caption{Configuration by the baseline approaches for Example \ref{example:random}.}
	\label{table:baseline_result}
	\begin{tabular}{l|ccc}
 \personalized\ &  Slot 1 & Slot 2 & Slot 3\\ \hline 
 $u_A$ &  $c_5$  & $c_2$  &  $c_1$    \\
 $u_B$ &  $c_2$  & $c_1$  &  $c_4$    \\
 $u_C$ &  $c_3$  & $c_4$  &  $c_2$  \\
 $u_D$ &  $c_4$  & $c_5$  &  $c_3$     \\\\
 
 \group\ &&&\\ \hline 
 $\{u_A, u_B, u_C, u_D\}$ &  $c_5$  & $c_1$  &  $c_2$    \\\\
 
 \friendship\ &&&\\ \hline 
 $\{u_A,u_D\}$ &  $c_5$  & $c_1$  &  $c_4$    \\
 $\{u_B,u_C\}$ &  $c_2$  & $c_4$  &  $c_3$    \\\\
 
 \preference\ &&&\\ \hline 
 $\{u_A,u_B\}$ &  $c_2$  & $c_1$  &  $c_5$    \\
 $\{u_C,u_D\}$ &  $c_4$  & $c_5$  &  $c_3$    \\
	\end{tabular}
\end{table}
}

\begin{example}
\sloppy For the instance in Example \ref{example:running}, the first iteration of \algod\ with $r = \frac{1}{4}$ sets $\alpha = {x^\ast}^{c_5}_{u_B,1} = 0$, and \stagethree\ with $(c,s,{x^\ast}^c_{u,s}) = (c_5,1,0)$ co-displays the SP camera to everyone at slot 1, where $S_{\text{tar}}(c,s,{x^\ast}^c_{u,s}) = \{ (u_A,1), (u_B,1), (u_C,1), (u_D,1)\}$, $S_{\text{fut}}(c,s,{x^\ast}^c_{u,s}) = \{ (u_A,2), (u_B,2), (u_C,2), (u_D,2), (u_A,3), (u_B,3), (u_C,3),$ $(u_D,3)\}$, ${\text{ALG}}(S_{\text{tar}}(c,s,{x^\ast}^c_{u,s})) =  \text{p}(u_A,c_5) + \text{p}(u_B,c_5) + \text{p}(u_C,c_5) + \text{p}(u_D,c_5) + \text{w}^{c_5}_{(u_A,u_B)} + \text{w}^{c_4}_{(u_A,u_C)} + \text{w}^{c_4}_{(u_A,u_D)} + \text{w}^{c_4}_{(u_B,u_C)} = 3.35$, ${\text{OPT}}_{\text{LP}}(S_{\text{fut}}(c,s,{x^\ast}^c_{u,s})) = 6.97$, and $f(c,s,{x^\ast}^c_{u,s}) = 3.35 + \frac{1}{4} \cdot 6.97 = 5.09$ is maximized. The next iteration selects $\alpha = {x^\ast}^{c_1}_{u_A,2} = 0.33$ to co-display the tripod to \{Alice, Bob, Dave\} at slot 2. \algod\ selects $\alpha = {x^\ast}^{c_2}_{u_D,3} = 0$, $\alpha = {x^\ast}^{c_4}_{u_A,3} = 0$, and $\alpha = {x^\ast}^{c_3}_{u_D,2} = 0$ in the next three iterations, resulting in a total \scenario\ utility of 9.85, which is slightly larger than \algor\ (9.75) in Example \ref{example:random}. The returned \configuration\ is shown in Table \ref{table:algd_result}. 

For this running example, the optimal solution is the \configuration\ shown at the top of Figure \ref{fig:example}, with a total \scenario\ utility of 10.35. The \personalized\ approach retrieves the top-$3$ preferred items for each user, e.g., $\langle c_5, c_2, c_1 \rangle$ for Alice, and achieves a total \scenario\ utility of 8.25. For the \group\ approach, the total \scenario\ utility for the universal subgroup \{Alice, Bob, Charlie, Dave\} viewing each item is aggregated. For example, the total \scenario\ utility of the subgroup of all users viewing $c_1$ would be $\sum\limits_{u \in V} \text{p}(u,c_1) + \sum\limits_{u,u' in V} \tau(u,u',c_1)$, which equals the summation of all values in the first row in Table \ref{table:sat}, or 2.6 (note that $\lambda = 0.5$ and the objective is scaled up by 2 here). The top-$3$ items, i.e., $\langle c_5, c_1, c_2 \rangle$, is retrieved and co-displayed to all users. This achieves a total \scenario\ utility of 8.35. 

We also compare with two variations of the subgroup approach, namely the \friendship\ and \preference\ approaches, where the former pre-partitions the subgroups based on social relations among the users, while the latter finds subgroups with similar item preferences. The \friendship\ approach first partitions the social network into two equally sized and dense subgroups \{Alice, Dave\} and \{Bob, Charlie\}. Items are then determined analogously to the \group\ approach, resulting in a total \scenario\ utility of 8.4. Finally, the \preference\ approach partitions the network into \{Alice, Bob\} and \{Charlie, Dave\}. The total \scenario\ utility would be 8.7.

The configurations returned by these four baseline approaches are consistent with Figure \ref{fig:comp_example} and are also summarized in Table \ref{table:baseline_result}. To sum up, the total \scenario\ utility for \algor, \algod, and the \personalized, \group, \friendship\ and \preference\ approaches, are 9.75, 9.85, 8.25, 8.35, 8.4 and 8.7, respectively. Therefore, both \algor\ and \algod\ achieves a near-optimal solution. \hfill \qedsymbol
\end{example}

 
In \algod, $r$ serves as a turning knob to strike a good balance between the current increment in \scenario\ utility and the potential \scenario\ utility in the future. Specifically, we prove that by setting $r = \frac{1}{4}$ in each iteration, \algod\ is a deterministic 4-approximation algorithm for \prob. We also empirically evaluate the performance of \algod\ with different values of $r$ in Section \ref{subsec:diff_r}.
}

\begin{restatable}{theorem}{derandombound} \label{thm:derandomized_bound_long}
\sloppy Given the optimal \fragsol\ $X^\ast$, \algod\ returns a worst-case 4-approximate \configuration\ in $O(n \cdot m \cdot k \cdot |E|)$-time.
\end{restatable}

\opt{full}{
In DPS, let ${\text{OPT}}_{\text{LP}}(S_{\text{cur}})$ denote the total \scenario\ utility from only the slots in $S_{\text{cur}}$ according to $X^\ast$. More specifically, 

\begin{align*}
    {\text{OPT}}_{\text{LP}}(S_{\text{cur}}) = \sum_{\hat{c} \in \mathcal{C}} \sum_{\hat{s}=1}^{k} &\sum_{\substack{(\hat{u},\hat{s}) \in\\ S_{\text{cur}}(c,s,{x^\ast}^c_{u,s})}} (\text{p}(\hat{u}, \hat{c}) {x^\ast}^{\hat{c}}_{\hat{u},\hat{s}} \\
    + &\sum_{\substack{e = (\hat{u},\hat{v}) \\ (\hat{v},\hat{s}) \in S_{\text{cur}}(c,s,{x^\ast}^c_{u,s})}} \text{w}^{\hat{c}}_e {y^\ast}^{\hat{c}}_{e,\hat{s}}).
\end{align*}

Similar to above, an \textit{iteration} for \algod\ refers to DPS selecting $(c,s,{x^\ast}^c_{u,s})$ and the subsequent \stagethree. We first outline the proof sketch. To prove the approximation ratio for \algod, we first need to show that, in an iteration of \algod, \textit{if $(c,s,{x^\ast}^c_{u,s})$ is selected with RPS instead of DPS}, so that the probability of selecting $(c,s,{x^\ast}^c_{u,s})$ is the sum of selecting all $(c,s,\alpha)$ that lead to an equivalent outcome as selecting $(c,s,{x^\ast}^c_{u,s})$, then the expected value of $f(c,s,{x^\ast}^c_{u,s})$ is at least $\frac{{\text{OPT}}_{\text{LP}}(S_{\text{cur}})}{4}$. Thus, for each iteration there exists at least one $(c,u,s)$ with $f(c,s,{x^\ast}^c_{u,s}) \geq \frac{{\text{OPT}}_{\text{LP}}(S_{\text{cur}})}{4}$, and \algod\ is therefore a 4-approximating algorithm.

We begin with proving that in every iteration $t$, if $(c,s,\alpha)$ is selected with RPS instead of DPS, $\mathbb{E}[f(c,s,\alpha)] \geq \frac{1}{4} \cdot {\text{OPT}}_{\text{LP}}(S_{\text{cur}})$, where the function $f(c,s,\alpha)$ is the evaluation criteria in DPS. Therefore, $f(c,s,{x^\ast}^c_{u,s})$ in DPS is also at least $\frac{1}{4} \cdot {\text{OPT}}_{\text{LP}}(S_{\text{cur}})$, since DPS maximizes $f(c,s,{x^\ast}^c_{u,s})$.

\begin{lemma} \label{lemma:derandomization}
In any iteration of \algo, if $(c,s,\alpha)$ is selected by RPS, $\mathbb{E}[f(c,s,\alpha)] \geq \frac{1}{4} \cdot {\text{OPT}}_{\text{LP}}(S_{\text{cur}})$.
\end{lemma}

\begin{proof}
In any iteration $t$, for each $e = (\hat{u},\hat{v})$, $c$, and $s$ with $(\hat{u},s)$ and $(\hat{v},s)$ in $S_{\text{cur}}$ (i.e., $\hat{u}$ and $\hat{v}$ are eligible for $(c,s)$), by Lemma \ref{observation:y}, the probability that $\hat{u}$ and $\hat{v}$ are co-displayed $c$ at slot $s$ is 
$$ \frac{\min \{ {x^\ast}^{c}_{\hat{u},s}, {x^\ast}^{c}_{\hat{v},s}\}}{k \cdot m} = \frac{{y^\ast}^{c}_{e,s}}{k \cdot m}. $$ Therefore,  

\begin{align*}
    & \mathbb{E}[{\text{ALG}}(S_{\text{tar}}(c,s,\alpha))]\\
    =& \mathbb{E}[\sum_{(\hat{u},s) \in S_{\text{tar}}(c,s,\alpha)} \big( \text{p}(\hat{u}, c) + \sum_{(\hat{v},s) \in S_{\text{tar}}(c,s,\alpha)} \tau(\hat{u},\hat{v},c) \big)]\\
    =& \sum_{c \in \mathcal{C}} \sum_{s=1}^{k} \big( \sum_{(\hat{u}, s) \in S_{\text{cur}}} \text{p}(\hat{u}, c) \frac{{x^\ast}^c_{\hat{u}, s}}{k \cdot m} + \sum_{(\hat{u},s), (\hat{v},s) \in S_{\text{cur}}} \text{w}^c_e \frac{{y^\ast}^c_{e,s}}{k \cdot m} \big)\\
    =& \frac{1}{k \cdot m} \sum_{c \in \mathcal{C}} \sum_{s=1}^{k} \big( \sum_{(\hat{u},s) \in S_{\text{cur}}} \text{p}(\hat{u}, c) {x^\ast}^c_{\hat{u}, s} + \sum_{\substack{e = (\hat{u},\hat{v}) (\hat{u},s), \\ (\hat{v},s) \in S_{\text{cur}}}} \text{w}^c_e {y^\ast}^c_{e,s} \big)\\ 
    =&\frac{{\text{OPT}}_{\text{LP}}(S_{\text{cur}})}{k \cdot m}.
\end{align*}

\noindent Next, from Lemma \ref{lemma:noteligible} and \ref{lemma:ratio}, for each $e = (\hat{u},\hat{v})$, $\hat{c}$, and $\hat{s}$ such that $(\hat{u},\hat{s}), (\hat{v},\hat{s}) \in S_{\text{cur}}$, the probability that $(\hat{u},\hat{s}) \in S_{\text{fut}}(c,s,\alpha)$ is at least $1 - P^{\hat{u}}_{\text{ne}}$, and the probability that $(\hat{u},\hat{s}), (\hat{v},\hat{s}) \in S_{\text{fut}}(c,s,\alpha)$ is at least $1 - P^{e}_{\text{ne}}$, which implies the probability that both $\hat{u}$ and $\hat{v}$ are eligible for $(\hat{c},\hat{s})$ in iteration $(t+1)$ is at least $1 - \frac{4}{k \cdot m}$. Therefore, the expected contribution of the social utility from the co-display of $\hat{c}$ to $\hat{u}$ and $\hat{v}$ at slot $\hat{s}$ is at least $\text{w}^{\hat{c}}_e \cdot (1 - \frac{4}{k \cdot m})$. 
Thus, 

\begin{small}
\begin{align*}
    &\mathbb{E}[{\text{OPT}}_{\text{LP}}(S_{\text{fut}}(c,s,\alpha))]\\
    \geq& \sum_{\hat{c} \in \mathcal{C}}\sum_{\hat{s}=1}^{k} [ \sum_{(\hat{u},\hat{s}) \in S_{\text{cur}}}  (1-P^{\hat{u}}_{\text{ne}}) \text{p}^{\hat{c}}_{\hat{u},\hat{s}} {x^\ast}^{\hat{c}}_{\hat{u},\hat{s}} + \sum_{\substack{(\hat{u},\hat{s}),\\ (\hat{v},\hat{s}) \in S_{\text{cur}}}} (1-P^{e}_{\text{ne}}) \text{w}^{\hat{c}}_e {y^\ast}^{\hat{c}}_{e,\hat{s}} ]\\
    \geq& (1 - \frac{4}{k \cdot m}) \cdot \sum_{\hat{c} \in \mathcal{C}}\sum_{\hat{s}=1}^{k} [ \sum_{(\hat{u},\hat{s}) \in S_{\text{cur}}} \text{p}^{\hat{c}}_{\hat{u},\hat{s}} {x^\ast}^{\hat{c}}_{\hat{u},\hat{s}} + \sum_{\substack{(\hat{u},\hat{s}),\\ (\hat{v},\hat{s}) \in S_{\text{cur}}}} \text{w}^{\hat{c}}_e {y^\ast}^{\hat{c}}_{e,\hat{s}}]\\
    =& (1-\frac{4}{k \cdot m}){\text{OPT}}_{\text{LP}}(S_{\text{cur}}).
\end{align*}
\end{small}

\noindent Combining the above, we have 

\begin{small}
\begin{align*}
\mathbb{E}[f(c,s,\alpha)] &= \mathbb{E}[{\text{ALG}}(S_{\text{tar}}(c,s,\alpha))] + \frac{1}{4} \cdot \mathbb{E}[{\text{OPT}}_{\text{LP}}(S_{\text{fut}}(c,s,\alpha))]\\ 
&\geq \frac{{\text{OPT}}_{\text{LP}}(S_{\text{cur}})}{k \cdot m} + \frac{1}{4} \cdot {\text{OPT}}_{\text{LP}}(S_{\text{cur}}) \cdot (1-\frac{4}{k \cdot m}) \\
&= \frac{{\text{OPT}}_{\text{LP}}(S_{\text{cur}})}{4}.
\end{align*}
\end{small}

\noindent The lemma follows.
\end{proof}

Back to \algod, let the total number of iterations be $T$. Let $S^{t}_{\text{tar}}$ and $S^{t}_{\text{fut}}$ denote $S_{\text{tar}}(c,s,{x^\ast}^c_{u,s})$ and $S_{\text{fut}}(c,s,{x^\ast}^c_{u,s})$ in the $t$-th iteration of \algod , and $S^{t}_{\text{cur}} = S^{t}_{\text{tar}} \cup S^{t}_{\text{fut}}$. By Lemma \ref{lemma:derandomization}, $f(c,s,{x^\ast}^c_{u,s}) = \max\limits_{\hat{c},\hat{s},\hat{u}} f(\hat{c},\hat{s},{x^\ast}^{\hat{c}}_{\hat{u},\hat{s}}) \geq \mathbb{E}[f(c,s,\alpha)] \geq \frac{\text{OPT}_{\text{LP}}(S_{\text{cur}})}{4}$ for every iteration $t$ in \algod. Therefore, since the solution of \algod\ performs no worse than the randomized approach in the above lemma,

\begin{small}
\begin{align*}
    \text{ALG}(S^1_{\text{tar}}) + \frac{1}{4} \cdot \text{OPT}_{\text{LP}}(S^1_{\text{fut}}) &\geq \frac{1}{4} \cdot \text{OPT}_{\text{LP}}(S^1_{\text{cur}})  \\
    \text{ALG}(S^2_{\text{tar}}) + \frac{1}{4} \cdot \text{OPT}_{\text{LP}}(S^2_{\text{fut}}) &\geq \frac{1}{4} \cdot \text{OPT}_{\text{LP}}(S^2_{\text{cur}}) = \frac{1}{4} \cdot \text{OPT}_{\text{LP}}(S^1_{\text{fut}})\\
    &\vdots\\
    \text{ALG}(S^{T}_{\text{tar}}) + \frac{1}{4} \cdot \text{OPT}_{\text{LP}}(S^{T}_{\text{fut}}) &\geq \frac{1}{4} \cdot \text{OPT}_{\text{LP}}(S^{T-1}_{\text{fut}}).
\end{align*}
\end{small}

\noindent By summing up all the above inequalities, we have
\begin{align*}
    \sum\limits_{t=1}^{T} \text{ALG}(S^{t}_{\text{tar}}) &\geq \frac{\text{OPT}_{\text{LP}}(S^1_{\text{cur}})}{4} - \frac{\text{OPT}_{\text{LP}}(S^{T}_{\text{fut}})}{4}\\
    &= \frac{\text{OPT}}{4} - 0 = \frac{\text{OPT}}{4},
\end{align*}
\noindent which proves the approximation ratio.

Similar to that in \algor, a slight modification on \algod\ is to select only $(c,s,{x^\ast}^c_{u,s})$ that leads to a nonempty target group, so that $T = O(nk)$. We first prove that there does exist such $(c,s,{x^\ast}^c_{u,s})$ that $f(c,s,{x^\ast}^c_{u,s}) \geq \frac{1}{4} \cdot {\text{OPT}}_{\text{LP}}(S_{\text{cur}}(c,s,{x^\ast}^c_{u,s}))$. Consider those cases where $\alpha = {x^\ast}^c_{u,s}$ is too large that the target group is empty. In this case $S_{\text{cur}}(c,s,{x^\ast}^c_{u,s}) = S_{\text{fut}}(c,s,{x^\ast}^c_{u,s})$, which implies $f(c,s,{x^\ast}^c_{u,s}) = \frac{1}{4} \cdot {\text{OPT}}_{\text{LP}}(S_{\text{cur}}(c,s,{x^\ast}^c_{u,s}))$, i.e., exactly the expected value. Suppose that no other $(c,s,{x^\ast}^c_{u,s})$ exists so that $f(c,s,{x^\ast}^c_{u,s}) \geq \frac{1}{4} \cdot {\text{OPT}}_{\text{LP}}(S_{\text{cur}}(c,s,{x^\ast}^c_{u,s}))$, then the expected value should be lower than $\frac{1}{4} \cdot {\text{OPT}}_{\text{LP}}(S_{\text{cur}}(c,s,{x^\ast}^c_{u,s}))$, a contradiction. Thus, such $(c,s,{x^\ast}^c_{u,s})$ must exist, and \algod\ can always choose such $(c,s,{x^\ast}^c_{u,s})$.

The number of iterations is thus also $O(nk)$. For each iteration, there are $O(nmk)$ combinations of pivot parameters. Finding $S_{\text{tar}}(c,s,{x^\ast}^c_{u,s})$ and $S_{\text{fut}}$, as well as executing \stagethree, require $O(n)$ time. To find the optimal pivot parameters in DPS in each iteration, a simple approach is to examine all $O(nmk)$ combinations of pivot parameters. For each $(c,s,{x^\ast}^c_{u,s})$, finding $f(c,s,{x^\ast}^c_{u,s})$ requires $O(E)$ time. Thus, each iteration needs $O(nmk|E|)$ time. Since $T = O(nk)$, the total time complexity of \algod, including the \stageone, is therefore $O(\text{LP}) + O(n^2 m k^2 |E|)$. By using heaps to store the \weight s and reordering the computation, repeated calculations can be avoided to reduce the complexity to $O(\text{LP}) + O(n m k^2 |E|)$. Alternatively, effectively parallelizing the computation of $f(c,s,{x^\ast}^c_{u,s})$ can achieve at most a speedup factor of $nmk$, which further lowers the time complexity to $O(\text{LP}) + O(n m k |E|)$. \hfill \qedsymbol
}

\subsection{Enhancements} \label{subsec:enhancements}
\revise{In the following, we detail enhancements of the \algo\ and \algod\ algorithms. First, we show that \algo\ and \algod\ support values of $\lambda \neq \frac{1}{2}$ via a simple scaling on the inputs. We then design two advanced strategies, including an \textit{advanced LP transformation technique} and a \textit{new focal parameter sampling scheme}, to improve the efficiency of \algo\ and \algod. The LP transformation technique derives a new LP formulation to reduce the number of decision variables and constraints from $O((n+|E|)mk)$ to $O((n+|E|)m)$ by condensing the $x^c_{u,s}$ variables ($k$ is the number of slots). We prove that the optimal objective in the new formulation is exactly that in the original one. The focal parameter sampling scheme maintains a \textit{maximum \weight} $\bar{x^\ast}^c_s$ (detailed later) for each pair of item $c$ and slot $s$ to avoid unnecessary sampling of focal parameters $(c,s,\alpha)$ when $\alpha \geq \max\limits_{u \in V} {x^\ast}^c_{u,s}$, especially for a large $k$. We prove that the sampling results of the new sampling technique and the original one are the same. Therefore, the efficiency of \algo\ can be improved without sacrificing the solution quality. Finally, we extend \algo\ and \algod\ to support \probtwo\ (and also the Social Event Organization (SEO)-type problems) by tailoring \stagethree\ with consideration of the additional VR-related constraints. \opt{short}{Due to the space constraint, we present the details in Section \ref{subsec:enhancements} in the full version \cite{Online}.}}

\opt{full}{
\para{Supporting Other Values of $\lambda$.}
First, observe that $\lambda = 0$ corresponds to a special case of maximizing only the total preference utility, where a simple greedy algorithm can find the exact optimal solution. Assume that $\lambda \neq 0$. To support the cases with $\lambda \neq \frac{1}{2}$, \algo\ first sets the \textit{scaled preference values} for each $(u,c)$ as $\text{p'}(u,c) = \frac{1-\lambda}{\lambda}\text{p}(u,c)$. This value is then used instead of $\text{p}(u,c)$ in the LP relaxation problem, as well as in the computation of $f(c,s,{x^\ast}^c_{u,s})$ in DPS in \algod. 
Since
\begin{align*}
    &(1-\lambda)\cdot\text{p}(u,c) + \lambda\cdot\sum\limits_{v | u \xleftrightarrow{c} v} \tau(u,v,c)\\
    =& 2 \lambda \cdot (\frac{1}{2}\cdot\text{p'}(u,c) + \frac{1}{2}\cdot\sum\limits_{v | u \xleftrightarrow{c} v} \tau(u,v,c)),
\end{align*}
each instance with $\lambda \neq \frac{1}{2}$ can be transformed to an instance with $\lambda = \frac{1}{2}$ by this scaling of input parameters, where \algo\ and \algod\ can achieve the approximation ratio.

\para{Advanced LP Transformation.} Recall that the LP relaxation of the \prob\ problem (denote by $\mathsf{LP}_{\mathsf{SVGIC}}$) is 

$$ \text{max} \sum\limits_{u \in V} \sum\limits_{c \in \mathcal{C}} [ (1-\lambda) \cdot \text{p}(u, c) \cdot x^c_u + \lambda \cdot \sum\limits_{e = (u,v) \in E} ( \tau(u,v,c) \cdot y^c_e) ] $$

\noindent subject to the following constraints:
\begin{small}
\begin{align*}
    \sum_{s=1}^{k} x^c_{u,s} \leq 1, \quad & \forall u \in V, c \in \mathcal{C} \\
    \sum_{c \in \mathcal{C}} x^c_{u,s} = 1, \quad & \forall u \in V, s \in [k] \\
    x^c_u = \sum_{s=1}^{k} x^c_{u,s}, \quad & \forall u \in V, c \in \mathcal{C} \\
    y^c_{e} = \sum_{s=1}^{k} y^c_{e,s}, \quad & \forall e = (u,v) \in E, c \in \mathcal{C} \\
    y^c_{e,s} \leq x^c_{u,s}, \quad & \forall e = (u,v) \in E, s \in [k], c \in \mathcal{C} \\
    y^c_{e,s} \leq x^c_{v,s}, \quad & \forall e = (u,v) \in E, s \in [k], c \in \mathcal{C} \\
    x^c_{u,s}, x^c_u, y^c_{e,s}, y^c_e \geq 0, \quad & \forall u \in V, e \in E, s \in [k], c \in \mathcal{C}. 
\end{align*}
\end{small}

Solving $\mathsf{LP}_{\mathsf{SVGIC}}$, which is the first step of the \algo\ algorithm, deals with $O((n+|E|)mk)$ decision variables and $O((n+|E|)mk)$ constraints. In the following, we transform $\mathsf{LP}_{\mathsf{SVGIC}}$ into a more compact linear program formulation $\mathsf{LP}_{\mathsf{SIMP}}$ with fewer decision variables as follows.

$$ \text{max} \sum\limits_{u \in V} \sum\limits_{c \in \mathcal{C}} [ (1-\lambda) \cdot \text{p}(u, c) \cdot \mathtt{x}^c_u + \lambda \cdot \sum\limits_{e = (u,v) \in E} ( \tau(u,v,c) \cdot \mathtt{y}^c_e) ] $$

\noindent subject to the following constraints:
\begin{small}
\begin{align*}
    \mathtt{x}^c_{u} \leq 1, \quad & \forall u \in V, c \in \mathcal{C} \\
    \sum_{c \in \mathcal{C}} \mathtt{x}^c_{u} = k, \quad & \forall u \in V \\
    \mathtt{y}^c_{e} \leq \mathtt{x}^c_{u}, \quad & \forall e = (u,v) \in E, c \in \mathcal{C} \\
    \mathtt{y}^c_{e} \leq \mathtt{x}^c_{v}, \quad & \forall e = (u,v) \in E, c \in \mathcal{C} \\
    \mathtt{x}^c_{u}, \mathtt{y}^c_{e} \geq 0, \quad & \forall u \in V, e \in E, c \in \mathcal{C}.
\end{align*}
\end{small}

To avoid confusion, we use normal math symbols ($x$ and $y$) to represent decision variables in $\mathsf{LP}_{\mathsf{SVGIC}}$ and Typewriter symbols ($\mathtt{x}$ and $\mathtt{y}$) to denote those in $\mathsf{LP}_{\mathsf{SIMP}}$. Let $\text{OPT}_{\mathsf{SIMP}}$ and $\text{OPT}_{\mathsf{SVGIC}}$ denote the optimal objective values in $\mathsf{LP}_{\mathsf{SIMP}}$ and $\mathsf{LP}_{\mathsf{SVGIC}}$. We have the following observation. 

\begin{observation} \label{scaledLP}
$\text{OPT}_{\mathsf{SIMP}} = \text{OPT}_{\mathsf{SVGIC}}$ always holds. Moreover, given an optimal solution $\mathtt{X}^{\ast} = \{ \mathtt{x}^c_u \}$ for $\mathsf{LP}_{\mathsf{SIMP}}$, there exists an optimal solution ${X^\ast}$ of $\mathsf{LP}_{\mathsf{SVGIC}}$ with ${{x^\ast}}^c_{u,s} = \frac{1}{k} \cdot {\mathtt{x}^{\ast}}^c_u, \forall c, u, s$. 
\end{observation}

\begin{proof}
We first prove $\text{OPT}_{\mathsf{SIMP}} \leq \text{OPT}_{\mathsf{SVGIC}}$. Given any feasible solution $\mathtt{X}$ of $\mathsf{LP}_{\mathsf{SIMP}}$, let $X$ denote a solution of $\mathsf{LP}_{\mathsf{SVGIC}}$ such that ${x}^c_{u,s} = \frac{1}{k} \cdot \mathtt{x}^c_u, \forall c, u, s$. By construction, we have $x^c_u = \sum_{s=1}^{k} x^c_{u,s} = k \cdot \frac{1}{k} \cdot \mathtt{x}^c_u = \mathtt{x}^c_u$ for all $c,u,s$, and analogously $y^c_e = \mathtt{y}^c_{e}$ for all $c,e$. Therefore, $X$ is feasible in $\mathsf{LP}_{\mathsf{SVGIC}}$. Moreover, it achieves the same objective value in $\mathsf{LP}_{\mathsf{SVGIC}}$ as $\mathtt{X}$ in $\mathsf{LP}_{\mathsf{SIMP}}$, which implies $\text{OPT}_{\mathsf{SIMP}} \leq \text{OPT}_{\mathsf{SVGIC}}$. We then prove $\text{OPT}_{\mathsf{SIMP}} \geq \text{OPT}_{\mathsf{SVGIC}}$. Given any feasible solution $X$ of $\mathsf{LP}_{\mathsf{SVGIC}}$, let $\mathtt{X}$ denote a solution of $\mathsf{LP}_{\mathsf{SIMP}}$ such that $\mathtt{x}^c_u = x^c_u = \sum_{s=1}^{k} x^c_{u,s}$. Similarly, $\mathtt{X}$ is a feasible solution of $\mathsf{LP}_{\mathsf{SIMP}}$, and it achieves exactly the same objective value in $\mathsf{LP}_{\mathsf{SIMP}}$ as $X$ in $\mathsf{LP}_{\mathsf{SVGIC}}$. Therefore, $\text{OPT}_{\mathsf{SIMP}} \geq \text{OPT}_{\mathsf{SVGIC}}$. Combining the above, we have $\text{OPT}_{\mathsf{SIMP}} = \text{OPT}_{\mathsf{SVGIC}}$, and the second part of the observation follows from the above construction.
\end{proof}

As \algo\ requires only \textit{one} optimal solution for $\mathsf{LP}_{\mathsf{SVGIC}}$, it suffices for \algo\ to first solve the above simplified linear program $\mathsf{LP}_{\mathsf{SIMP}}$ with only $O((n+|E|)m)$ decision variables and constraints. \algo\ then scales the optimal solution by a factor of $\frac{1}{k}$ according to Observation \ref{scaledLP} to find an optimal solution for $\mathsf{LP}_{\mathsf{SVGIC}}$. Therefore, $\mathsf{LP}_{\mathsf{SIMP}}$ effectively improves the efficiency of \algo. 

\para{Advanced Focal Parameter Sampling.} In \algo, a set of focal parameters $(c,s,\alpha)$ is sampled uniformly at random for each iteration. However, for input instances with a large $k$, the utility factor ${x^\ast}^c_{u,s}$ of assigning item $c$ to user $u$ at slot $s$ tends to become smaller, e.g., roughly $O(\frac{1}{k})$, especially after the above LP transformation strategy is applied. Recall that \algo\ compares $\alpha$ to every user eligible for viewing item $c$ at slot $s$ and displays $c$ to user $u$ if ${x^\ast}^c_{u,s} \geq \alpha$. If $\alpha > \max\limits_{u \in V} {x^\ast}^c_{u,s}$, \algo\ essentially remains idle (i.e., making no progress) in the current round and needs to re-sample another set of focal parameters in the next iteration. For instances with large $k$, it is thus more difficult for \algo\ to effectively sample a set of \textit{good} focal parameters without the above property.

To address the above issue, we design an advanced focal parameter sampling scheme as follows. For each pair of $(c,s)$, \algo\ maintains the \textit{maximum \weight} $\bar{x^\ast}^c_s = \max\limits_{u \in V} \{ {x^\ast}^c_{u,s} | u \, \text{is eligible for} \, (c,s)\}$. Specifically, $\bar{x^\ast}^c_s$ is initialized as $\max\limits_{u \in V} {x^\ast}^c_{u,s}$ after $X^\ast$ (the optimal solution of LP) is retrieved, and it is constantly updated along the rounding process (as the number of users eligible for $(c,s)$ decreases). In each iteration, \algo\ first randomly samples a pair $(c,s)$ with probability $\frac{\bar{x^\ast}^c_s}{\sum\limits_{c \in \mathcal{C}, s \in [k]}\bar{x^\ast}^c_s}$, i.e., proportional to the value $\bar{x^\ast}^c_s$. It then samples $\alpha$ uniformly at random from $[0,\bar{x^\ast}^c_s]$. 

Specifically, let $\mathsf{GOOD}$ denote the event that a set of good focal parameters is sampled, i.e., $\alpha$ is at most $\bar{x^\ast}^c_s$ for the selected $(c,s)$ in the original sampling scheme. Moreover, for a set of good focal parameter $(c,s,\alpha)$, let $\mathsf{TARGET}(c,s,\alpha)$ denote the event that \stagethreefull\ ends up with the same behavior as if $(c,s,\alpha)$ is sampled. We have the following observation. Let $\text{Pr}_{\text{orig}}(\mathsf{TARGET}(c,s,\alpha))$ denote the probability that $\mathsf{TARGET}(c,s,\alpha)$ happens in the original sample scheme, and let $\text{Pr}_{\text{adv}}(\mathsf{TARGET}(c,s,\alpha))$ denote the probability of $\mathsf{TARGET}(c,s,\alpha)$ in the advanced sample scheme. The following observation indicates that sampling with the advanced sampling scheme is equivalent (in terms of outcome distribution) to sampling \textit{good} parameters, i.e., picking $(c,s,\alpha)$ directly and only from the pool of focal parameters that lead to a successful assignment of at least one item to one user at some slot.

\begin{observation}
For any set of good focal parameter $(c,s,\alpha)$ such that $\alpha \leq \max\limits_{u \in V} {x^\ast}^c_{u,s}$, $\text{Pr}_{\text{adv}}(\mathsf{TARGET}(c,s,\alpha)) = \text{Pr}_{\text{orig}}(\mathsf{TARGET}(c,s,\alpha) | \mathsf{GOOD})$.
\end{observation}
\begin{proof}
By construction, $\text{Pr}_{\text{orig}}(\mathsf{GOOD}) = \sum\limits_{c \in \mathcal{C}, s \in [k]} \frac{1}{k \cdot m}(\bar{x^\ast}^c_s)$. Next, by construction, it is straightforward that the advanced sampling scheme only samples good focal parameters. Consider a set of good focal parameter $(c,s,\alpha)$. Let $u_1$ and $u_2$ denote the users such that $u_1 = \argmax\limits_{{x^\ast}^c_{u,s} < \alpha} {x^\ast}^c_{u,s}$ and $u_2 = \argmin\limits_{{x^\ast}^c_{u,s} \geq \alpha} {x^\ast}^c_{u,s}$. By definition, if \stagethree\ chooses any $\alpha \in ({x^\ast}^c_{{u_1},s}, {x^\ast}^c_{{u_2},s}]$, it ends up with the same behavior as if $(c,s,\alpha)$ is sampled, i.e., $\mathsf{TARGET}$ happens. With $(c,s)$ given, this happens with probability $p_\alpha = {x^\ast}^c_{{u_2},s} - {x^\ast}^c_{{u_1},s}$. Note that $u_2$ is guaranteed to exist, and we replace ${x^\ast}^c_{{u_1},s}$ with 0 if $u_1$ does not exist (i.e., the case where every eligible user is assigned $c$ by \stagethreefull). In the advanced sampling scheme, the probability of this event is thus $\text{Pr}_{\text{adv}}(\mathsf{TARGET}(c,s,\alpha)) = \frac{\bar{x^\ast}^c_s}{\sum\limits_{c \in \mathcal{C}, s \in [k]}\bar{x^\ast}^c_s} \cdot p_\alpha$. On the other hand, in the original sampling scheme, the probability of this event is $\text{Pr}_{\text{orig}}(\mathsf{TARGET}(c,s,\alpha)) = \frac{1}{k \cdot m} \cdot p_\alpha$. We have
\begin{align*}
    \text{Pr}_{\text{orig}}(\mathsf{TARGET}(c,s,\alpha) | \mathsf{GOOD}) &= \frac{\text{Pr}_{\text{orig}}(\mathsf{TARGET}(c,s,\alpha))}{\text{Pr}_{\text{orig}}(\mathsf{GOOD})}\\
    &= \frac{\frac{1}{k \cdot m} \cdot p_\alpha}{\sum\limits_{c \in \mathcal{C}, s \in [k]} \frac{1}{k \cdot m}(\bar{x^\ast}^c_s)}\\
    &= \frac{p_\alpha}{\sum\limits_{c \in \mathcal{C}, s \in [k]}\bar{x^\ast}^c_s}\\
    &= \text{Pr}_{\text{adv}}(\mathsf{TARGET}(c,s,\alpha)).
\end{align*}
The observation follows.
\end{proof}

Therefore, the advanced focal sampling scheme effectively improves the efficiency of \algo\ since it discards the non-contributing sampling events in advance. Note that it can also be applied to improve \algod\ such that the events of selecting non-controlling focal parameters will not be simulated in the derandomization process. 
The main algorithm \algo\ with the advanced LP transformation strategy and focal sampling scheme are detailed in Algorithm \ref{alg:modified_main_algo}. The effects of these two speedup techniques are tested in Section \ref{subsec:large_exp_scalability}.

\begin{algorithm}[h]  
    \caption{\algo\ with advanced LP transformation and focal parameter sampling}
    \label{alg:modified_main_algo}
    \begin{algorithmic}[1]
        \Require $X^\ast$
        \Ensure An \configuration\ $\mathbf{A}$
        \State $\mathbf{A}(\hat{u},\hat{s}) \gets \text{NULL}$ for all $\hat{u},\hat{s}$
        \State Construct $\mathsf{LP}_{\mathsf{SVGIC}}$ (Original Relaxed LP)
        \State Transform into $\mathsf{LP}_{\mathsf{SIMP}}$ (Simplified LP)
        \State $\bar{X} \gets$ optimal solution for $\mathsf{LP}_{\mathsf{SIMP}}$
        \State $\bar{x^\ast}^c_s \gets \max\limits_{u \in V} {x^\ast}^c_{u,s}$
        \For{$u \in V, c \in \mathcal{C}, s \in [k]$}
            \State ${x^\ast}^u_{c,s} \gets \frac{1}{k} \cdot \bar{x}^c_u$
        \EndFor
        \While {some entry in $\mathbf{A}$ is $\text{NULL}$} 
            \State Sample $(c,s)$ randomly with probability $\frac{\bar{x^\ast}^c_s}{\sum\limits_{c \in \mathcal{C}, s \in [k]}\bar{x^\ast}^c_s}$
            \State Sample $\alpha \in [0,\bar{x^\ast}^c_s]$ uniformly at random
            \For{$\hat{u} \in V$}
                \If{$\mathbf{A}(\hat{u},s) = \text{NULL}$ and $\mathbf{A}(\hat{u},t) \neq c \, \forall t \neq s$} \\ \hfill \Comment{$\hat{u}$ eligible for $(c,s)$}
                    \If{${x^\ast}^c_{\hat{u},s} \geq \alpha$}
                        \State $\mathbf{A}(\hat{u},s) \gets c$
                    \EndIf
                \EndIf
            \EndFor
            \State Update $\bar{x^\ast}^c_s$
        \EndWhile
        \Return $\mathbf{A}$
    \end{algorithmic}
\end{algorithm}
\normalsize

\para{Extending \algo\ for \probtwo.} To modify \algo\ (and consequently, \algod) to support the teleportation and subgroup size constraints in VR, the LP relaxation problem is first replaced by the corresponding formulation in Section \ref{subsec:hardness_ip}, i.e., with Constraints (\ref{ilp2:indirect_relation_1}), (\ref{ilp2:indirect_relation_2}) and (\ref{ilp2:integrality}). Next, in each iteration of \stagethree, instead of displaying the \roundingitem\ $c$ to all eligible users $\hat{u}$ with ${x^\ast}^c_{u,s} \geq \alpha$, \stagethree\ first checks the number of users already displayed $c$ at the \roundingslot\ $s$. It then iteratively adds \textit{the eligible user with the highest \weight\ of $c$ at slot $s$} until the number of users displayed $c$ at slot $s$ reaches the size constraint $M$ or every user is examined. If the size constraint is reached, \stagethree\ \textit{locks} the item $c$ at slot $s$ by setting ${x^\ast}^c_{u,s} = 0$ for all remaining eligible users, as well as removes $(c,s,\alpha)$ from the set of candidate parameter sets in \algod. Therefore, \algo\ never displays more than $M$ users to the same item in the final \configuration.

\para{Supporting Social Event Organization.} We have also identified \textit{Social Event Organization} (SEO) as another important application of the targeted problem. Specifically, SEO~\cite{KL14KDD,SEO15SIGMOD,SEO16TKDE,SEO17ICDE,XW18WWW} is a line of research that studies the problem of organizing overlapping and conflicting events for users of Event-Based Social Network (EBSN) (e.g. Meetup \cite{Meetup}, Facebook Events \cite{fbevents}, and Douban Location \cite{dblocation}). Given a set of users and a set of events, SEO assigns each user to a single or a series of events, such that the total preferences of the selected users to their assigned events are maximized, while several constraints are satisfied, e.g. traveling cost budget, time conflict avoidance, and event sizes. Consequently, the optimization problems for SEO are closely related to \prob: social event attendees directly correspond to VR shopping users, while social events correspond to the displayed items in VR. \prob\ can be viewed as a general version of the SEO problem that incorporates both social-based utilities to support a series of social events. It is then straightforward to see the above extension for \probtwo\ supports SEO, as the social event size constraints can be modelled as subgroup size constraints.
}

\section{Extensions for Practical Scenarios}\label{sec:extension}

\revise{In this section, we extend \prob\ and \algo\ to support a series of practical scenarios. 1) \textit{Commodity values}. Each item is associated with a commodity value to maximize the total profit.  2) \textit{Layout slot significance}. Each slot location is associated with a different significance weight (e.g., center is better) according to retailing research \cite{slot, shelf}. 3) \textit{Multi-View Display}, where a user can be displayed multiple items in a slot, including one default, personally preferred item in the primary view and multiple items to view with friends in group views, and the primary and group views can be freely switched. 4) \textit{Generalized social benefits}, where social utility can be measured not only pairwise (each pair of friends) but also group-wise (any group of friends). 5) \textit{Subgroup change}, where the fluctuations (i.e., change of members) between the partitioned subgroups at consecutive slots are limited to ensure smooth social interactions as the elapse of time. 6) \textit{Dynamic scenario}, where users dynamically join and leave the system with different moving speeds. \opt{short}{Due to the space constraint, we present the details in Section \ref{sec:extension} in the full version of the paper \cite{Online}.}}

\opt{full}{
\para{A. Commodity Value.} The current objective function, i.e., maximizing the total utility among all users, is not designed from a retailer's perspective. In contrast, it will be more profitable to maximize the \textit{total expected profit} by taking into account the commodity values, while regarding the final \scenario\ utility as the willingness/likelihood of purchases. Therefore, the ILP formulation of \prob\ and \algo\ can be modified to examine the total \scenario\ utility weighted by a commodity value $\omega_c$ for each item $c$. Let $\omega_c$ denote the commodity price of each item $c$. The new objective of \prob\ is the weighted total \scenario\ utility as follows.

$$ \sum\limits_{u \in V} \sum\limits_{s=1}^{k} \omega_{\mathbf{A}^{\ast}(u,s)} \cdot [\text{p}(u, \mathbf{A}^{\ast}(u,s)) + \sum\limits_{v | u \xleftrightarrow{\mathbf{A}^{\ast}(u,s)} v} \tau(u,v,\mathbf{A}^{\ast}(u,s))]. $$

\para{B. Layout Slot Significance.} 
Retailing research \cite{slot, shelf} manifests that different slot locations in the item placement layout have varying significance for users. For example, the slots at the center of an aisle are nine times more important than the slots at the ends \cite{slot}. Vertical slots right at the eye-level of users are also more significant \cite{shelf}. Therefore, it is important for \prob\ to not only select the items but also optimize the corresponding placements on the shelf. Specifically, let $\gamma_s$ represent the relative significance of slot $s$ in the layout. The new objective of \prob\ is the \scenario\ utility, weighted by the relative significance as follows.

$$ \sum\limits_{u \in V} \gamma_s \cdot \sum\limits_{s=1}^{k} [\text{p}(u, \mathbf{A}^{\ast}(u,s)) + \sum\limits_{v | u \xleftrightarrow{\mathbf{A}^{\ast}(u,s)} v} \tau(u,v,\mathbf{A}^{\ast}(u,s))]. $$


\para{C. Multi-View Display} (MVD). \configuration\ with CID allows each user to see only \textit{one} item at each slot. In contrast, MVD \cite{RL14,LH18} has been proposed to support the display of multiple items at the same slot, so that the primary view (for a personally preferred item) and the group view (for common items with friends) can be freely switched. 
Let $\beta$ denote the maximum number of items (i.e., views) displayed for each user in each slot (depending on the user interface). 
an \textit{MVD-supportive} \configuration\ $\mathbb{A}$ maps a tuple $(u,s)$ to an itemset $\mathbb{A}(u,s)$ with at most $\beta$ items, and the new objective function of \prob\ is as follows.

$$ \sum\limits_{u \in V} \sum\limits_{s=1}^{k} \big( \sum\limits_{c \in \mathbb{A}^{\ast}(u,s)} [\text{p}(u,c) + \sum\limits_{v | u \xleftrightarrow{c} v} \tau(u,v,c)] $$ 

\para{D. Generalized Social Benefits.} In \prob\, $\tau(u,v,c)$ considers only the \textit{pair-wise} social interaction, i.e., social discussion and influence between two users viewing the common item. However, existing research on social influence \cite{HJH16KDD, HyperInfluence17, HyperInfluence19} manifests that quantifying the magnitude of social influence from \textit{group-wise} interactions is more general since the pair-wise interaction is a special case of the group-wise interaction. Therefore, let $\tau(u,\mathcal{V},c)$ denote the subgroup social utility for user $u$ when co-displaying $c$ to \textit{a subgroup of friends} $\mathcal{V}$ learned from group-wise social interaction models such as \cite{HJH16KDD}. Let $u \xleftrightarrow[s]{c} \mathcal{V}$ represent the \textit{maximal group} co-display\footnote{Here only the maximal groups, instead of every subgroup, are examined in the objective function; otherwise, duplicated social utilities will be summed up because multiple overlapping subgroups are contained in a maximal groups}, i.e., user $u$ sees item $c$ at slot $s$ with every user in $\mathcal{V}$, and \textit{no friend} of $u$ outside $\mathcal{V}$ also sees $c$ at slot $s$. Thus, the new objective of \prob\ is generalized as follows.

$$ \sum\limits_{u \in V} \sum\limits_{s=1}^{k} [\text{p}(u, \mathbf{A}^{\ast}(u,s)) + \sum\limits_{\mathcal{V} | u \xleftrightarrow[s]{c} \mathcal{V}} \tau(u,\mathcal{V},\mathbf{A}^{\ast}(u,s))].$$


\para{Algorithm Extension for Scenarios A--D.} 
\noindent In the following, we first introduce a new ILP formulation to support the above practical scenarios. Decision variable $x^c_{u,s}$ now indicates whether user $u$ views item $c$ \textit{in the primary view} at slot $s$. On the other hand, let decision variable $w^c_{u,s}$ denote whether $u$ can see $c$ \textit{in some view}, i.e., in the primary view or a group view, at slot $s$. Therefore, $x^c_{u,s} = 1$ implies $w^c_{u,s} = 1$. To consider the layout slot significance, let $w^c_u = \sum\limits_{s=1}^{k} \gamma_s \cdot w^c_{u,s}$ be the weighted sum of all $w$ decision variables, and here $w^c_u$ is not required to be binary. Therefore, $w^c_u$ represents the total significance of the assigned slots for $u$ to view $c$. To capture the maximal co-display, for all subgroups of users $\mathcal{V} \subseteq V$, let $y^c_{u,\mathcal{V},s}$ now indicate the \textit{maximal group} co-display. In other words, $u \xleftrightarrow[s]{c} \mathcal{V}$ if and only if $y^c_{u,\mathcal{V},s} = 1$. Analogously, let $y^c_{u,\mathcal{V}}$ be the weighted sum of $y^c_{u,\mathcal{V},s}$. The new objective function is as follows.


\small
\begin{align*}
\text{max} \quad \sum\limits_{c \in \mathcal{C}} \omega_c \cdot [\sum\limits_{u \in V} \big( \text{p}(u, c) \cdot w^c_u + \sum\limits_{\mathcal{V} \subseteq V} ( \tau(u,\mathcal{V},c) \cdot y^c_{u,\mathcal{V}}) \big)] 
\end{align*}
\normalsize

subject to the constraints,
\footnotesize
\begin{align}
    \sum_{c \in \mathcal{C}} x^c_{u,s} = 1, \quad & \forall u \in V, s \in [k] \label{ilp_ext:one_slot_one_default_item}\\
    \sum_{c \in \mathcal{C}} w^c_{u,s} \leq \beta, \quad & \forall u \in V, s \in [k] \label{ilp_ext:one_slot_many_item}\\
    x^c_{u,s} \leq w^c_{u,s}, \quad & \forall u \in V, s \in [k], c \in \mathcal{C} \label{ilp_ext:default_item_can_see} \\
    \sum_{s=1}^{k} x^c_{u,s} \leq 1, \quad & \forall u \in V, c \in \mathcal{C} \label{ilp_ext:no_replicated_item}\\
    x^c_u = \sum_{s=1}^{k} (\gamma_s \cdot x^c_{u,s}), \quad & \forall u \in V, c \in \mathcal{C} \label{ilp_ext:user_recommended}\\
    y^c_{\mathcal{V}} = \sum_{s=1}^{k} (\gamma_p \cdot y^c_{\mathcal{V},s}), \quad & \forall \mathcal{V} \subseteq V, c \in \mathcal{C} \label{ilp_ext:group_recommended}\\
    y^c_{\mathcal{V},s} \leq x^c_{u,s}, \quad & \forall u \in \mathcal{V}, s \in [k], c \in \mathcal{C} \label{ilp_ext:direct_relation_1} \\
    y^c_{\mathcal{V},s} + x^c_{u,s} \leq 1, \quad & \forall u \notin \mathcal{V}, s \in [k], c \in \mathcal{C} \label{ilp_ext:direct_relation_2} \\
    x^c_{u,s}, y^c_{\mathcal{V},s} \in \{ 0,1 \}, \quad & \forall u \in V, e \in E, s \in [k], c \in \mathcal{C} \label{ilp_ext:integrality}.
\end{align}

\normalsize
Constraint (\ref{ilp_ext:one_slot_one_default_item}) guarantees that a user is displayed exactly one default item (for the primary view) at each slot. Constraint (\ref{ilp_ext:one_slot_many_item}) ensures that at most $\beta$ items can be viewed at one slot in MVD. Constraint (\ref{ilp_ext:default_item_can_see}) states that the default item can be viewed in the MVD, i.e., if $c$ is displayed ($x^c_{u,s} = 1$), then $w^c_{u,s}$ will become 1. Constraint (\ref{ilp_ext:no_replicated_item}) guarantees that the default displayed items for each user are not replicated. Constraint (\ref{ilp_ext:user_recommended}) ensures that $x^c_u$ (significance of the slot where $u$ is displayed $c$) is the weighted sum of $x^c_{u,s}$. Similarly, constraint (\ref{ilp_ext:group_recommended}) 
states that $y^c_{\mathcal{V}}$
is the weighted sum of $y^c_{\mathcal{V},s}$ (on the slot significance). 
Constraints (\ref{ilp_ext:direct_relation_1}) and (\ref{ilp_ext:direct_relation_2}) specify the maximal co-display. Constraint (\ref{ilp_ext:direct_relation_1}) states that $y^c_{\mathcal{V},s}=1$ only if every $x^c_{u,s}=1$, implying that $c$ is displayed to every $u$ in $\mathcal{V}$ at slot $s$. When $y^c_{\mathcal{V},s}=1$, constraint (\ref{ilp_ext:direct_relation_2}) ensures no other $x^c_{u,s}$ is 1 for $u$ outside of $\mathcal{V}$, so that $\mathcal{V}$ is indeed the maximal set. 
Finally, constraint (\ref{ilp_ext:integrality}) ensures the decision variables to be binary, i.e., $\in \{0,1\}$.

The LP relaxation of the above formulation also relaxes the last constraint, so that binary decision variables $x^c_{u,s}, y^c_{\mathcal{V},s}$ become real numbers in $[0,1]$. The other two phases of \algor\ and \algod\ remain unchanged. The commodity value $\omega_c$ and the layout slot significance $\gamma_s$ will not affect the theoretical guarantee of \algor\ and \algod\ because the original proof can be modified by replacing the preference utility and social utility values with weighted versions. Furthermore, in any iteration $t$ in \algo, for any subgroup of users $\mathcal{V}$ with all users $u \in \mathcal{V}$ eligible for $(c,s)$, as the maximal co-displayed subgroups contain at most $(\max|\mathcal{V}|)$ users, Lemma \ref{lemma:ratio} can be analogously modified to upper bound the probability that \textit{at least one} user in $\mathcal{V}$ is not eligible for $(c,s)$ in iteration $(t+1)$ by summing exactly $2 (\max|\mathcal{V}|)$ summation terms as in the last inequality in the proof of Lemma \ref{lemma:ratio}. Consequently, \algor\ and \algod\ can achieve a $2 (\max|\mathcal{V}|)$-approximation. Moreover, the advanced focal parameter sampling scheme can be employed in the above model. 

\para{E. Subgroup Changes.} While the partitions of subgroups in \algo\ enable flexible configurations and facilitates social discussion, drastic subgroup changes between adjacent slots may undermine the smoothness of social interactions and discussions within the group when users are exploring the VE. To address this issue, one approach is to incorporate a constraint on \textit{edit distances} between the partitioned induced subgroups at consecutive slots. Specifically, a pair of friends $u$ and $v$ in the same subgroup (i.e., viewing a common item) at slot $s$ but separated into different subgroups at slot $s+1$ will contribute $1$ to the edit distance between the two slots. To reduce the subgroup changes, IP of the extended \algo\ discourages the intermediate solution (the \fragsol) to introduce large subgroup changes by constraining the total edit distances across all consecutive slots in the \configuration. In the rounding scheme, \algo\ examines the potential outcomes before co-displaying a \roundingitem\ to the target subgroup. It voids the current iteration and re-sample another set of pivot parameters, if the solution incurs a large subgroup change. Finally, after the complete \configuration\ is constructed, \algo\ can facilitate a local search by exchanging the sub-configuration at different slots to reduce the total subgroup change.

\para{F. Dynamic Scenario.} To support the dynamic join and leave of users with different walking speeds, executing the whole \algo\ for every dynamic event is computationally intensive. Therefore, one approach is to first solve the relaxed LP for the partial \prob\ instance with the new setting, whereas the majority of \weight s remain the same (temporarily regarded as constants in LP), so that the \weight s are updated efficiently. \stagethree\ then assigns new users to existing target subgroups according to their \weight s, and it also invokes local search to examine the potential profit from exchanging users among existing target subgroups. Moreover, similar to in \probtwo, \textit{teleportation} in VR enables a user to move directly to her friends at a different slot. Therefore, when two users are co-displayed an item (directly or indirectly) but are far apart from each other in the store, the shopping application could suggest one of them to teleport to the other's location.
}

\section{Experiments} \label{sec:exp}
In this section, we evaluate the proposed \algor\ and \algod\ along with various baseline algorithms on three real datasets. \opt{short}{We also build a prototype of a VR store with Unity and SteamVR to conduct a user study.}\opt{full}{Moreover, a case study on the real dataset is provided to demonstrate the characteristics of solutions derived from different approaches. Finally, we build a prototype of a VR store with Unity and SteamVR to conduct a user study.}

\subsection{Experiment Setup and Evaluation Plan} \label{subsec:setup}
\para{Datasets.} To evaluate the proposed algorithms, three real datasets are tested in the experiment. The first dataset \textit{Timik} \cite{Timik} is a 3D VR social network containing 850k vertices (users) and 12M edges (friendships) with 12M check-in histories of 849k virtual Point-of-Interests (POIs). The second dataset, \textit{Epinions}\cite{Epinions}, is a website containing the reviews of 139K products and a social trust network with 132K users and 841K edges. The third dataset, \textit{Yelp} \cite{Yelp}, is a location-based social network (LBSN) containing 1.5M vertices, 10M edges, and 6M check-ins. For Timik and Yelp, we follow the settings in \cite{CYS18,AG19,CC17} to treat POIs in the above datasets as the candidate items in \prob. The preference utility and social utility values are learned by the PIERT learning framework \cite{YL18TOIS} which jointly models the social influence between users and the latent topics of items. Following the scales of the experiments in previous research \cite{SD17,SBR15}, the default number of slots $k$, number of items $m$, and size of user set $n$ selected from the social networks to visit a VR store are set as 50, 10000, and 125, respectively.

\para{Baseline Methods.} We compare \algor\ and \algod\ with five baseline algorithms: Personalized Top-$k$ (PER), Fairness Maximization in Group recommendation (FMG) \cite{SD17}, Social-aware Diverse and Preference selection (SDP) \cite{CYS18}, Group Recommendation and Formation (GRF) \cite{SBR15}, and Integer Programming (IP). PER and FMG correspond to the two baseline approaches outlined in Section \ref{sec:intro}. Specifically, PER retrieves the top-$k$ preferred items for each user (the \personalized\ approach), while FMG selects a bundled itemset for all users as a group (the \group\ approach) with considerations of fairness of preference among the users. SDP selects socially-tight subgroups to display their preferred items, which corresponds to the subgroup approach outlined in Section \ref{sec:algo}. GRF splits the input users into subgroups with similar item preferences without considering the social network topology, which can be viewed as a variation of the subgroup approach where the subgroups are partitioned based on preferences instead of social connections. Finally, IP is the integer program formulated in Section \ref{subsec:hardness_ip} that finds the optimal solutions of small \prob\ instances by Gurobi \cite{Gurobi}. All algorithms are implemented in an HP DL580 Gen 9 server with four 3.0 GHz Intel CPUs and 1TB RAM. Each result is averaged over 50 samples. \opt{short}{The datasets and algorithms used are released as a public download in \cite{OnlineRepo}.}

\para{Evaluation Metrics.} To evaluate the algorithms considered for \prob\ and analyze the returned \configuration s, we introduce the following metrics: 1) total \scenario\ utility achieved, 2) total execution time (in seconds), 3) the percentages of personal preference utility (\textit{Personal\%}) and social utility (\textit{Social\%}) in total \scenario\ utility, 4) the percentage of \textit{Inter}-subgroup edges (\textit{Inter\%}) and \textit{Intra}-subgroup edges (\textit{Intra\%}) in the returned partition of subgroups, 5) the average network density among partitioned subgroups, normalized by the average density of the original social network, 6) the percentage of friend pairs viewing common items together (\textit{Co-display\%}), 7) the percentage of users viewing items alone (\textit{Alone\%}), and 8) \textit{regret ratio} \revise{(a fairness measure detailed later in Section \ref{subsec:group_formation_exp}.)} \opt{full}{Moreover, for evaluating the performance of all methods in \probtwo, we measure 9) \textit{feasibility ratio}: the number of feasible solutions divided by the total number of instances, and 10) \textit{group size constraint violation}: the number of partitioned subgroups exceeding the predefined size constraint.}

\para{Evaluation Plan.} \revise{To evaluate the performance of the above algorithms, \opt{full}{in Section \ref{subsec:small_exp}, }we use IP to derive the optimal total \scenario\ utility on small datasets.\footnote{Since solving IP is $\mathsf{NP}$-hard, Gurobi cannot solve large instances within hours.} The social networks and items in the small datasets are respectively sampled by random walk and uniform sampling from Timik according to the setting of \cite{AN15VLDB}.\opt{short}{ Due to the space constraint, please see Section 6.2 in the full version \cite{Online} for the experimental results on small datasets.} Next, in Section \ref{subsec:large_exp_sensitivity}, we evaluate the efficacy of \algor\ and \algod\ in large datasets with input scales following previous research \cite{SD17,SBR15}, while conducting experiments on different inputs (the $\text{p}$ and $\tau$ values) generated by PIERT~\cite{YL18TOIS} (default), AGREE, and GREE~\cite{DC18}. We examine the efficiency of all algorithms, including various configurations of mixed integer programming (MIP) algorithms, in Section \ref{subsec:large_exp_scalability}. We then compare the algorithms on the aforementioned group formation-related performance metrics in Section \ref{subsec:group_formation_exp}. A case study on a 2-hop ego network in Yelp is shown in \opt{full}{Section \ref{subsec:case_study}}\opt{short}{Section 6.6 in \cite{Online}} to examine the subgroup partition patterns in different algorithms. Result of a sensitivity test on $r$, an important algorithmic parameter of the deterministic \algod\ algorithm and experimental results on the \probtwo\ problem are reported in \opt{short}{Sections 6.7 and 6.8 in \cite{Online}}\opt{full}{Sections \ref{subsec:diff_r} and \ref{subsec:subgroup_constraint}}, respectively. Finally, we build a prototype of VR store with Unity 2017.1.1.1 (64bit), Photon Unity Network free 1.86, SteamVR Plugin 1.2.2, VRTK, and 3ds Max 2016 for hTC VIVE HMD to validate the proposed objective in modelling \prob. We detail the user study in Section \ref{sec:user_study}.}


\opt{full}{
\begin{figure}[t]
	\centering
	\subfigure[][\centering Total \scenario\ utility \newline vs. size of user set ($n$).] {\
		\centering \includegraphics[width = 0.44 \columnwidth]{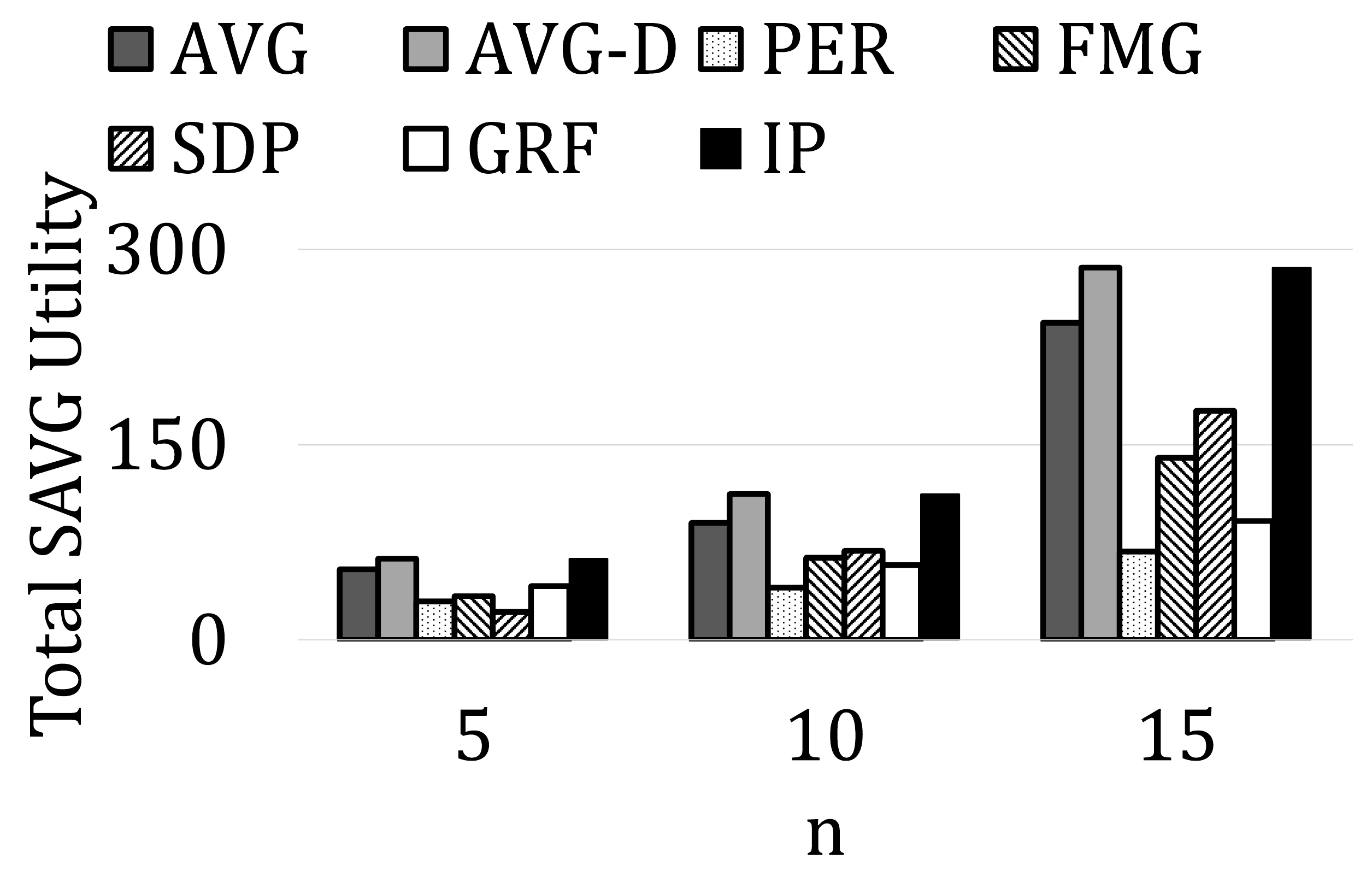}
		\label{fig:small_obj_n} } 
	\subfigure[][\centering Execution time vs. \newline size of user set ($n$).] {\
		\centering \includegraphics[width = 0.44 \columnwidth]{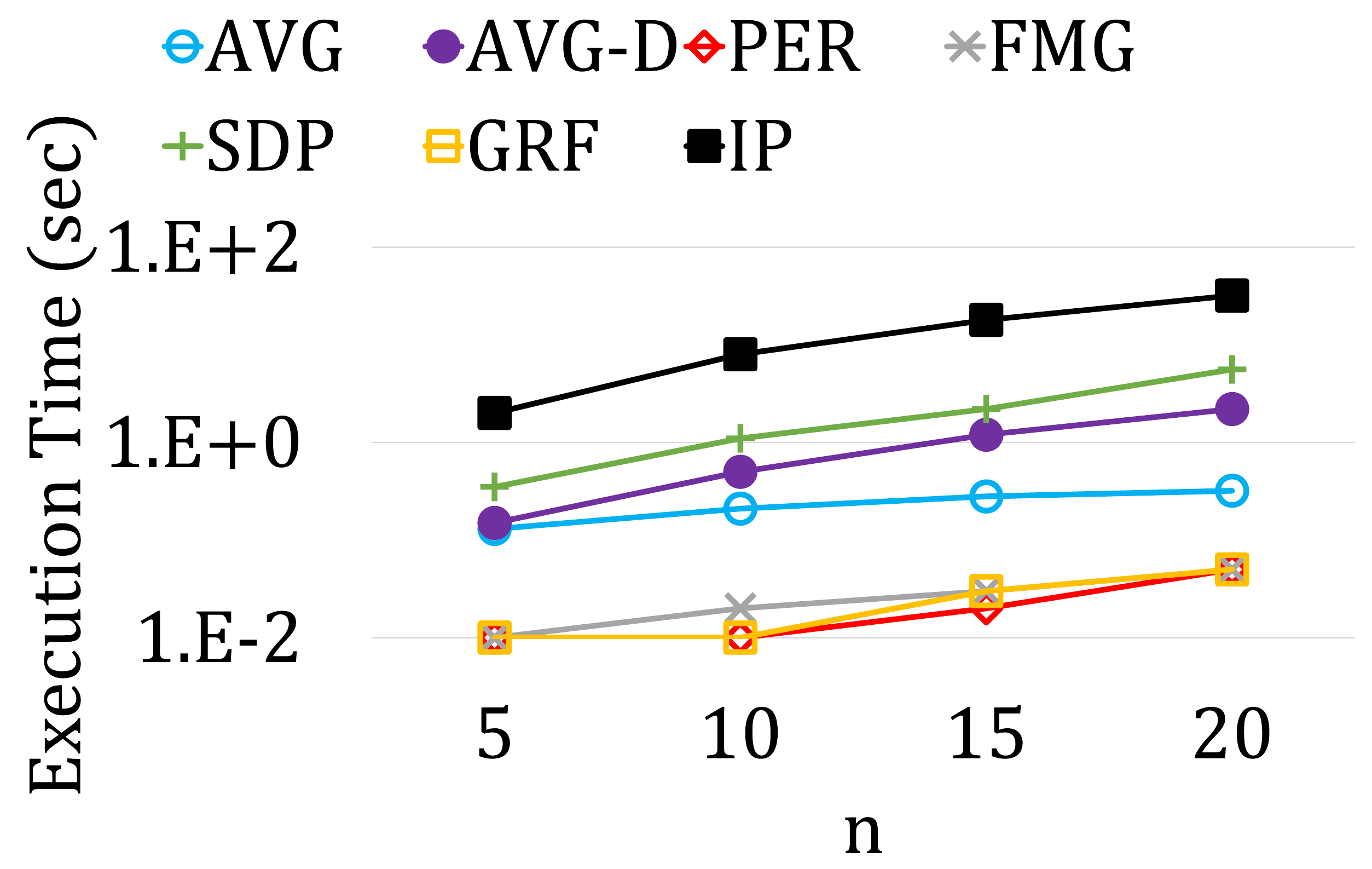}
		\label{fig:small_time_n} }
	\subfigure[][\centering Total \scenario\ utility \newline vs. size of item set ($m$).] {\
		\centering \includegraphics[width = 0.44 \columnwidth]{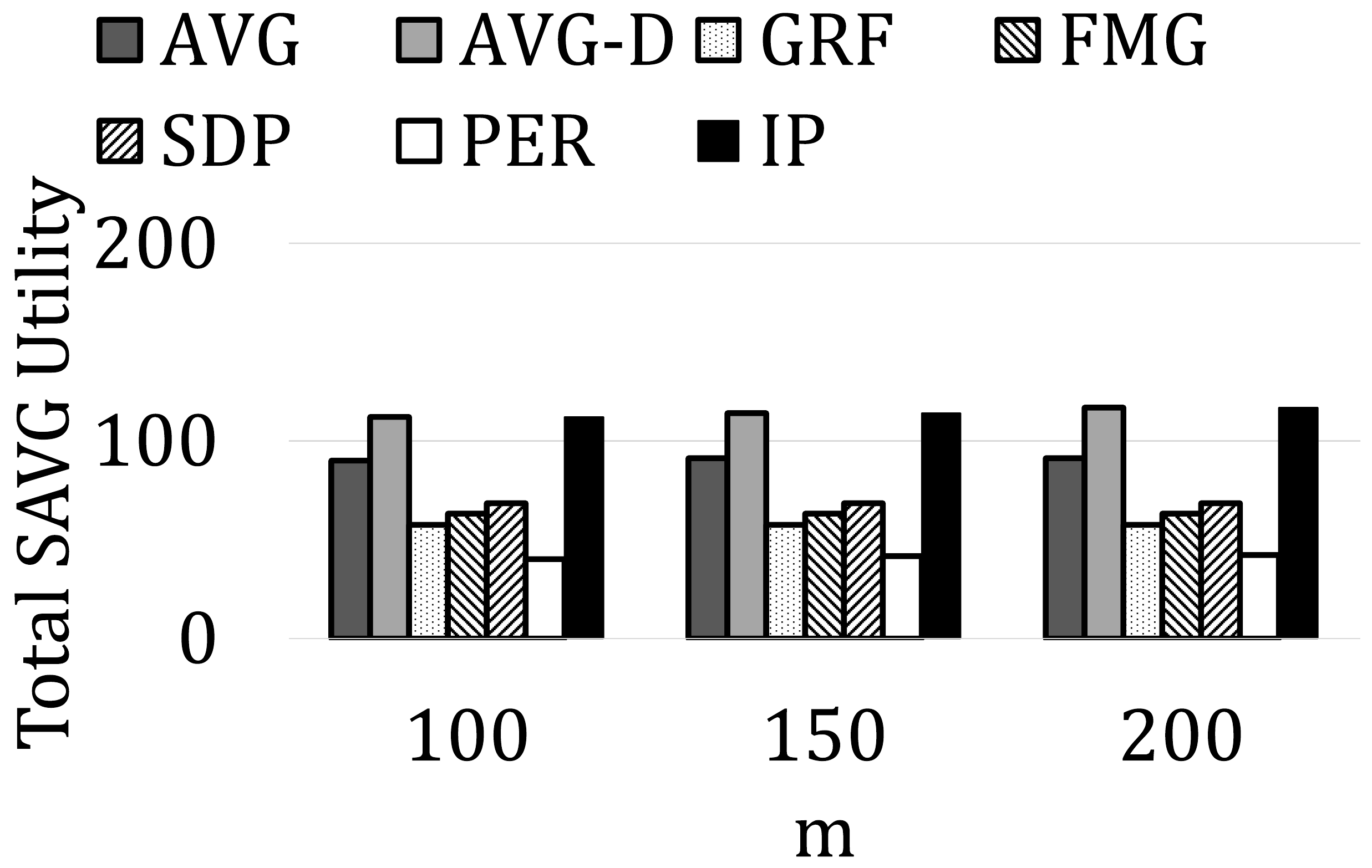}
		\label{fig:small_obj_m} } 
	\subfigure[][\centering Execution time vs. \newline size of item set ($m$).] {\
		\centering \includegraphics[width = 0.44 \columnwidth]{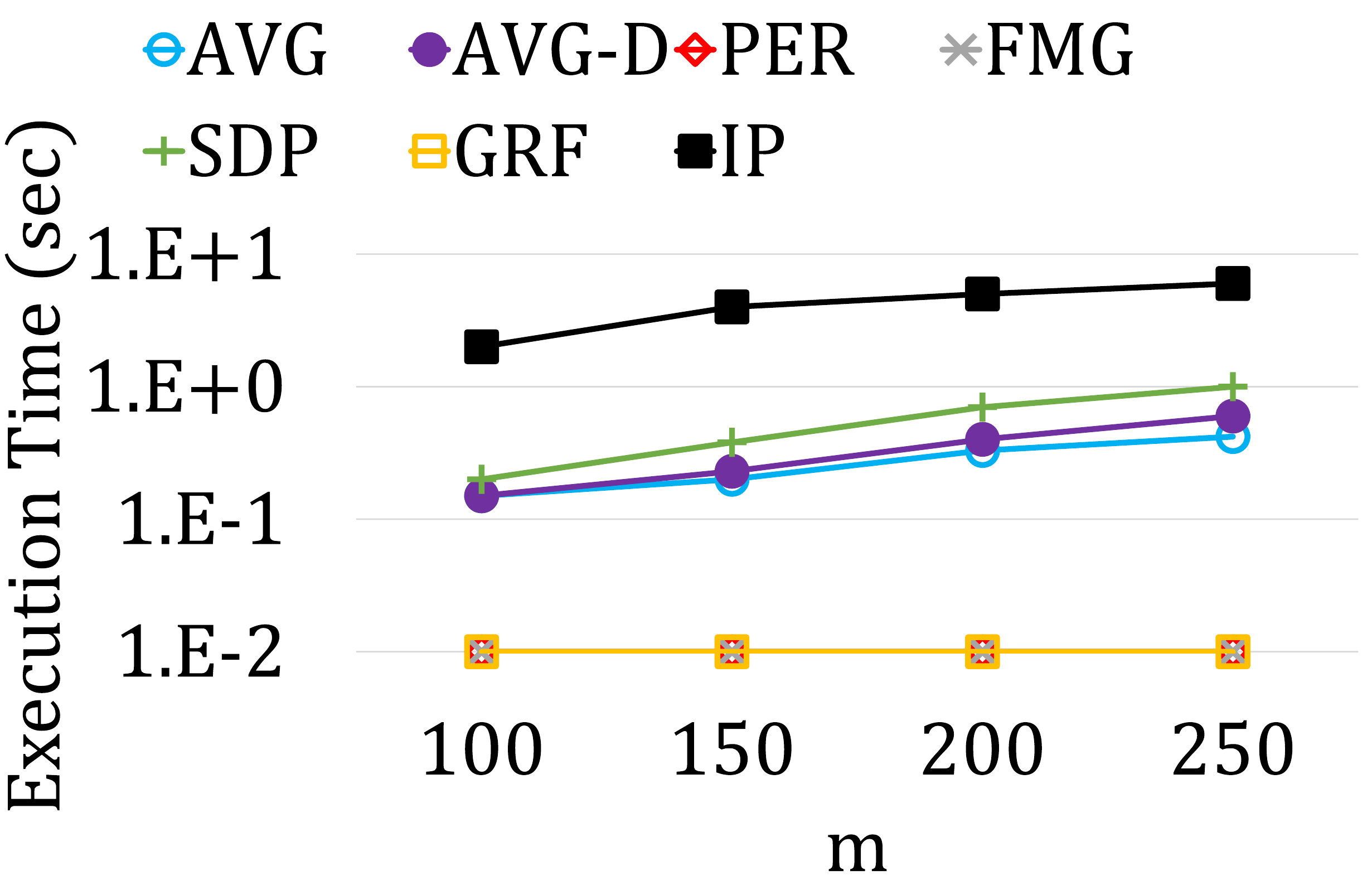}
		\label{fig:small_time_m} }		
	\subfigure[][\centering Total \scenario\ utility \newline vs. number of slots ($k$).] {\
		\centering \includegraphics[width = 0.44 \columnwidth]{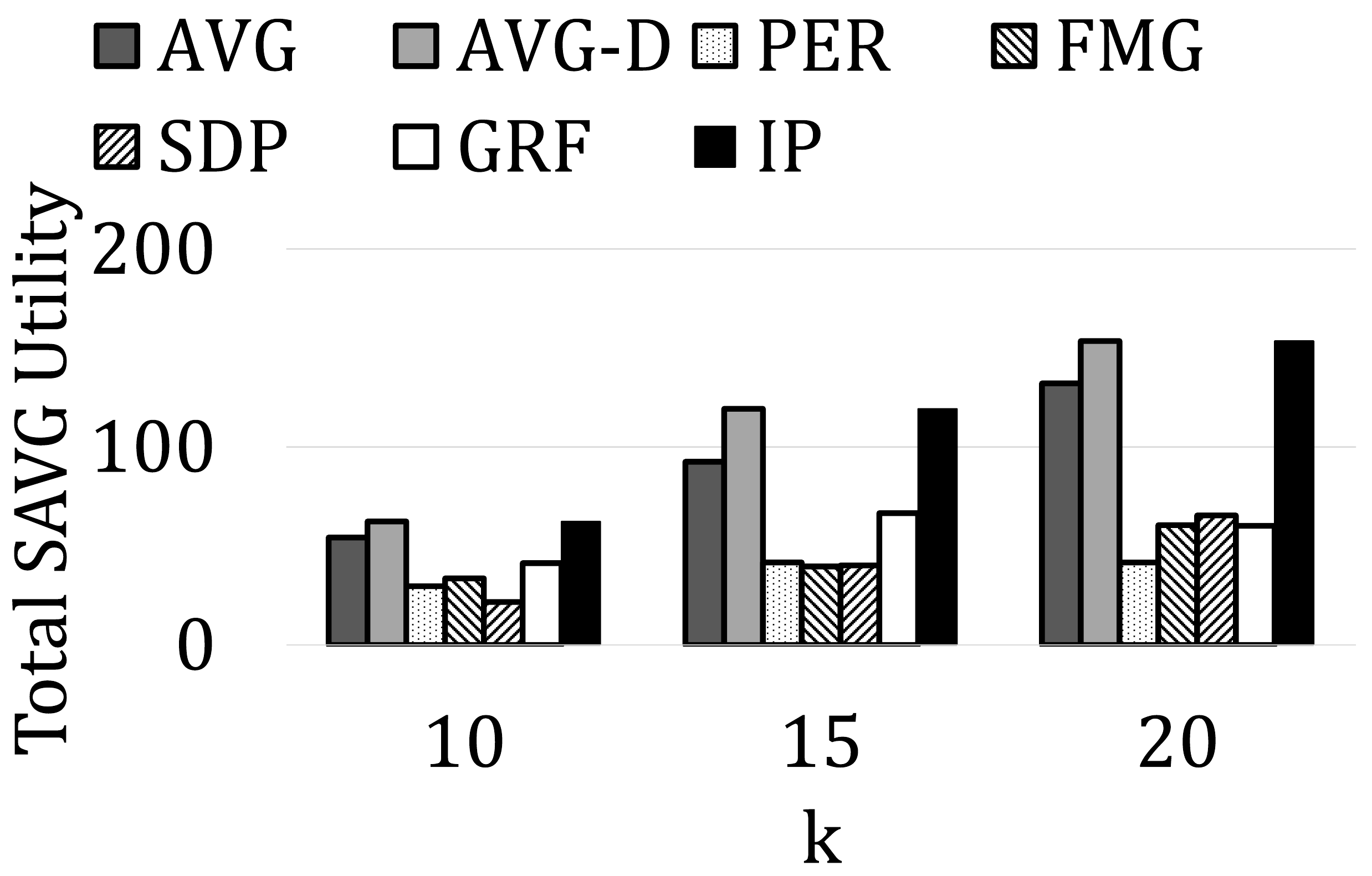}
		\label{fig:small_obj_k} } 
	\subfigure[][\centering Execution time vs. \newline number of slots ($k$).] {\
		\centering \includegraphics[width = 0.44 \columnwidth]{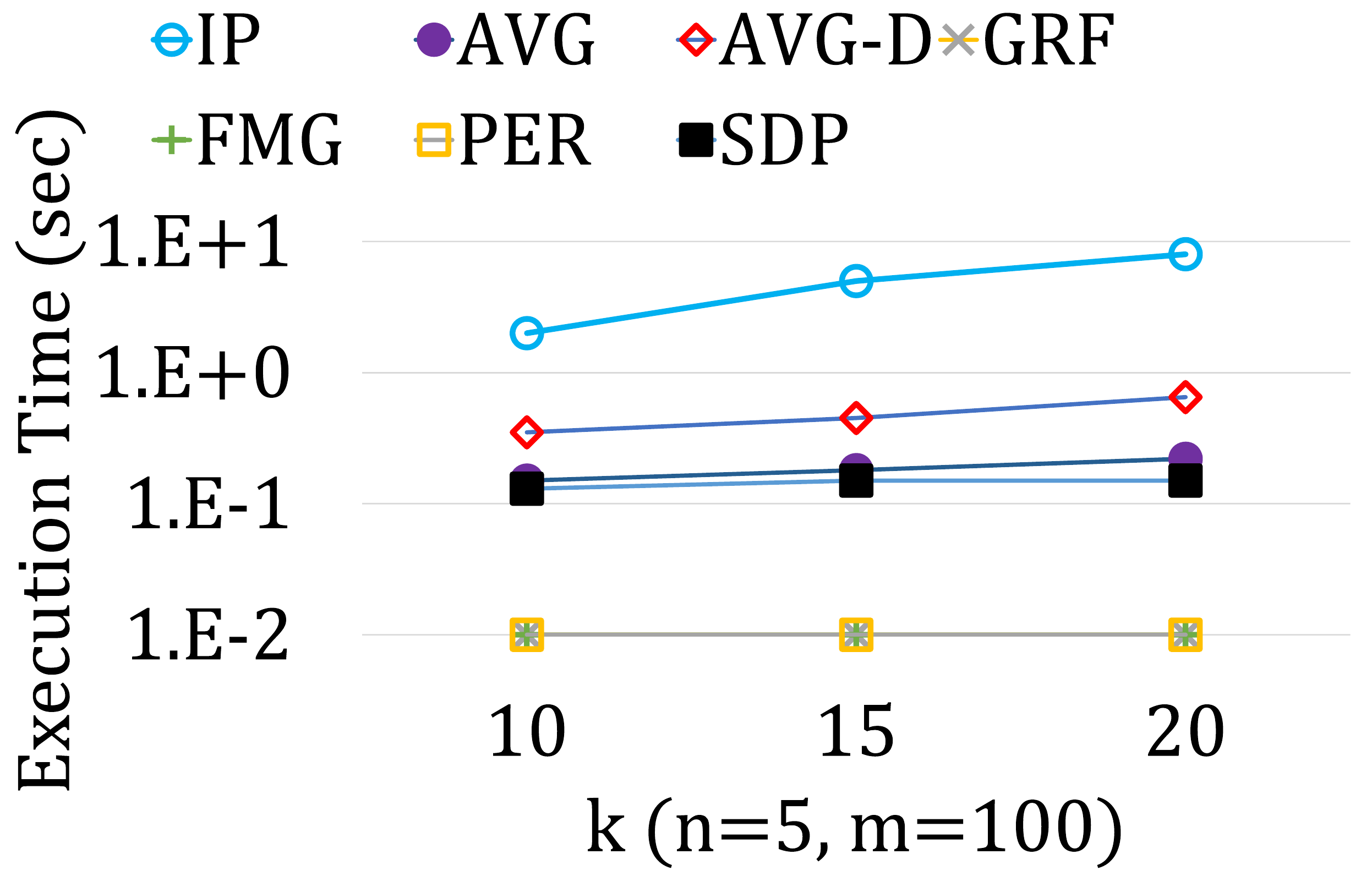}
		\label{fig:small_time_k} } 
	\caption{Comparisons on small datasets.}
	\label{exp:small_exp}
\end{figure}
}

\opt{full}{
\subsection{Comparisons on Small Datasets}
\label{subsec:small_exp}

Figure \ref{fig:small_obj_n} manifests that \algor\ and \algod\ outperform the other approaches regarding the total \scenario\ utility for different $n$ (i.e., the size of user set). Solutions of \algod\ are close to optimal since \algod\ extracts the target subgroups by jointly optimizing the current and expected future total \scenario\ utility in each iteration. \algod\ outperforms other baselines by at least 50.8\% to 62.8\% as $n$ grows, because the social interactions on extracted target subgroups become increasingly important when the size of user set increases. In contrast, the values of total \scenario\ utility of PER grows slowly because it is not designed to cope with the interplay and trade-off between personal preferences and social interactions. Moreover, while FMG and GRF seem to benefit from large values of $n$, their solution qualities are still limited due to the fixed partition of subgroups generated without leveraging CID.

\begin{figure}[tp]
  \centering \includegraphics[width = 0.95\columnwidth]{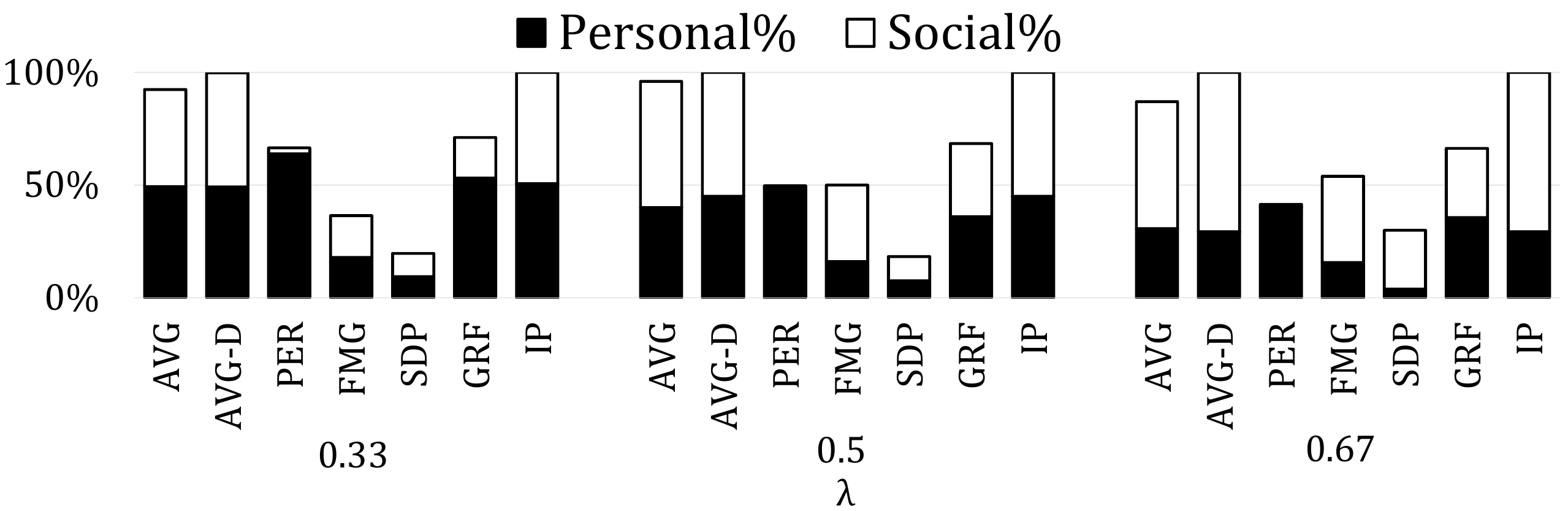}
  \caption{Normalized total \scenario\ utility of diff. $\lambda$.}
  \label{fig:small_obj_lambda}
\end{figure}

\opt{full}{Figure \ref{fig:small_obj_m} presents the total \scenario\ utility under varied numbers of items ($m$). However, $m$ does not seem to affect the total \scenario\ utility by much since any user's top preferred items are already contained in the top-100 items.} Moreover, Figure \ref{fig:small_obj_k} presents the total \scenario\ utility under varied number of slots ($k$). \algod\ and \algor\ significantly outperform the baselines by at least 134.7\% and 102.1\% in terms of the total \scenario\ utility when $k$ grows to 20, because \stagethree\ leverages the flexibility provided by \notion\ to optimize the social utility for different slots, which is beneficial for a large $k$ as it becomes difficult to find more commonly interested items for static subgroup members. \opt{short}{The values of total \scenario\ utility with regard to different numbers of items ($m$) are almost the same when $n$ and $k$ are fixed (please refer to \cite{Online} due to space constraint).} \opt{short}{Figures \ref{fig:small_time_n} and \ref{fig:small_time_m} show the execution time by varying the size of user set ($n$) and the size of item set ($m$).}\opt{full}{Figures \ref{fig:small_time_n}, \ref{fig:small_time_m} and \ref{fig:small_time_k} show the execution time by varying the size of user set ($n$), the size of item set ($m$), and the number of slots ($k$).} The running times of \algor\ and \algod\ are at most 7.5\% and 17.4\% of that of IP. \algor\ and \algod\ require slightly more time than PER, GRF, and FMG, since the baseline approaches focus only on one of preference utility, social utility, or subgroup partition, instead of considering all of them jointly. Figure \ref{fig:small_obj_lambda} evaluates the impact of $\lambda$ on the total \scenario\ utility of all schemes normalized by that of IP. The normalized total \scenario\ utilities of FMG and SDP improve when the social utility becomes more important as $\lambda$ grows. However, it is difficult for them to address the diverse personal preferences. In contrast, PER achieves the highest preference utility and the lowest social utility, but tends to generate a small total \scenario\ utility. 
}

\subsection{Sensitivity Tests on Large Datasets}
\label{subsec:large_exp_sensitivity}
\begin{figure}[tp]
  \centering \includegraphics[width = 0.95\columnwidth]{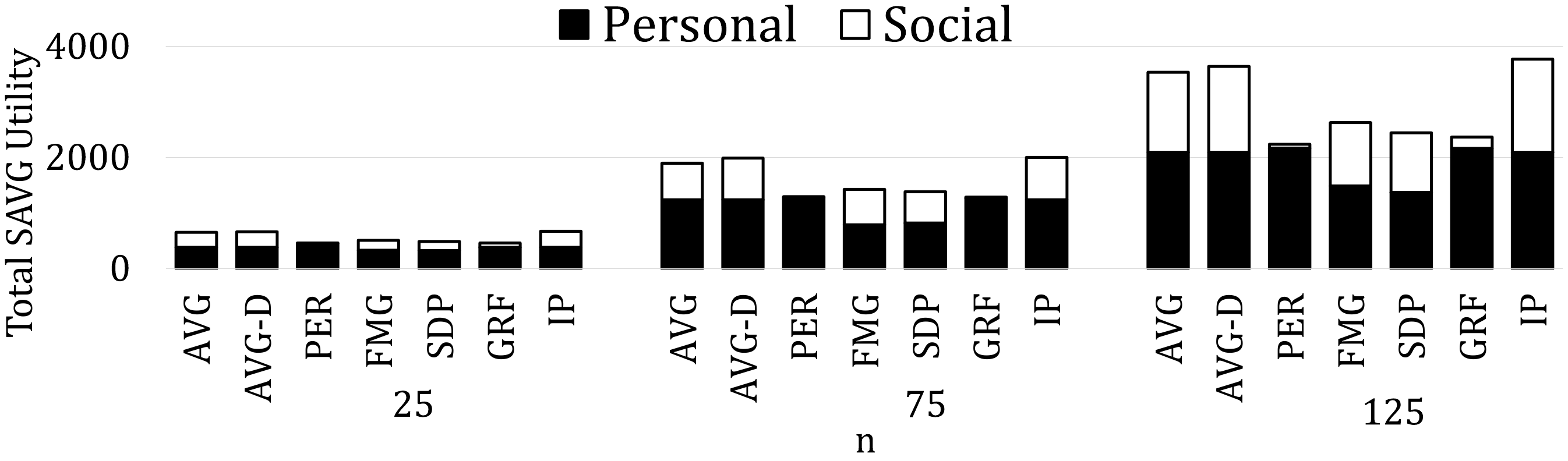}
  \caption{Total \scenario\ utility vs. size of user set ($n$).}
  \label{fig:timik_obj_n}
\end{figure}
\begin{figure}[tp]
  \centering \includegraphics[width = 0.95\columnwidth]{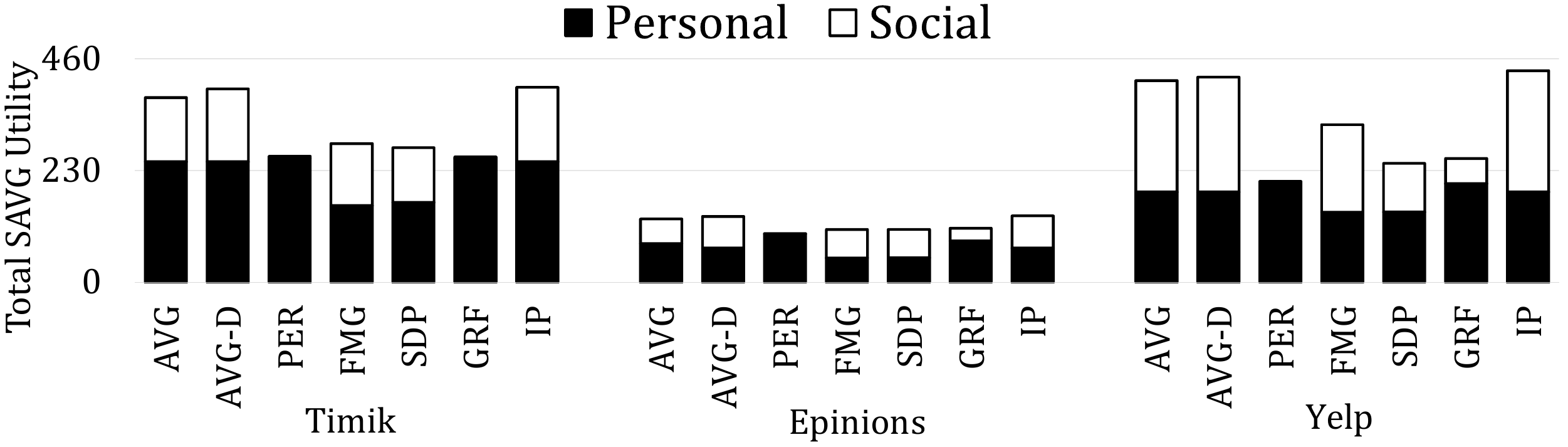}
  \caption{Total \scenario\ utility in diff. datasets.}
  \label{fig:dataset_obj}
\end{figure}
\begin{figure}[tp]
  \centering \includegraphics[width = 0.95\columnwidth]{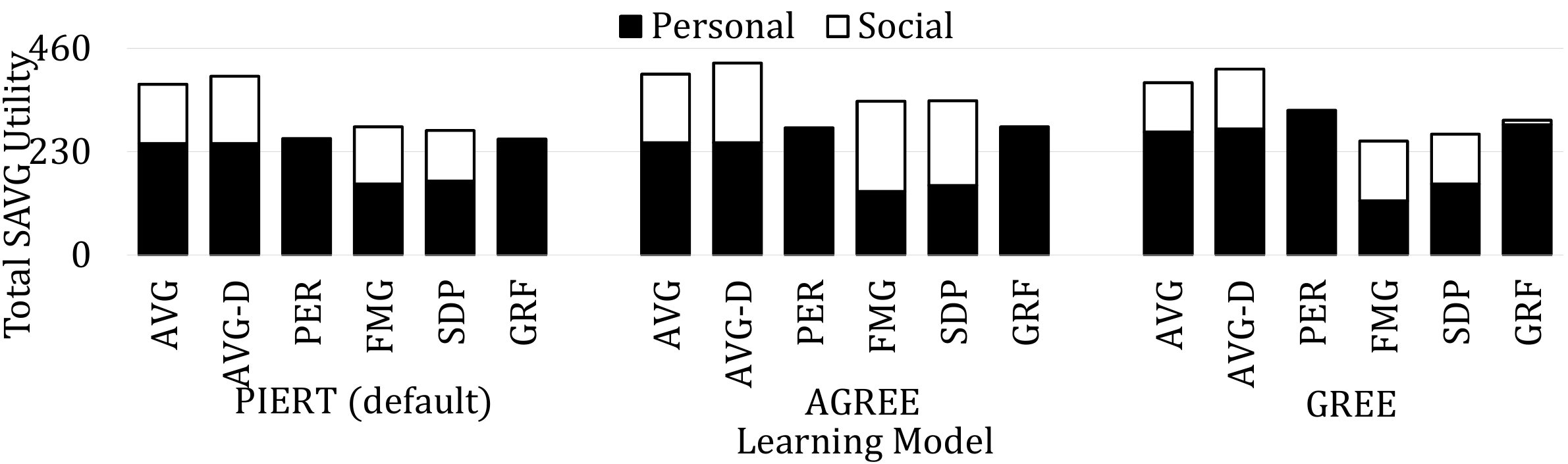}
  \caption{Total \scenario\ utility vs. different input.}
  \label{fig:timik_input}
\end{figure}

We evaluate the efficiency and efficacy of \algor\ and \algod\ in large datasets with the scales of the three dimensions following previous research \cite{SD17,SBR15}, i.e., $m=10000$, $k=50$, and $n=125$. 
Figure \ref{fig:timik_obj_n} presents the total \scenario\ utility by varying the sizes of user set in Timik. The results manifest that \algor\ and \algod\ outperform all baselines by at least 30.1\%, while \algod\ is slightly better than \algor\ since it selects better pivot parameters for \stagethree. \revise{Moreover, the returned objective values of \algor\ and \algod\ respectively achieve at least 93.7\% and 96.4\% of the objective value of IP, manifesting that our algorithms are effective.} Compared with GRF, the improvement of \algod\ grows from 43.6\% to 54.6\% as $n$ increases, since GRF splits the users into subgroups without considering social relations, but social interactions among \emph{close friends} become more important for a larger group. By striking a balance between preference and social utility, \algor\ and \algod\ achieve a greater total \scenario\ utility. Compared with PER and GRF, FMG achieves a higher social utility but a lower preference utility because it displays a universal configuration to all users.

Figure \ref{fig:dataset_obj} compares the results on Timik, Yelp, and Epinions. The social utility obtained in Epinions is lower than in Yelp due to the sparser social relations in the review network, and group consensus thereby plays a more important role in Yelp. 
Despite the different characteristics of datasets, \algor\ and \algod\ prevail in all datasets and outperform all baselines since \stagethree\ operates on the \weight s from the optimal LP solution, which does strike a good balance among all factors. FMG and SDP benefit from the high social utility in Yelp and outperform PER. By contrast, PER performs nearly as good as FMG and SDP in Epinions since the social utility is lower. 

\revise{Next, to examine the influence of input models on the tackled problem, Figure \ref{fig:timik_input} shows the experimental results on different inputs generated by PIERT~\cite{YL18TOIS} (default), AGREE and GREE~\cite{DC18}. PIERT jointly learns the preference and social utilities by modeling the social influence between users and the latent topics of items. For AGREE and GREE~\cite{DC18}, the former assumes the social influence between users is equal, and the later learns sophisticated weights of the triple (user, user, item). \algo\ and \algod\ outperform the baselines with regards to all different input models, manifesting that our method is generic to different distributions of inputs. Note that the social utilities returned by \algor\ and \algod\ with PIERT and AGREE are slightly greater than the ones with GREE. The result manifests that \algor\ and \algod\ can select better items for users to enjoy if social utilities are different across items.}

\subsection{Scalability Tests on Large Datasets}
\label{subsec:large_exp_scalability}
\begin{figure}[tp]
	\centering
	\subfigure[b][\centering Execution time vs. \newline size of user set ($n$).] {\
		\centering \includegraphics[width = 0.44 \columnwidth]{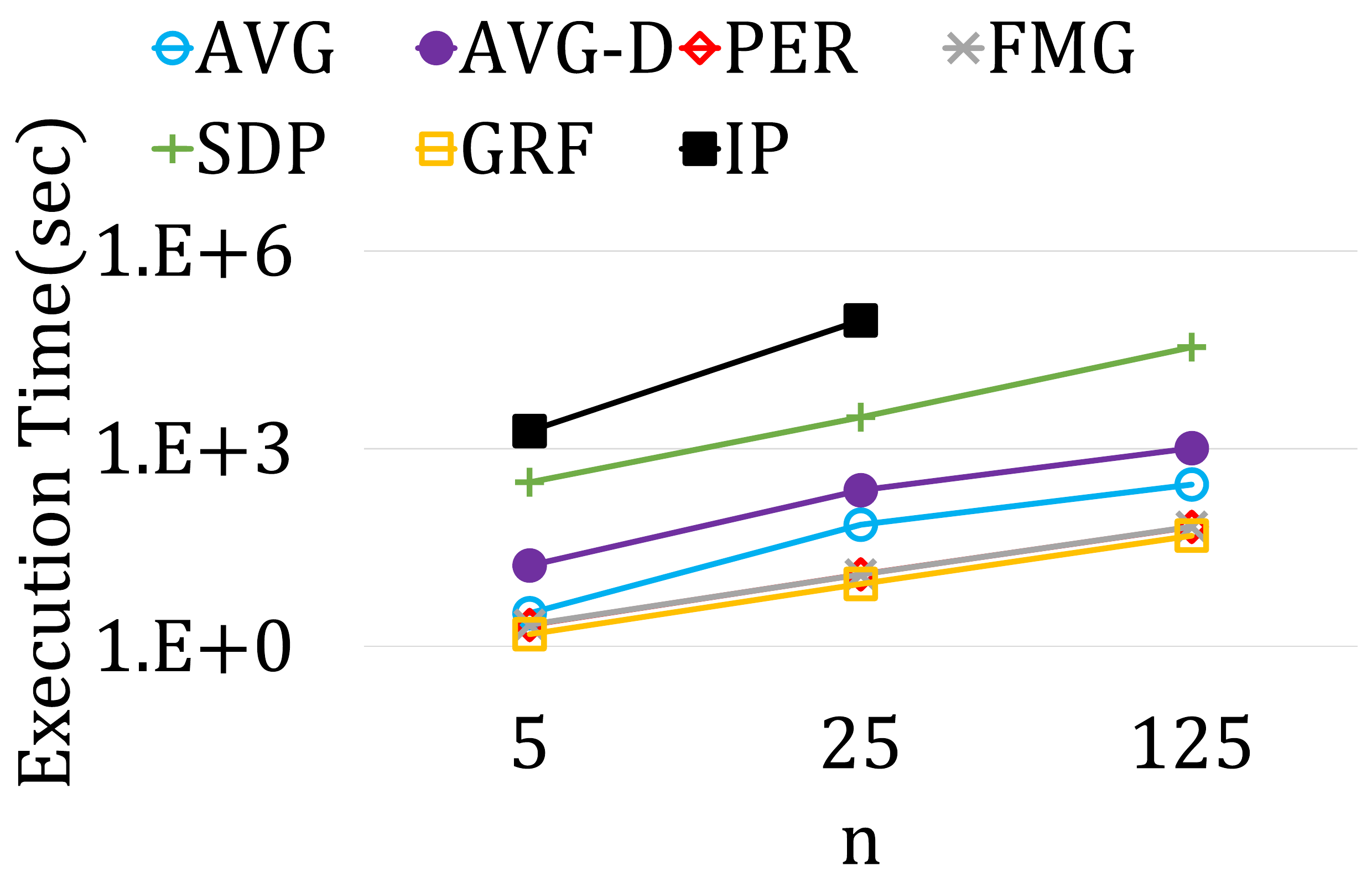}
		\label{fig:timik_time_n} } 
	\subfigure[b][\centering Execution time vs. \newline size of item set ($m$).] {\
		\centering \includegraphics[width = 0.44 \columnwidth]{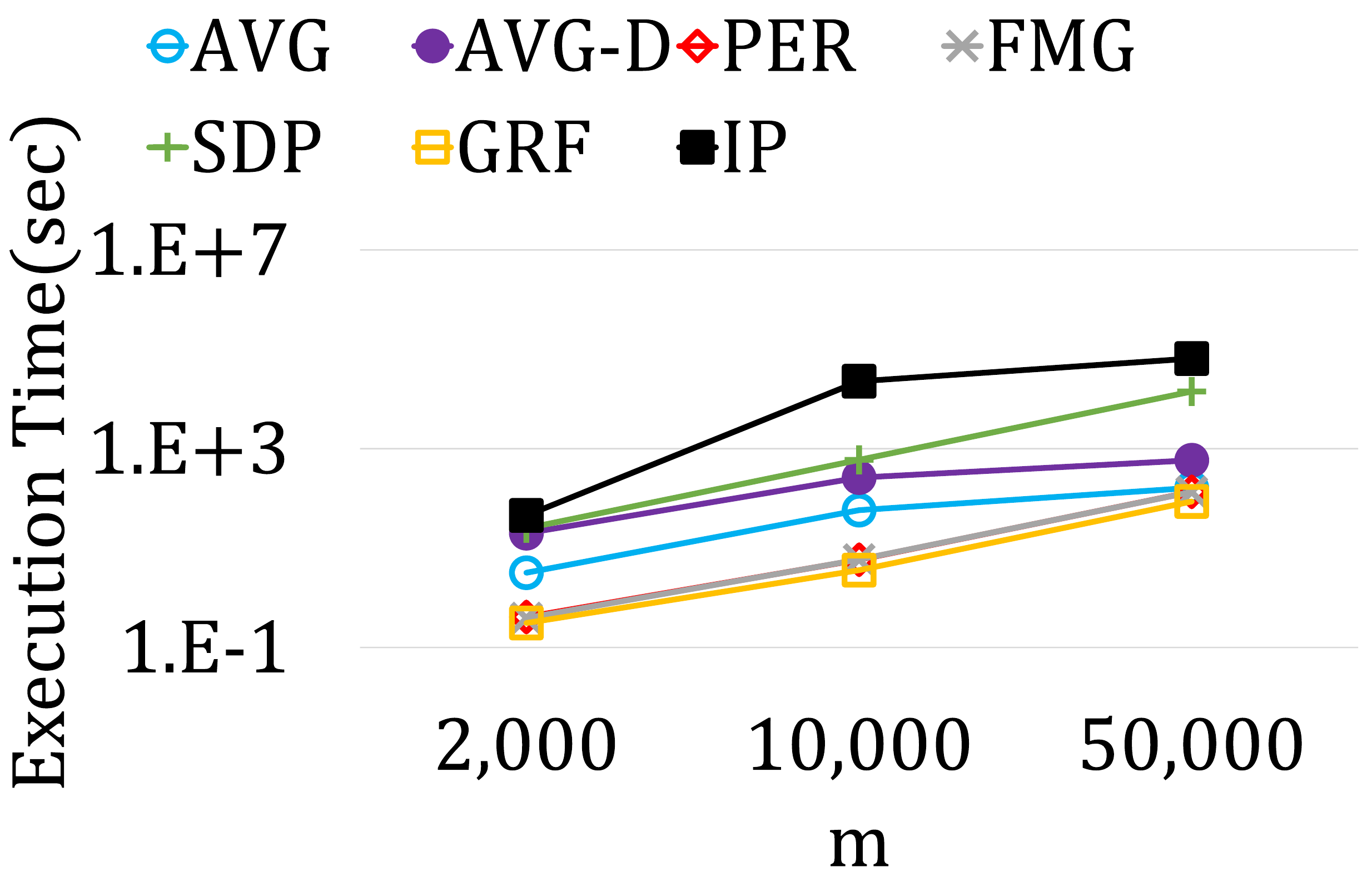}
		\label{fig:timik_time_m} }
	\caption{Execution time in Yelp.}
	\label{fig:large_scalability}
\end{figure}
\opt{short}{ 
\begin{figure}[tp]
  \centering \includegraphics[width = 0.48\columnwidth]{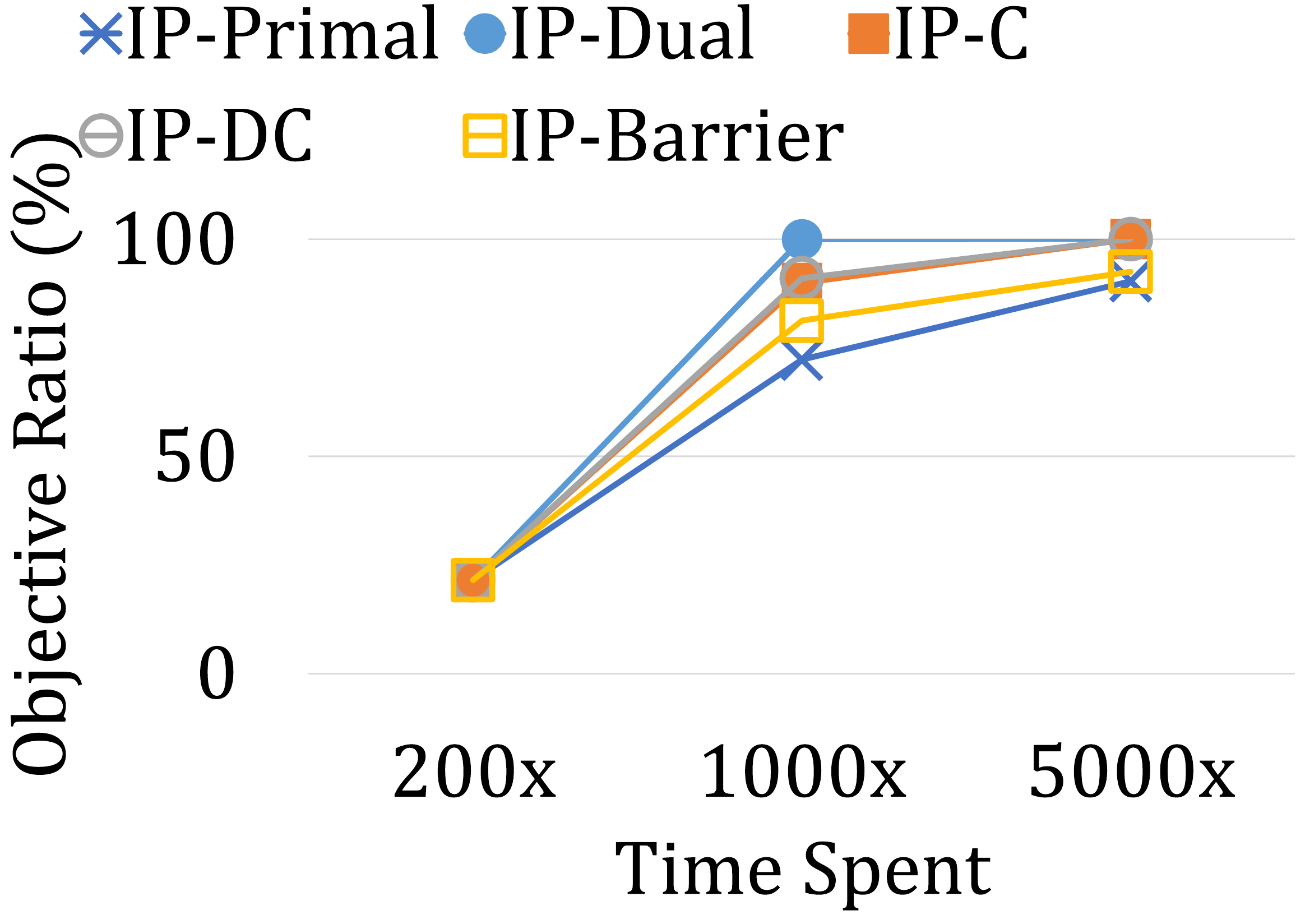}
  \caption{Results of different IP heuristics.}
  \label{fig:timik_IP_tradeoff}
\end{figure}
}
\opt{full}{ 
\begin{figure}[tp]
	\centering
	\subfigure[b][\centering Results of different IP heuristics.] {\
		\centering \includegraphics[width = 0.44 \columnwidth]{fig/large_scale/timik_IP_tradeoff.pdf}
		\label{fig:timik_IP_tradeoff} } 
	\subfigure[b][\centering Effects of speedup strategies.] {\
		\centering \includegraphics[width = 0.44 \columnwidth]{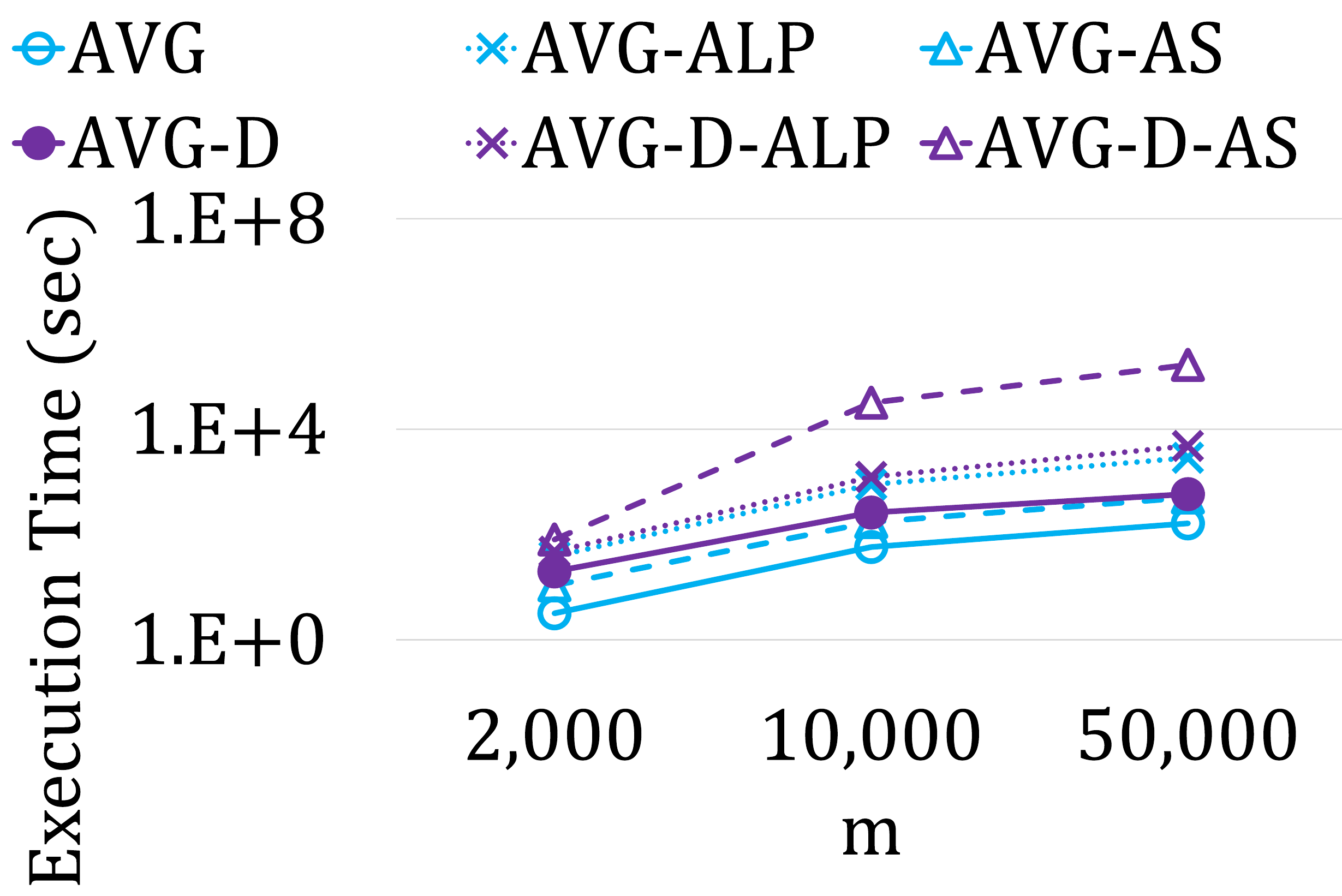}
		\label{fig:timik_time_speedup} }
	\caption{Sensitivity Test on Efficiency and Efficacy.}
	\label{fig:new_sensitivity_tests}
\end{figure}
}
Figure \ref{fig:timik_time_n} presents the execution time in Yelp with different $n$. IP cannot terminate within 24 hours when $n \geq 25$, and SDP needs 300 seconds to return a solution even when $n=5$. Figure \ref{fig:timik_time_m} shows the execution time with different $m$. Note that \algor\ and \algod\ are both more scalable to $m$ than the baseline approaches because \stagethree\ exploits the \fragsol\ without $m$ in the complexity. Although \algod\ provides a stronger theoretical guarantee, the scalability of \algor\ on $n$ is better than \algod\ because \algor\ samples the target subgroups randomly. However, practical VR applications, e.g., VRChat \cite{VRChat} and IrisVR \cite{IrisVR}, seldom have $n > 25$. 

\opt{full}{\algo, \algod, and IP are all integer program (IP)-based approaches, which means solving an IP (or a relaxed linear program) is the first step. For a fixed $k$ (number of display slots for a user), when $n$ and $m$ grow, the numbers of variables and constraints in IP and LP also increase. However, as the number of \textit{core decisions} (which items to be displayed in the $k$ slots) remains the same, an increasing number of newly introduced constraints become \textit{dominated} (i.e., redundant) and thus pre-solved (relaxed) by the highly optimized commercial IP/LP solvers. Accordingly, the empirical complexity of the problem models \textit{after} pre-solving will not grow explicitly with $n$ and $m$, rendering the practical computational cost of solving the IP/LP more scalable regarding $n$ and $m$. Moreover, in \algo\ and \algod, after the optimal fractional solution for the LP is retrieved, decision variables with values of 0 or 1 do not require the randomized/deterministic rounding procedure and can be directly filled into the final solution. This usually happens when 1) some user has a very strong preference toward a specific item such that viewing the item by her own still outweighs viewing other items with friends, and 2) some item is extremely popular and thus is co-displayed to a majority of users. As $n$ and $m$ increases in this case, the total number of decision variables $x^\ast(u,\cdot)$ grows, but the summation of all decision variables remains constant in Constraint (\ref{ilp:one_slot_one_item}). Therefore, more decision variables have values close to 0 in this case. On the other hand, all other baselines are based on local greedy heuristics that need either a linear scan on all users ($O(n)$) or all items ($O(m)$) at each step, and do not benefit from the above ``decision dilution'' phenomenon.}

\revise{To examine the performance of different mixed integer programming (MIP) algorithms, we further conducted experiments with Primal-first Mixed Integer Programming (IP-Primal), Dual-first Mixed Integer Programming (IP-Dual), Concurrent Mixed Integer Programming (IP-C), Deterministic Concurrent Mixed Integer Programming (IP-DC), and the Barrier Algorithm (IP-Barrier) in the Gurobi \cite{Gurobi} package. Figure \ref{fig:timik_IP_tradeoff} shows the trade-off between efficiency and efficacy of five different IP algorithms on the Timik dataset with the default parameters ($(k,m,n) = (50,10000,125)$). For every MIP algorithm, we evaluate the solution quality of different algorithms with the running time constrained by 200, 1000 and 5000 times the running time of our proposed \algod\ algorithm for the same instance to compare the obtained solutions in different time limits. The y-axis shows the objective value normalized by the solution of \algod, which manifests none of the 5 baselines achieves any solution better than that of \algod\ in 5000X of the running time of \algod. As such, although there is some subtle difference in performance across different MIP algorithms, none of the examined algorithms shows a reasonable scalability. \opt{short}{Please see the full version \cite{Online} for more discussions on the scalability.}}

\opt{full}{Figure \ref{fig:timik_time_speedup} shows the effects of the speedup strategies on \algo\ and \algod. \algo--ALP and \algod--ALP are \algo\ and \algod\ without the advanced LP transformation technique, respectively, while \algo--AS and \algod--AS are \algo\ and \algod\ without the advanced focal parameter sampling technique, respectively. The results manifest that all speedup strategies effectively boost the efficiency of \algo\ and \algod. For \algo, the advanced LP transformation plays a more dominant role than the advanced focal parameter sampling scheme since the computational cost of solving the linear program is more evident in \algo. Conversely, as \algod\ needs to carefully examine every combination of \textit{good} focal parameters (i.e., those sets of parameters that do assign an item to at least one user), the advanced sampling technique has a significant effect in filtering out bad focal parameters. Therefore, the effect of the advanced LP transformation pales in comparison to that of the advanced sampling technique in \algod.}

\opt{short}{
\begin{figure}[tb]
	\centering
	\subfigure[][\centering Inter/Intra\% and subgroup density (Timik).] {\
		\centering \includegraphics[width = 0.44 \columnwidth]{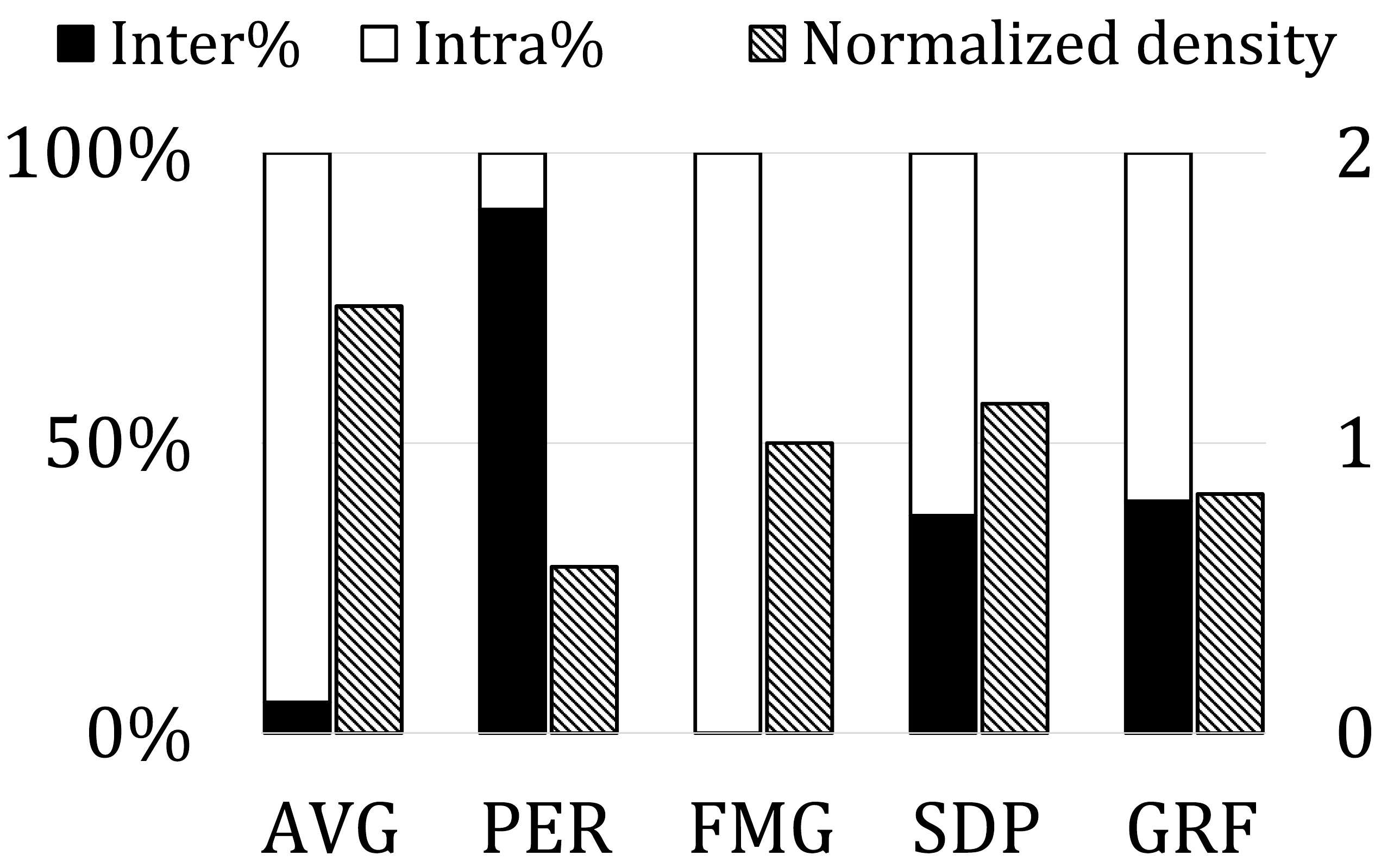}
		\label{fig:timik_inter_intra_density} }
	\subfigure[][\centering Inter/Intra\% and subgroup density (Epinions).] {\
		\centering \includegraphics[width = 0.44 \columnwidth]{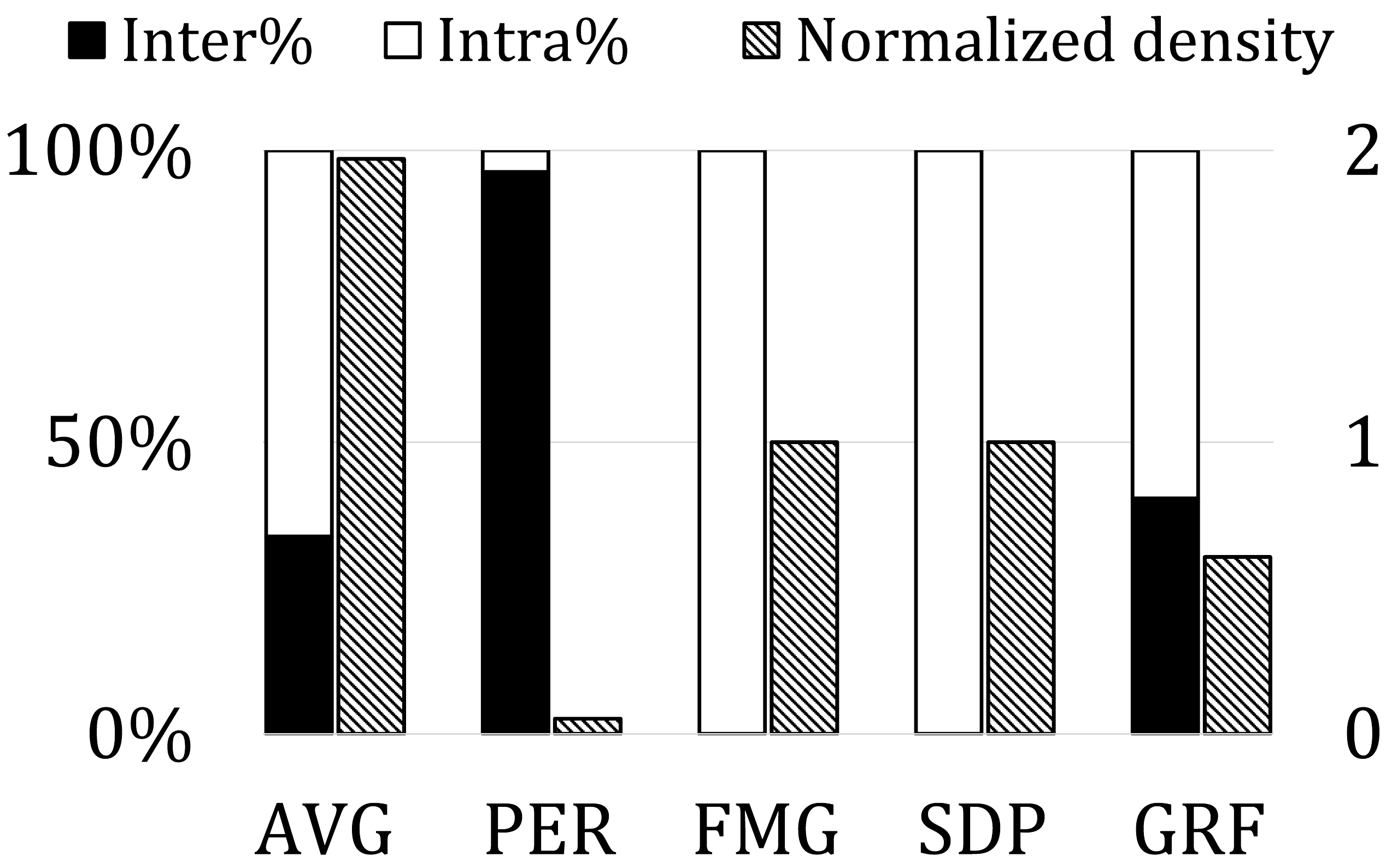}
		\label{fig:ep_inter_intra_density} }
		
	\subfigure[][\centering CDF of Regret ratio (Timik).] {\
		\centering \includegraphics[width = 0.44 \columnwidth]{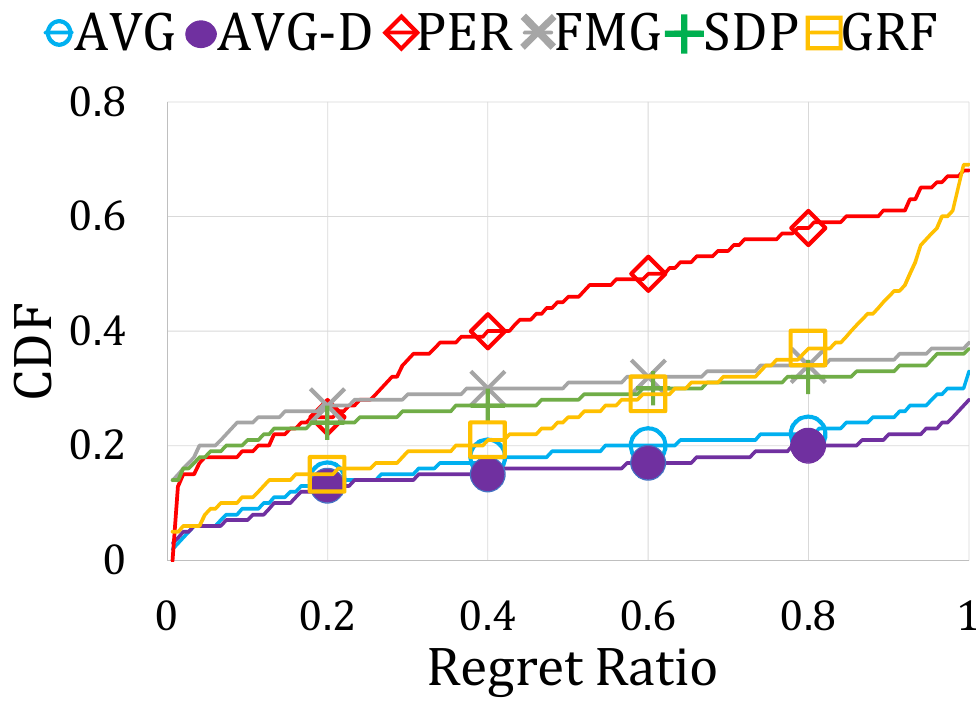}
		\label{fig:timik_regret} } 
	\subfigure[][\centering CDF of Regret ratio (Epinions).] {\
		\centering \includegraphics[width = 0.44 \columnwidth]{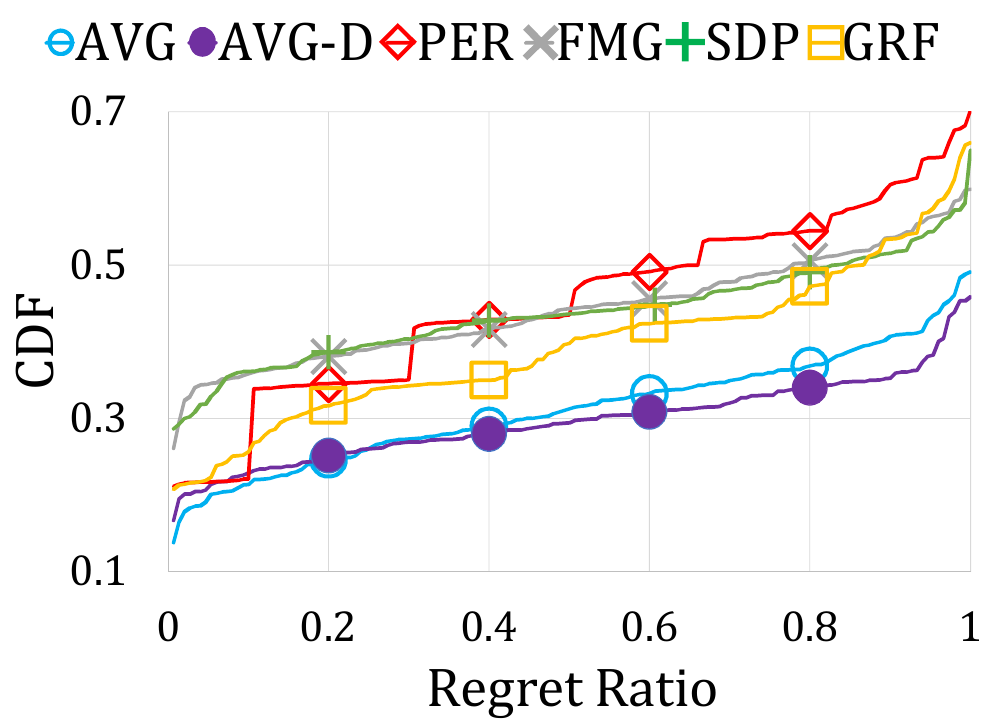}
		\label{fig:ep_regret} }
		
	\caption{Comparisons on subgroup metrics.}
	\label{exp:group_all}
\end{figure}
}

\opt{full}{	
\begin{figure*}[tb]
	\centering
	\subfigure[][\centering Inter/Intra\% and subgroup density (Timik).] {\
		\centering \includegraphics[width = 0.44 \columnwidth]{fig/subgroup_metrics/timik_density.pdf}
		\label{fig:timik_inter_intra_density} } 
	\subfigure[][\centering Inter/Intra\% and subgroup density (Epinions).] {\
		\centering \includegraphics[width = 0.44 \columnwidth]{fig/subgroup_metrics/ep_density.pdf}
		\label{fig:ep_inter_intra_density} }
	\subfigure[][\centering Inter/Intra\% and subgroup density (Yelp).] {\
		\centering \includegraphics[width = 0.44 \columnwidth]{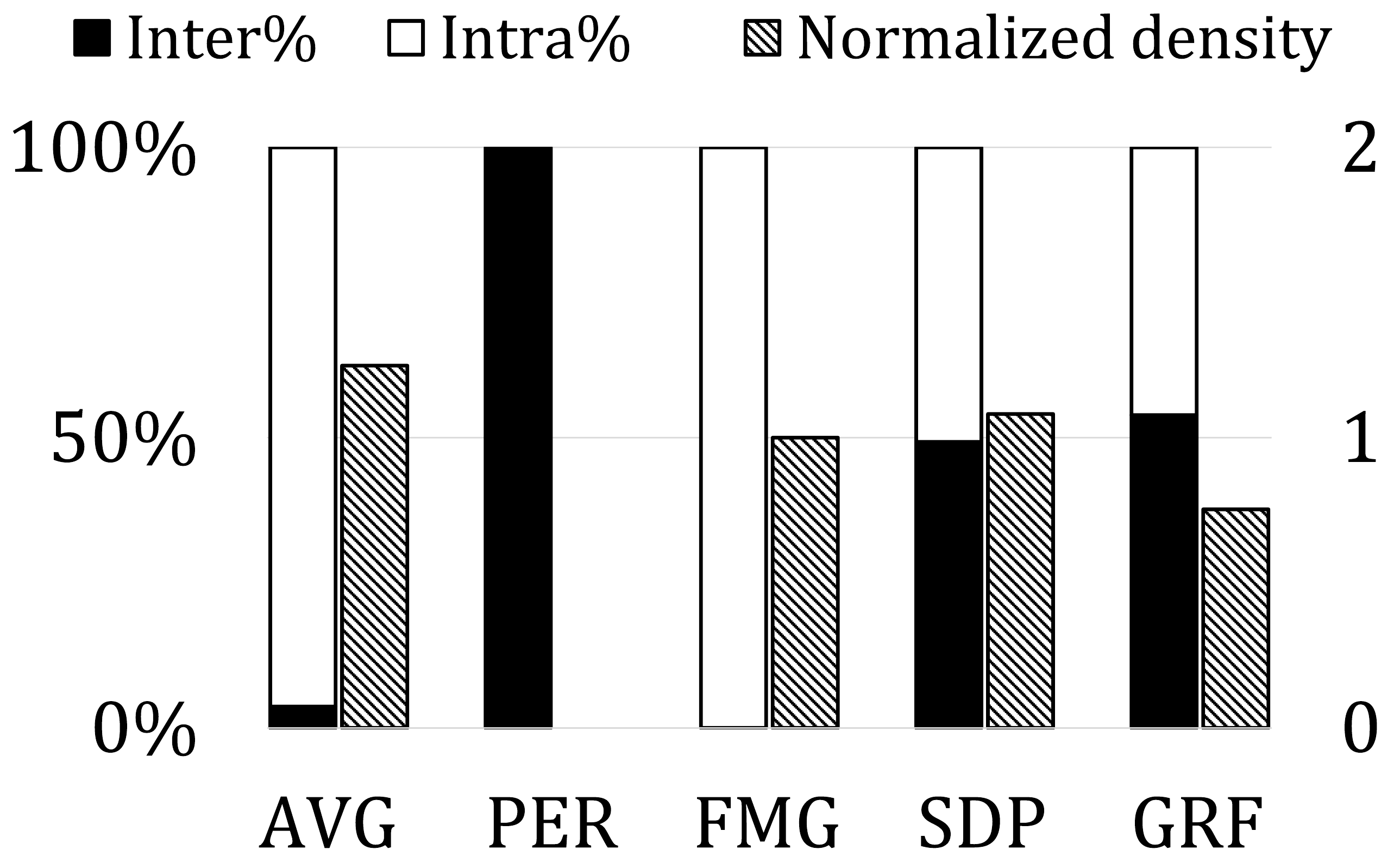}
		\label{fig:yelp_inter_intra_density} } 

	\subfigure[][\centering Share rate and alone rate (Timik).] {\
		\centering \includegraphics[width = 0.44 \columnwidth]{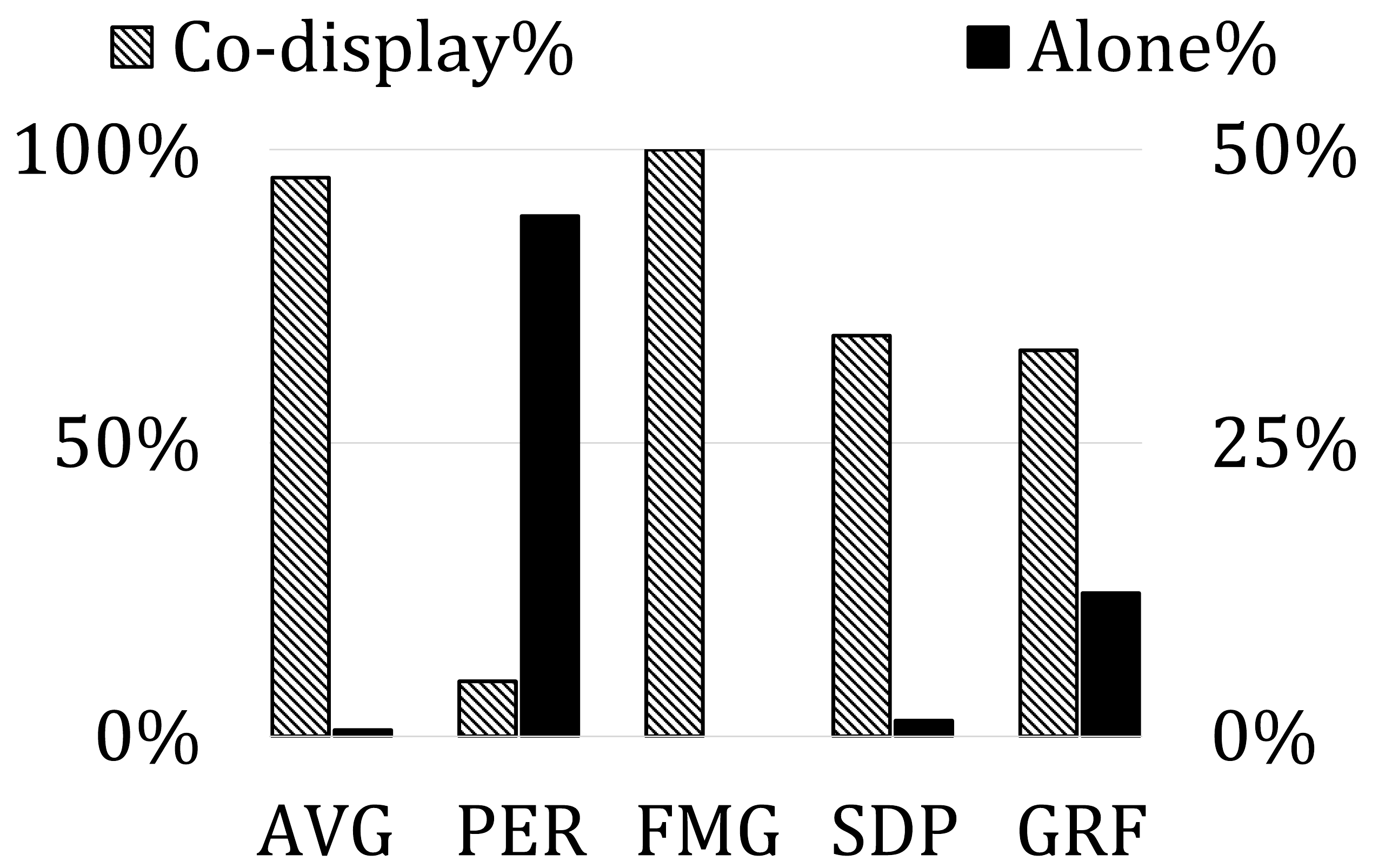}
		\label{fig:timik_share_alone} } 
	\subfigure[][\centering Share rate and alone rate (Epinions).] {\
		\centering \includegraphics[width = 0.44 \columnwidth]{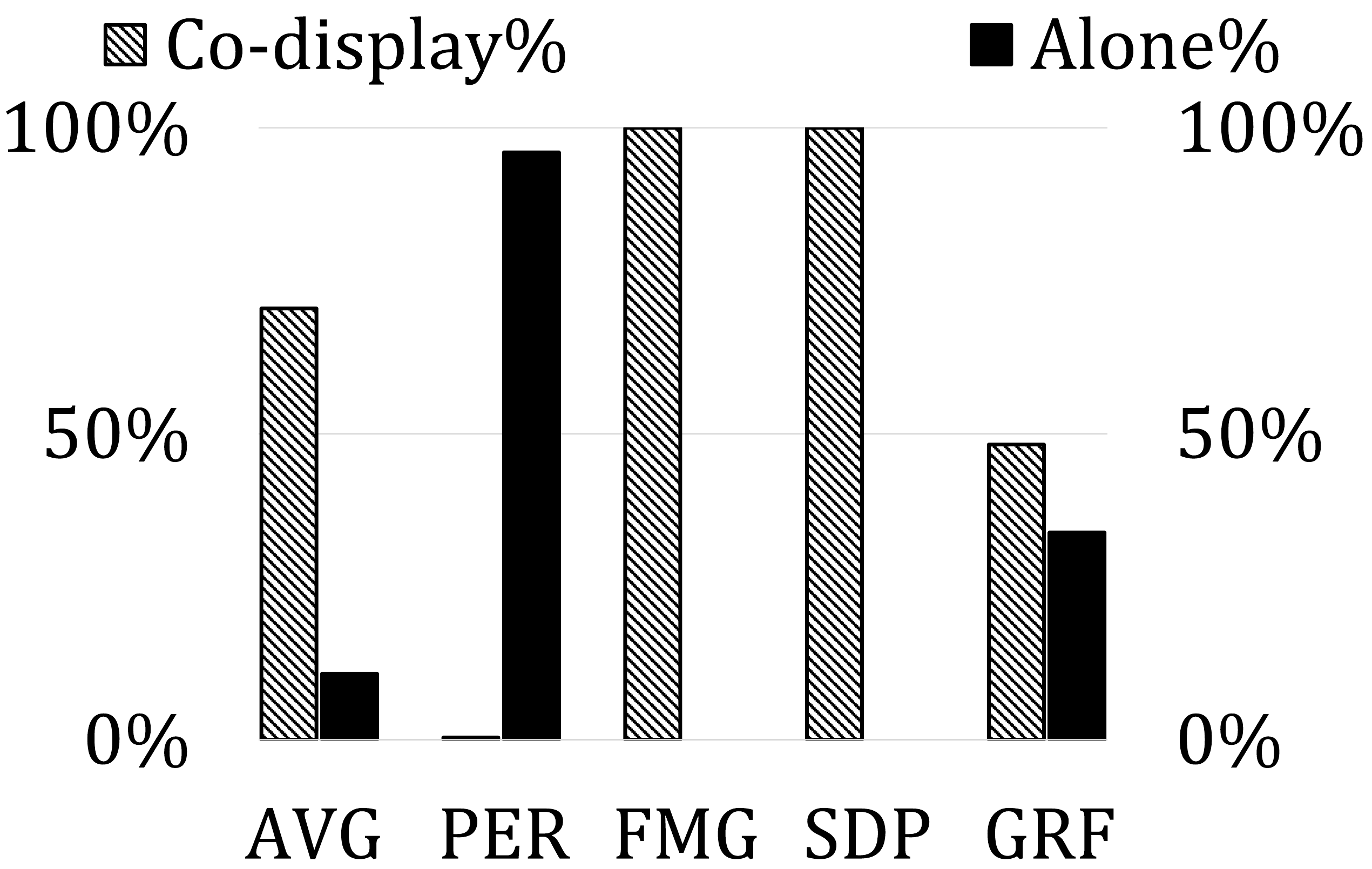}
		\label{fig:ep_share_alone} }
	\subfigure[][\centering Share rate and alone rate (Yelp).] {\
		\centering \includegraphics[width = 0.44 \columnwidth]{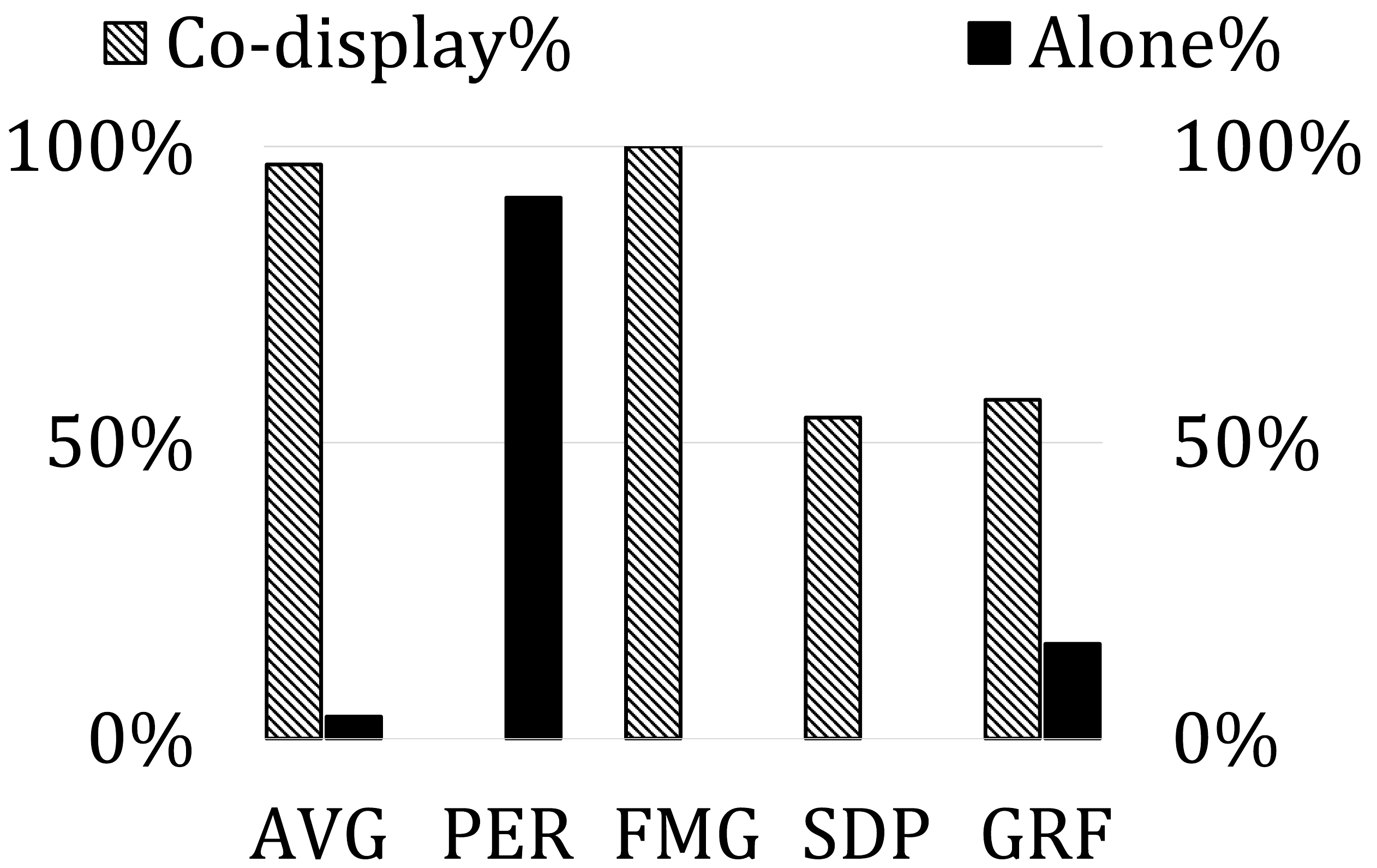}
		\label{fig:yelp_share_alone} }
	
	\subfigure[][\centering Regret ratio (Timik).] {\
		\centering \includegraphics[width = 0.44 \columnwidth]{fig/subgroup_metrics/timik_regret.pdf}
		\label{fig:timik_regret} } 
	\subfigure[][\centering Regret ratio (Epinions).] {\
		\centering \includegraphics[width = 0.44 \columnwidth]{fig/subgroup_metrics/ep_regret.pdf}
		\label{fig:ep_regret} }
	\subfigure[][\centering Regret ratio (Yelp).] {\
		\centering \includegraphics[width = 0.44 \columnwidth]{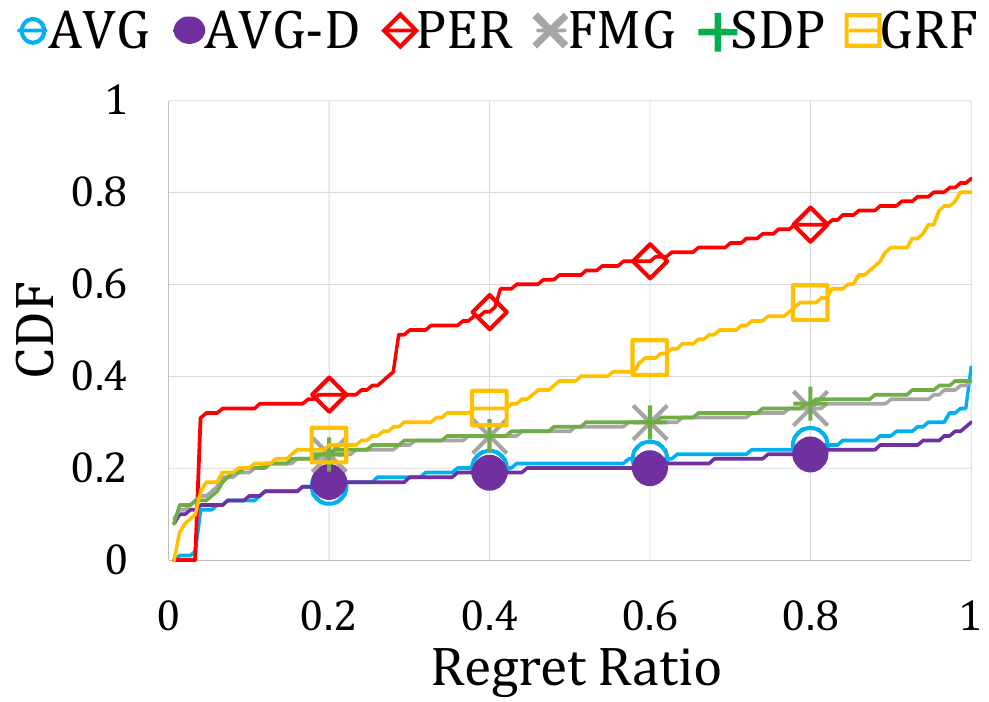}
		\label{fig:yelp_regret} } 
	\caption{Comparisons on subgroup metrics.}
	\label{exp:group_all}
\end{figure*}
}

\subsection{Comparisons on Subgroup Metrics}
\label{subsec:group_formation_exp}

\opt{short}{ 
In the following, we analyze the subgroups in the \configuration s returned from all algorithms in terms of subgroup-related metrics. Figures \ref{fig:timik_inter_intra_density} and \ref{fig:ep_inter_intra_density} compare the ratios of Inter-/Intra-subgroup edges averaged across all slots in the \configuration s in the Timik and Epinions datasets, respectively. It also shows the average network density among the partitioned subgroups normalized by the original density of the input social network. The results from all datasets indicate that the majority of preserved edges by \algo\ are within the same subgroups (large Intra\%). FMG achieves 100\% in Intra\%, 0\% in Inter\%, and 1 in normalized density because it consistently views the whole network as a large subgroup. In contrast, PER has a high Inter\% as it separates most users into independent subgroups to display their favorite items. There exist a small subset of widely liked or adopted items in Epinions that appear as most users' favorite (hence the small nonzero Intra\% of PER), while famous VR locations (e.g., transportation hubs in Timik) are inclined to be associated with high preference utilities by the exploited recommendation learning models as they generate a lot of check-ins among all users. Therefore, they have a higher chance to be co-displayed by PER. Among all methods, \algo\ achieves the largest normalized density as \stagethree\ carefully considers the \weight s to partition the network into dense communities. 

\revise{Figures \ref{fig:timik_regret} and \ref{fig:ep_regret} report the Cumulative Distribution Function (CDF) of \textit{regret ratios} of all algorithms in the Timik and Epinions datasets. The regret ratio reg$(u)$ \cite{KL14KDD} is a game-theory based measurement for the satisfaction of individual users and the overall \textit{fairness} of the solution. For each user $u$, the regret ratio reg$(u)$, and its converse, \textit{happiness} ratio $\text{hap}(u)$, are defined as follows.
\vskip -0.15in
$$\text{reg}(u) \equiv 1 - \text{hap}(u); \quad \text{hap}(u) \equiv \frac{\sum\limits_{c \in \mathbf{A}^{\ast}(u,\cdot)} \text{w}_{\textbf{A}}(u,c)}{\max\limits_{\mathcal{C}_u} \sum\limits_{c \in \mathcal{C}_u} \Bar{\text{w}}_{\textbf{A}}(u,c)}$$
\noindent where the numerator in $\text{hap}(u)$ is the achieved \scenario\ utility, the denominator with $\Bar{\text{w}}_{\textbf{A}}(u,c) = (1-\lambda)\cdot\text{p}(u,c) + \lambda\cdot\sum_{v \in V} \tau(u,v,c)$ is an \textit{upper bound} of possible \scenario\ utility when all users view $c$ together with user $u$, and $\mathcal{C}_u$ is a $k$-itemset, corresponding to a very optimistic scenario favoring $u$ the most. Note that the second term of $\Bar{\text{w}}_{\textbf{A}}(u,c)$ is different from $\text{w}_{\textbf{A}}(u,c)$ in Definition \ref{def:savg_utility}. In other words, the upper bound is the \scenario\ utility of $u$ if she dictates the whole \configuration\ selfishly. Therefore, a high $\text{hap}(u)$ (equivalently, a low reg$(u)$) implies that user $u$ is relatively satisfied with the \configuration, and \textit{fairness} among the users can be observed by inspecting the distribution of regret ratio in the final configurations.}

\algor\ and \algod\ consistently have the lowest regret ratios. Among the other approaches, PER incurs the highest regret ratios for users in all datasets, since it does not foster social interactions in shared views. Consistent with the performances on \scenario\ utility, FMG and SDP outperform PER and GRF in the Timik dataset, but their performances are comparable in Epinions because the sparser social relations in the review network generates lower social utility. Interestingly, in Timik, some users in GRF are highly satisfied (as the CDF of GRF matches well with that of \algor\ and \algod\ from the beginning) but some others have very high regret ratios (as the CDF dramatically rises near the end). This indicates that a portion of the users in GRF may actually have their preferences sacrificed, i.e., they are forced to share views on uninterested items with other subgroup members. In contrast, FMG and SDP show flatter CDFs, implying the user preferences are more balanced. However, users in FMG and SDP consistently have regret ratios over 20\%, while the regret ratios seldom exceed 20\% in \algor/\algod; this is because the randomly chosen pivot parameters (in \algo) and the deterministically optimized one (in \algod) can effectively form dense subgroups with similar item preferences. More experimental results on subgroup metrics can be found in Section 6.5 in the full version \cite{Online}.}

\opt{full}{ 
In the following, we analyze the subgroups in the \configuration s returned from all algorithms\footnote{Some of the results of \algod\ are omitted as they are similar to \algor.} in terms of subgroup-related metrics. Figures \ref{fig:timik_inter_intra_density}, \ref{fig:ep_inter_intra_density}, and \ref{fig:yelp_inter_intra_density} compare the ratios of Inter-/Intra-subgroup edges averaged across all slots in the \configuration s in the Timik, Yelp, and Epinions datasets, respectively. It also shows the average network density among the partitioned subgroups normalized by the original density of the input social network. The results from all datasets indicate that the majority of preserved edges by \algo\ are within the same subgroups (large Intra\%). FMG achieves 100\% in Intra\%, 0\% in Inter\%, and 1 in normalized density because it consistently views the whole network as a large subgroup. In contrast, PER has a high Inter\% as it separates most users into independent subgroups to display their favorite items. This phenomenon is more obvious in Yelp (with a 100\% Inter\%, i.e., \textit{every} user is left alone as a tiny subgroup) than in Timik (about 30\% of the social edges are Intra-subgroup edges) or Epinions (a nonzero Intra\%). This is because Yelp is a product-review application, where the POIs (restaurants, stores, service centers, etc.) are highly diversified, making it more difficult for different users to have aligned preferences, e.g., user A has her $k$-th favorite POI (indicated by check-in records) identical to user B's $k$-th favorite POI, such that PER happens to co-display the POI to them. On the other hand, there exist a small subset of widely liked or adopted items in Epinions that appear as most users' favorite (hence the small nonzero Intra\% of PER), while famous VR locations (e.g., transportation hubs in Timik) are inclined to be associated with high preference utilities by the exploited recommendation learning models as they generate a lot of check-ins among all users. Therefore, they have a higher chance to be co-displayed by PER.

Among all methods, \algo\ achieves the largest normalized density as \stagethree\ carefully considers the \weight s to partition the network into dense communities. The normalized density achieved by \algo\ in Yelp is the lowest among the three datasets due to the aforementioned issue of diversified interests but is still higher than all other algorithms. As a result, \algo\ has the most abundant social connections among subgroup members, which usually trigger enthusiastic discussions and purchases. It is worth noting that VR users generally interact with more strangers during the trip, whereas users in traditional LBSN mainly interact with friends in their spatial proximity. Therefore, the local community structures in Timik are less apparent than those of Yelp. As such, despite adopting different partitioning criteria, the average densities of the subgroups retrieved by SDP and GRF do not differ much in Timik.

Figures \ref{fig:timik_share_alone}, \ref{fig:yelp_share_alone}, and \ref{fig:ep_share_alone} illustrate the Co-display\% and Alone\% among the returned \configuration s. As shown, \algo\ has a high Co-display\% near $1.0$ (implying almost all pairs of friends are sharing the views of common items) and a near-zero Alone\% (showing that almost no users are left alone in the configuration) in all datasets. Note that the Co-display\% is based on friend pairs and the Alone\% is calculated based on all users, i.e., they are not complementary statistics (thus do not sum up to 100\%). 
Although FMG achieves 100\% in Co-display\% and 0\% in Alone\%, it sacrifices the preferences of group members by forming huge subgroups for co-display. On the other hand, both \algo\ and GRF are able to maintain high values of Co-display\% while taking into account the personal preferences. In all datasets, GRF leaves a considerably high portion of users alone because it forms the subgroups according to item preferences. Therefore, some users with unique profiles of interests are more inclined to be left alone. 
The only other method with a higher Alone\% than GRF is PER, which does not facilitate shared views. 

Figures \ref{fig:timik_regret}, \ref{fig:yelp_regret}, and \ref{fig:ep_regret} report the CDF of regret ratios of all algorithms in the the Timik, Yelp, and Epinions datasets. The regret ratio reg$(u)$ \cite{KL14KDD} is a game-theory based measurement for satisfaction of individual users and the overall \textit{fairness} of the solution. For each user $u$, the regret ratio reg$(u)$, and its converse, \textit{happiness} ratio $\text{hap}(u)$, are defined as
$$\text{reg}(u) \equiv 1 - \text{hap}(u), \quad \text{hap}(u) \equiv \frac{\sum\limits_{c \in \mathbf{A}^{\ast}(u,\cdot)} \text{w}_{\textbf{A}}(u,c)}{\max\limits_{\mathcal{C}_u} \sum\limits_{c \in \mathcal{C}_u} \Bar{\text{w}}_{\textbf{A}}(u,c)}$$
\noindent where the numerator in $\text{hap}(u)$ is the achieved \scenario\ utility, the denominator with $\Bar{\text{w}}_{\textbf{A}}(u,c) = (1-\lambda)\cdot\text{p}(u,c) + \lambda\cdot\sum_{v \in V} \tau(u,v,c)$ is an \textit{upper bound} of possible \scenario\ utility when all users view $c$ together with user $u$, and $\mathcal{C}_u$ is a $k$-itemset, corresponding to a very optimistic scenario favoring $u$ the most. Note that the second term of $\Bar{\text{w}}_{\textbf{A}}(u,c)$ is different from $\text{w}_{\textbf{A}}(u,c)$ in Definition \ref{def:savg_utility}. In other words, the upper bound is the \scenario\ utility of $u$ if she dictates the whole \configuration\ selfishly. Therefore, a high $\text{hap}(u)$ (equivalently, a low reg$(u)$) implies that user $u$ is relatively satisfied with the \configuration, and \textit{fairness} among the users can be observed by inspecting the distribution of regret ratio in the final configurations.
\algor\ and \algod\ consistently have the lowest regret ratios. Among the other approaches, PER incurs the highest regret ratios for users in all datasets, since it does not foster social interactions in shared views. Consistent with the performances on \scenario\ utility, FMG and SDP outperform PER and GRF in the Timik and Yelp datasets, but their performances are comparable in Epinions because the sparser social relations in the review network generates lower social utility. Interestingly, in Timik and Yelp, some users in GRF are highly satisfied (as the CDF of GRF matches well with that of \algor\ and \algod\ from the beginning) but some others have very high regret ratios (as the CDF dramatically rises near the end). This indicates that a portion of the users in GRF may actually have their preferences sacrificed, i.e., they are forced to share views on uninterested items with other subgroup members. In contrast, FMG and SDP show flatter CDFs, implying the user preferences are more balanced. However, users in FMG and SDP consistently have regret ratios over 20\%, while the regret ratios seldom exceed 20\% in \algor/\algod; this is because the randomly chosen pivot parameters (in \algo) and the deterministically optimized one (in \algod) can effectively form dense subgroups with similar item preferences. }

\opt{full}{
\subsection{A Case Study on Partitioning Subgroups} \label{subsec:case_study}

\begin{figure*}[tp]
  \centering \includegraphics[width = 0.95\textwidth]{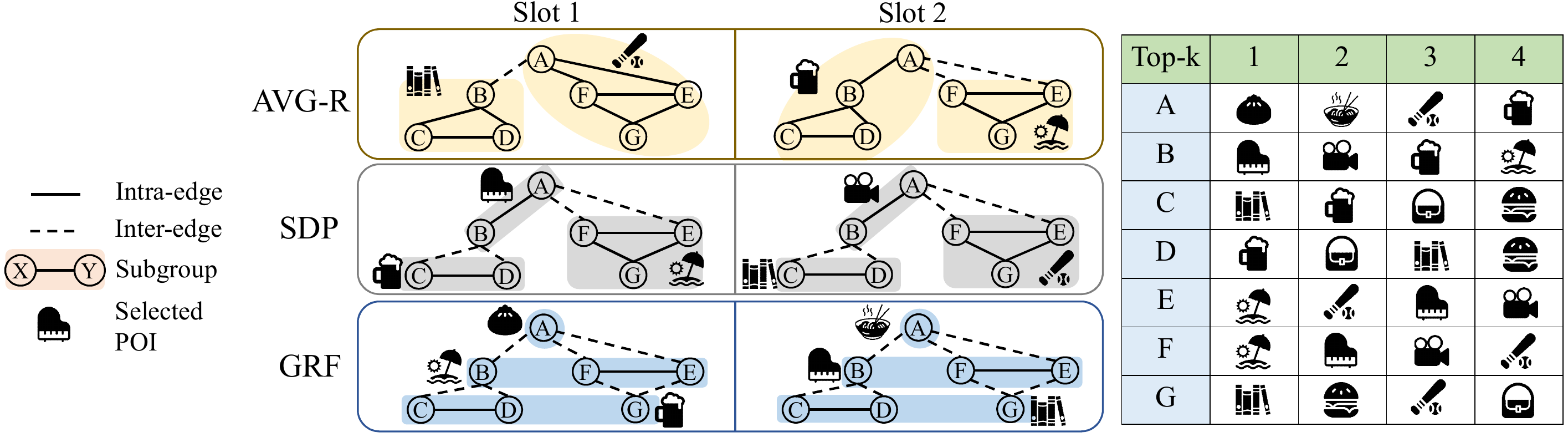}
  \caption{Illustration of subgroup partitioning approaches.}
  \label{fig:case_study}
\end{figure*}

Figure \ref{fig:case_study} depicts a 2-hop ego network (the subnetwork consisting of all 2-hop friends) of a user A in Yelp and two slots with the highest regret. Note that user A is studied here because she has a unique profile of preference that does not resemble any of her friends (B, E, and F), as shown in the table listing the top-4 POIs of the preference utility for each user. The shaded areas illustrate the partitioned subgroups at the specific slot, while solid and dashed lines depict intra-subgroup and inter-subgroup social edges. The icon beside each subgroup represents the selected POI (displayed item). 

As indicated in the table, A has a unique profile of preference that does not resemble any of her friends B, E and F. Therefore, A and her friends are partitioned into distinct subgroups by GRF. On the other hand, \algo\ assigns a baseball field to $\{A,E,F,G\}$ at slot 1, and a bar to $\{A,B,C,D\}$ at slot 2. Therefore, \algo\ is able to capture different interests of A and select proper subgroups of friends for them. However, A is not co-displayed the library with $\{B,C,D\}$ at slot 1 and the beach with $\{E,F,G\}$ at slot 2 because none of the two POIs is in A's favor. On the other hand, SDP partitions the ego-network into three cliques based on the network topology and then assigns a music place and a movie theater to $\{A,B\}$, favoring B's but sacrificing A's interests. Moreover, A is inclined to be dissatisfied for not joining $\{E,F,G\}$ to the baseball field (the third preferred POI) at slot 2. GRF gathers the users with similar interests and leaves A alone, since none of her friends is a big fan of the two restaurants. Therefore, the regret ratios for A are 35.2\% (SDP), 41.2\% (GRF), and 19.6\% (\algo) accordingly. Indeed, this case study manifests that flexible partitions of subgroups are crucial the situations of friends not necessarily sharing similar interests. 
}

\opt{full}{
\subsection{Sensitivity Test of \algod} \label{subsec:diff_r}
\begin{figure}[tp]
	\centering
	\subfigure[][\centering Total \scenario\ utility vs. diff. $r$.] {\
		\centering \includegraphics[width = 0.44 \columnwidth]{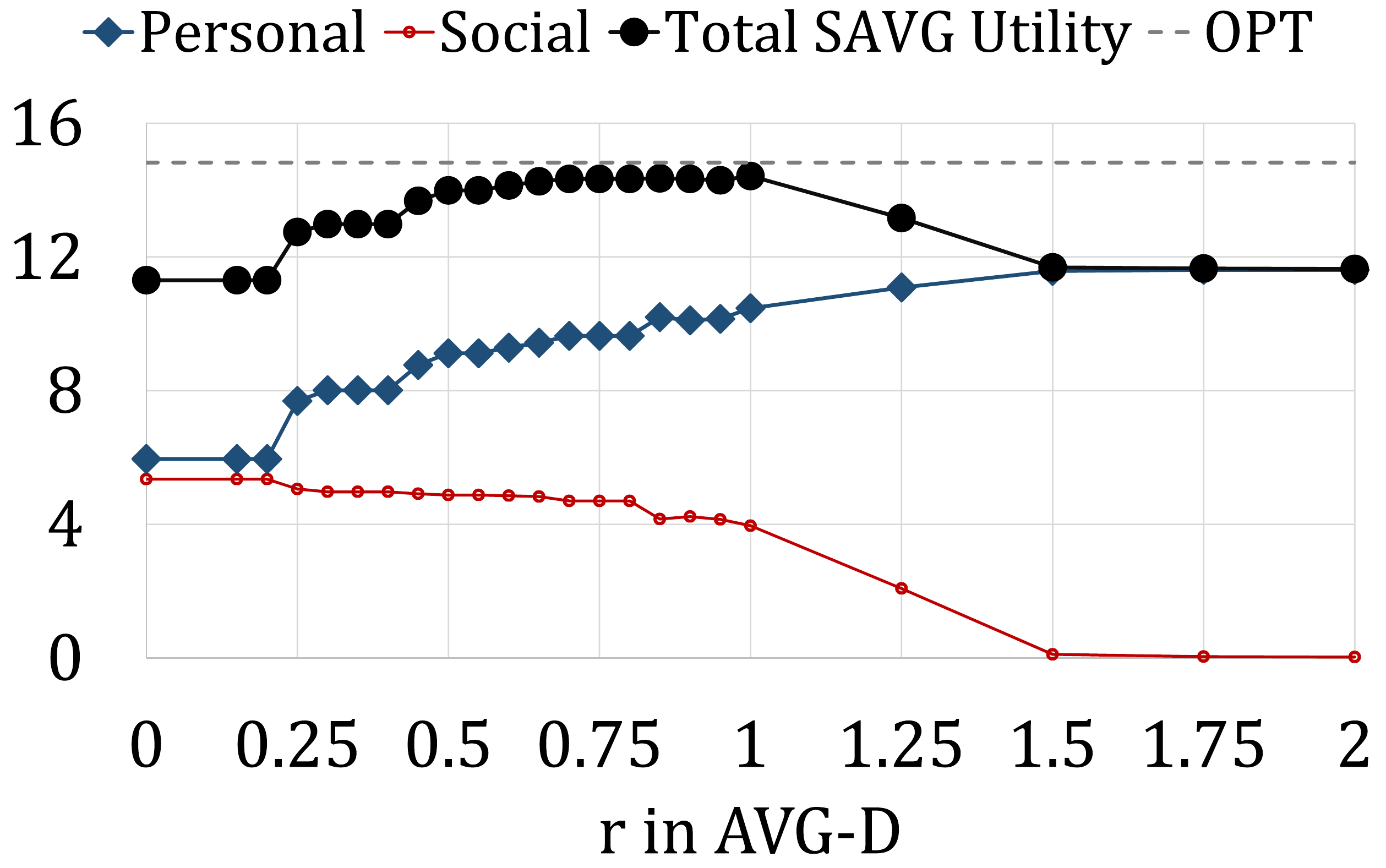}
		\label{fig:diff_r_obj} } 
	\subfigure[][\centering Execution time vs. diff. $r$.] {\
		\centering \includegraphics[width = 0.44 \columnwidth]{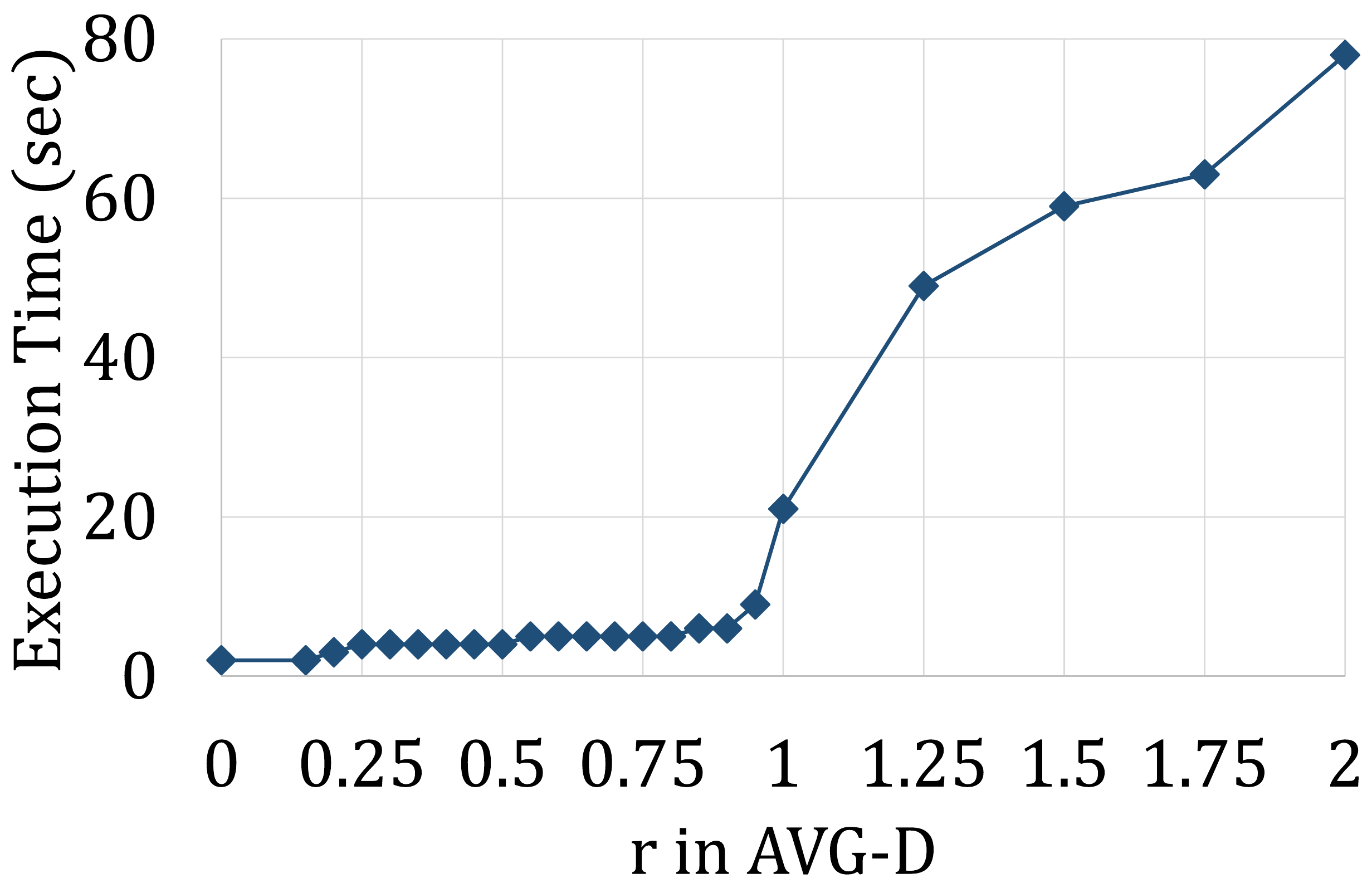}
		\label{fig:diff_r_time} } 
	\subfigure[][\centering Normalized density vs. diff. $r$.] {\
		\centering \includegraphics[width = 0.44 \columnwidth]{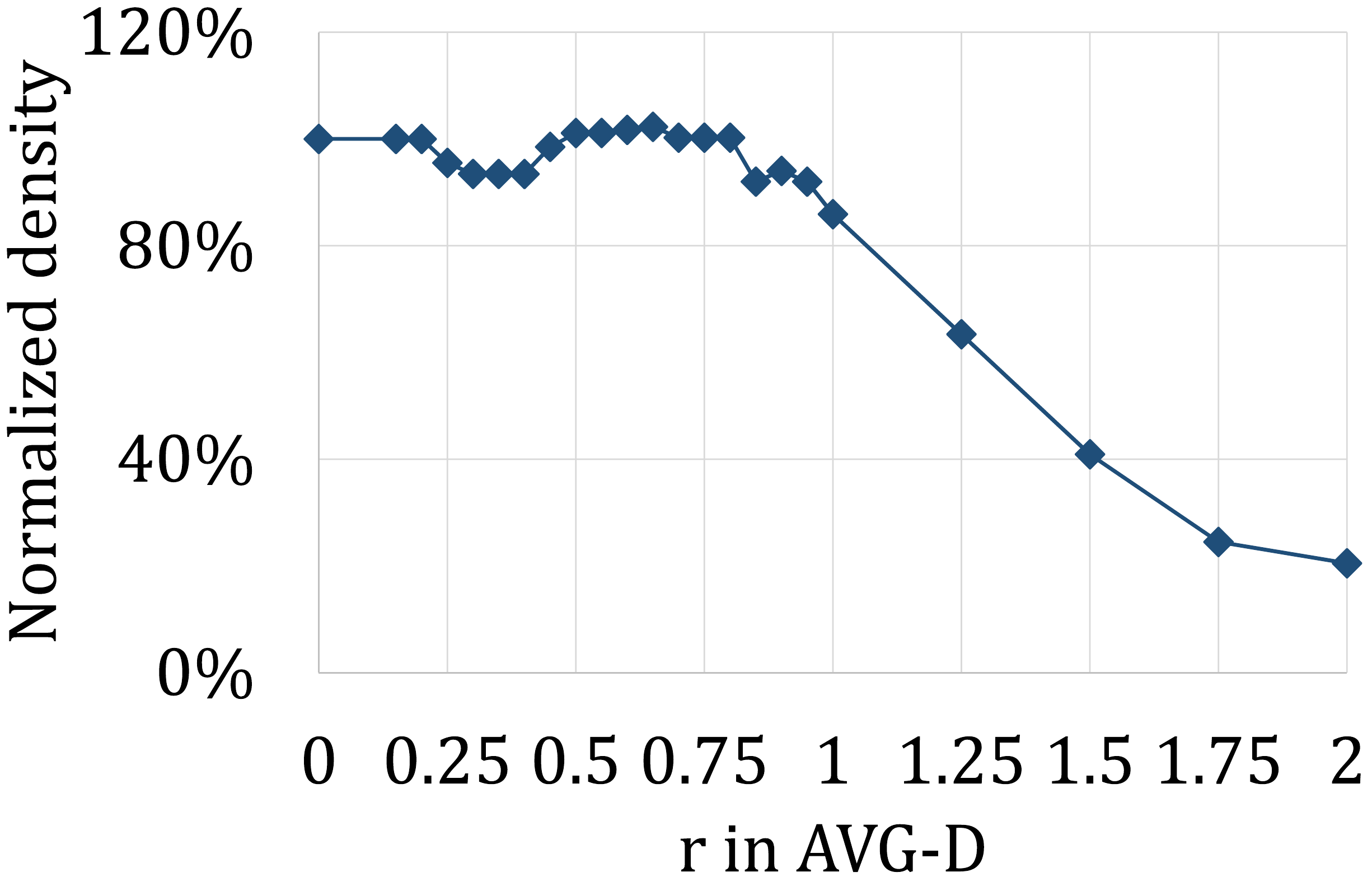}
		\label{fig:diff_r_density} } 
	\subfigure[][\centering Inter/Intra\% vs. diff. $r$.] {\
		\centering \includegraphics[width = 0.44 \columnwidth]{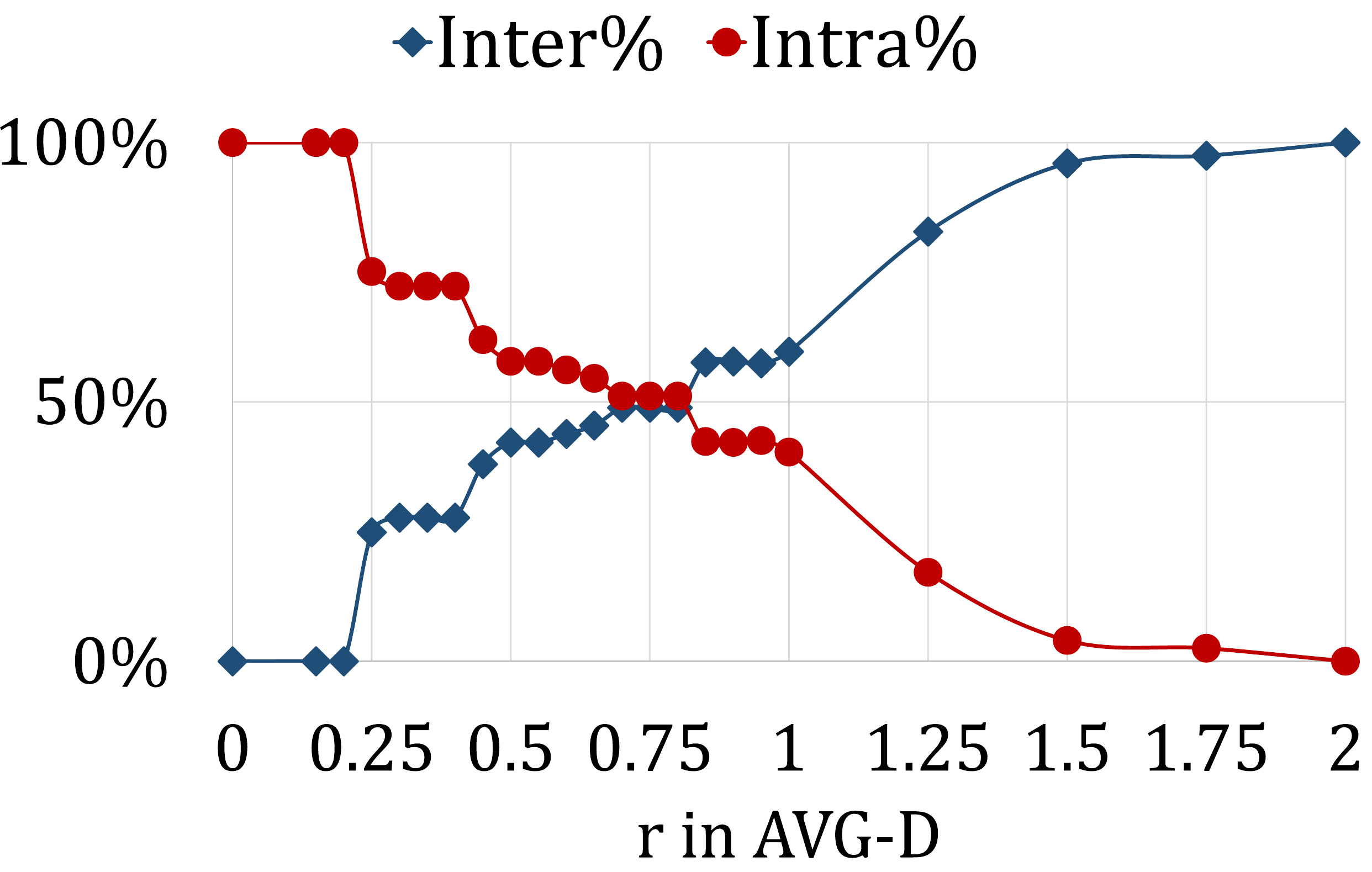}
		\label{fig:diff_r_inter_intra} } 

	\caption{Sensitivity test on $r$ in \algod.}
	\label{fig:diff_r}
\end{figure}
In the following, we compare \algod\ with the optimal solution in the Timik dataset. Figure \ref{fig:diff_r_obj} shows that with $r$ values from 0.7 to 1.0, \algod\ finds nearly-optimal solutions as \stagethreefull\ with $r$ balances between the current increment in \scenario\ utility and the potential \scenario\ utility in the future. This phenomenon is also observed in other datasets. However, setting $r = 0.25$ is still promising as it achieves 86.1\% of the optimal objective, which is consistent with the theoretical guarantee. Figure \ref{fig:diff_r_obj} also manifests that \algod\ with small values of $r$ tends to resemble the group approach in Section \ref{sec:intro} (i.e., displaying the same items to every user). For instance, \algod\ finds a large subgroup of all users with $r \leq 0.2$. \algod\ gradually becomes more alike to the personalized approach as $r$ grows (e.g., the social utility is close to 0 for $r\geq 1.5$, implying almost every user is displayed her own favorite items). This is actually consistent with the design of \algod, since 1) a small $r$ essentially ignores the potential utility gain in the future from allocating other items than the current \roundingitem, leading to a greedy behavior of always including every eligible user into the targeted subgroup in \stagethree. Consequently, running every iteration with $r$ close to 0 finds a large subgroup of all users in every slot; 2) a large $r$ prioritizes only the future gain and thereby tends to select very few users into the targeted subgroup. This reluctant behavior in the extreme (consider $r = \infty$) thus leads to subgroups consisting of only one user, essentially simulating the personalized approach. In summary, this again manifests that \algod\ can strike a good balance between the personalized and group approaches.

Figure \ref{fig:diff_r_time} depicts the execution time of \algod\ with different values of $r$. Intuitively, when \algod\ chooses fewer users in the current targeted subgroup, more iterations are required to construct a complete \configuration. For $n$ users and $k$ slots, the total number of iterations is $k$ if every targeted subgroup includes all $n$ users but would increase to $kn$ if every iteration displays the \roundingitem\ to only one user. Therefore, \algod\ has a larger execution time with larger values of $r$. Figures \ref{fig:diff_r_density} and \ref{fig:diff_r_inter_intra} further show the average normalized density and Inter/Intra\% among the subgroups in the solutions of \algod\ with different values of $r$. Consistent with the above discussion, the behavior of \algod\ essentially spans a spectrum from the personalized to the group approach.
}

\opt{full}{
\begin{figure}[tp]
	\centering
	\subfigure[][\centering Total violation vs. subgroup size constraint (Timik, $n=25$).] {\
		\centering \includegraphics[width = 0.44 \columnwidth] {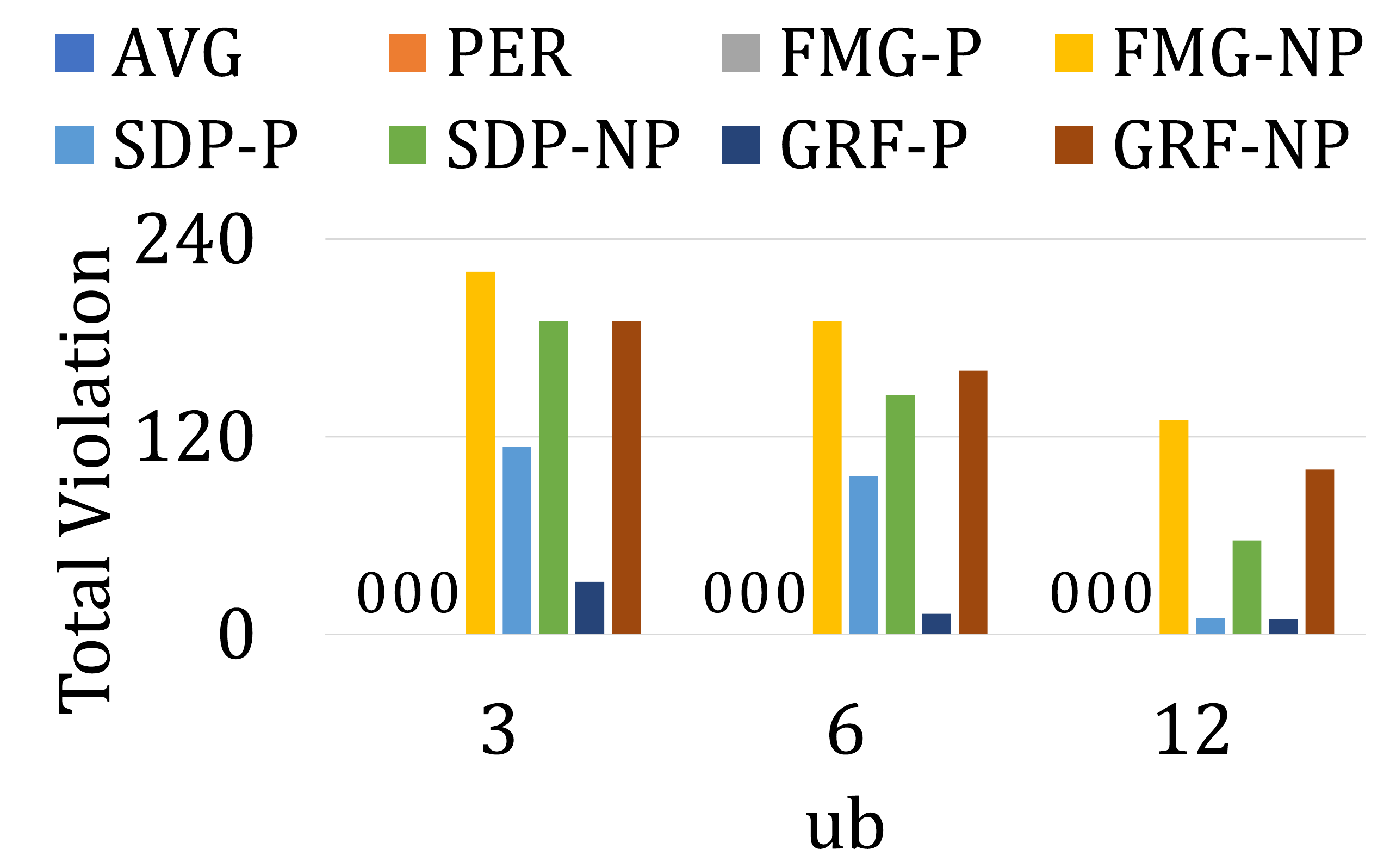}
		\label{fig:violation_M_Timik25} } 
	\subfigure[][\centering Total violation vs. subgroup size constraint (Epinions, $n=15$).] {\
		\centering \includegraphics[width = 0.44 \columnwidth] {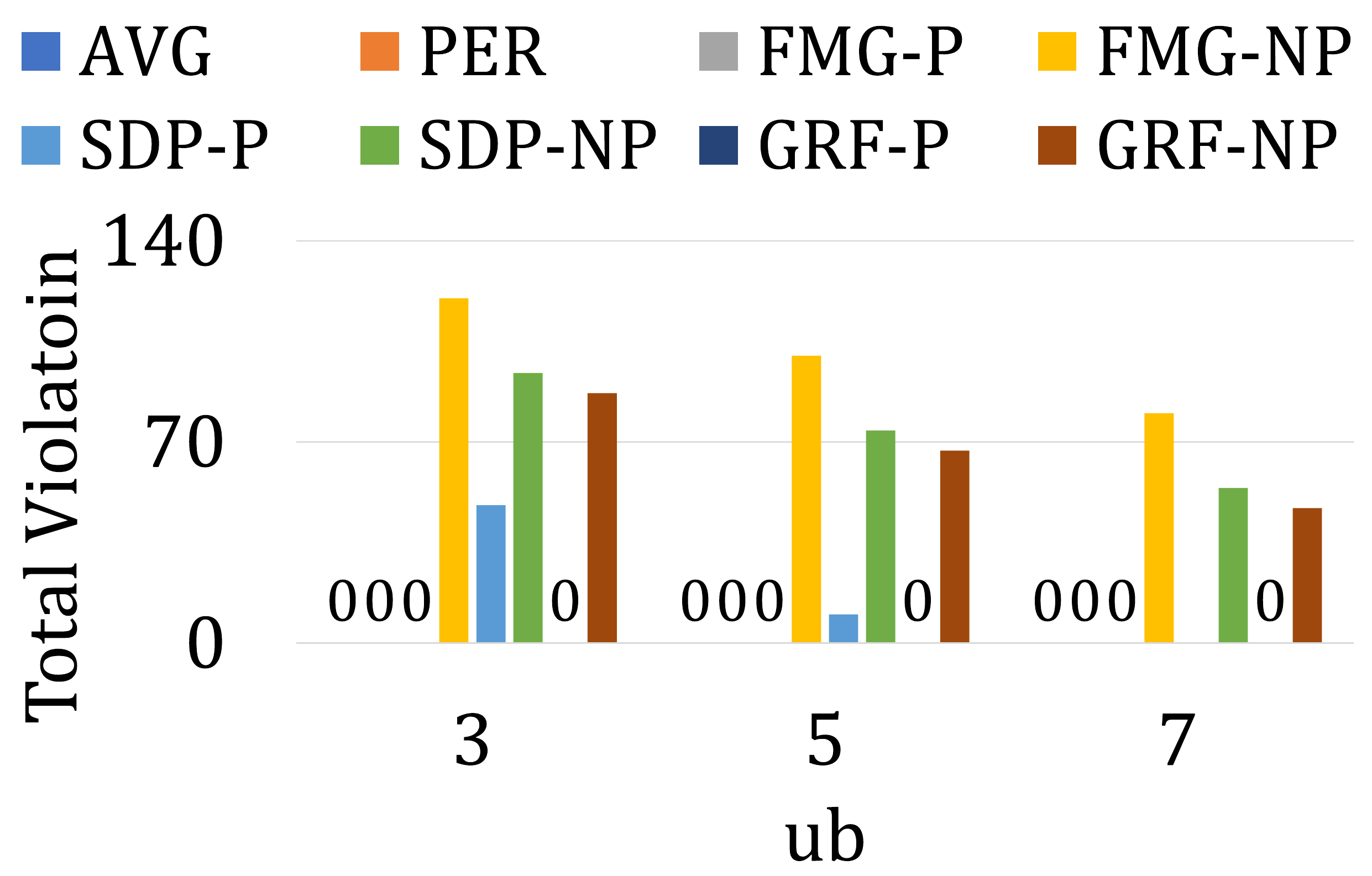}
		\label{fig:violation_M_Ep15} }
	\caption{Comparisons on \probtwo.}
	\label{exp:subgroup_constraint}
\end{figure}

\begin{figure}[tp]
    \centering
    \includegraphics[width = 0.9 \columnwidth] {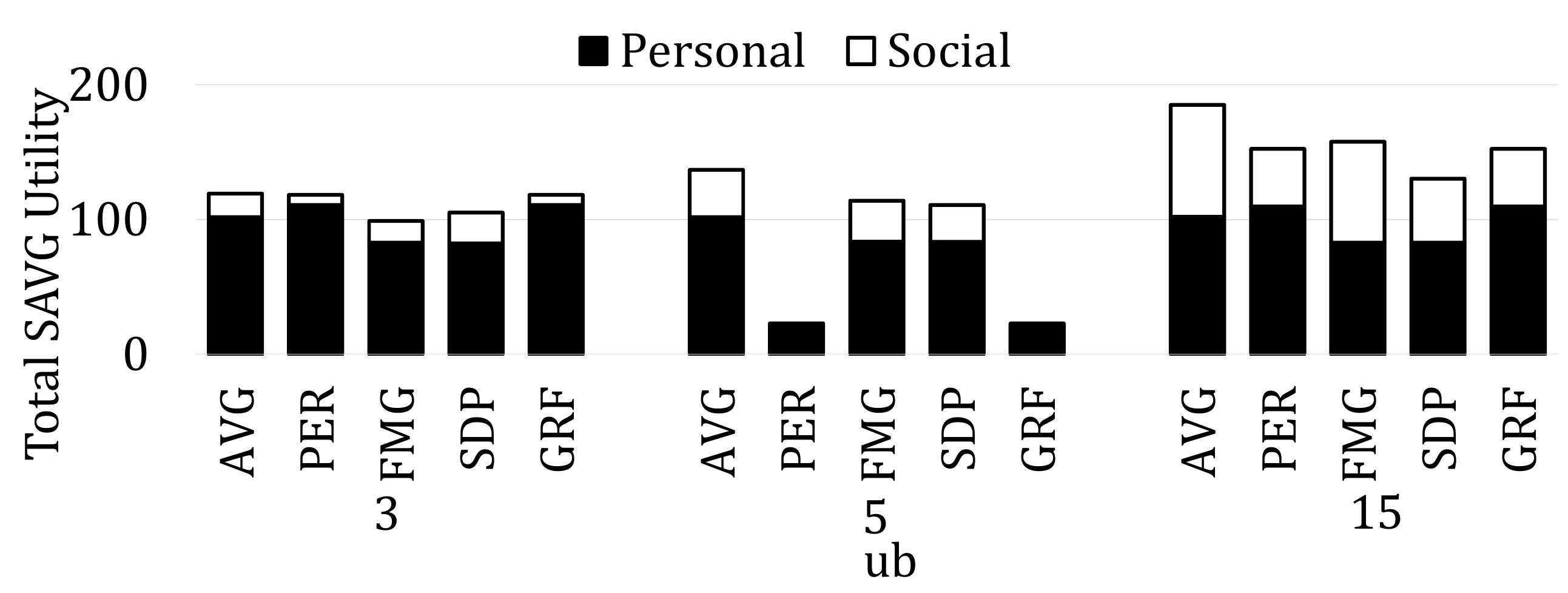}
    \caption{Total \scenario\ utility vs. subgroup size constraint (Timik, $n=15$).}
    \label{fig:obj_M_Timik15}
\end{figure}

\begin{figure}[tp]
    \centering
    \includegraphics[width = 0.9 \columnwidth] {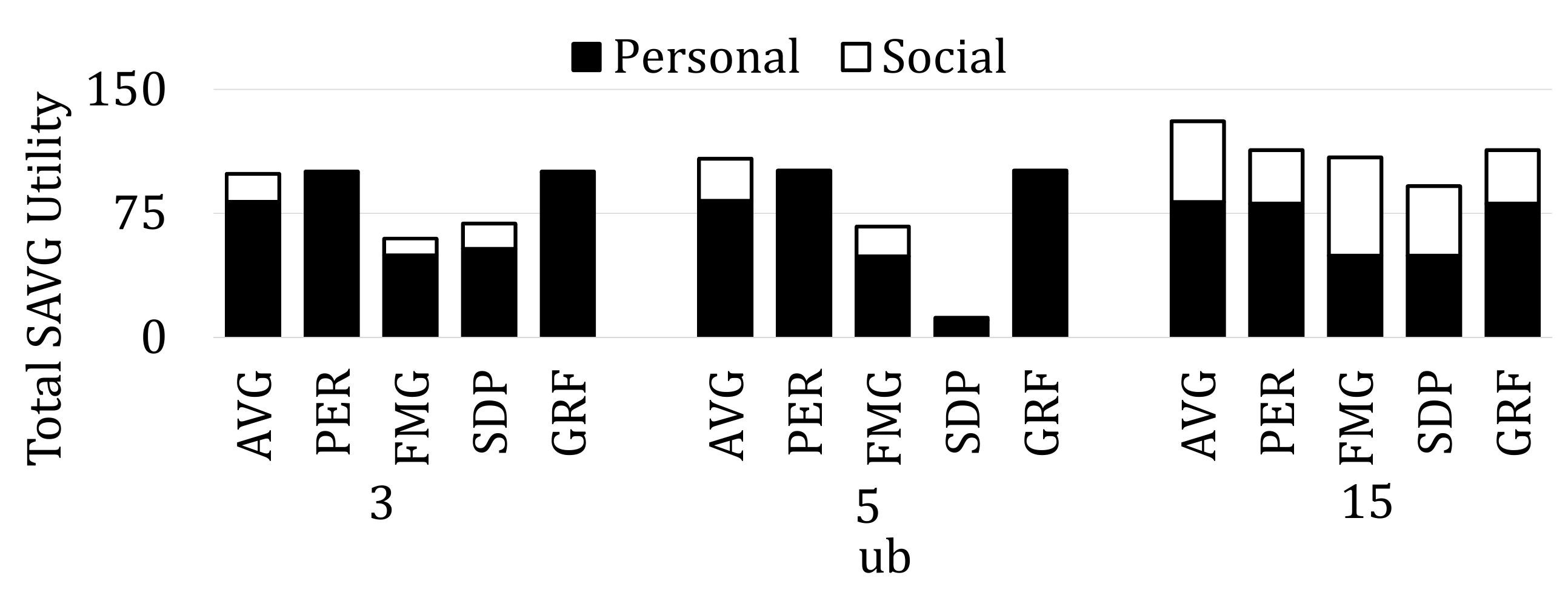}
    \caption{Total \scenario\ utility vs. subgroup size constraint (Epinions, $n=15$).}
    \label{fig:obj_M_Ep15}
\end{figure}

\subsection{Experimental Results for \probtwo} \label{subsec:subgroup_constraint}
In the following, we investigate the behavior of all methods in \probtwo, where the default discount factor is 0.5, and the subgroup size constraint $M$ is varied. As none of the baseline algorithms considers the subgroup size constraint, for \probtwo\ with subgroup size constraint $M$, we first pre-partition the user set in to $\lceil \frac{N}{M} \rceil$ subgroups with balanced sizes. The results of all baseline algorithms on the partitioned subproblems are then aggregated into the final returned \configuration\ for them. Figures \ref{fig:violation_M_Timik25} and \ref{fig:violation_M_Ep15} show the total violation of subgroup size constraint (in total number of users) of all methods aggregated over all display slots in a total of 10 sampled instances in Timik (with $n=25$) and Epinions (with $n=15$), respectively. The suffix ``-P'' in the method name represents prepartitioning of the user set into $\lceil \frac{N}{M} \rceil$ subgroups with balanced sizes, and the suffix ``-NP'' indicates no prepartitioning is applied. The results manifest that the pre-partitioning helps decrease the total violation of the subgroup size constraint for all baseline methods except PER because the methods do not take the cap on subgroup sizes into account. \algo\ never violates the constraint since the greedy cutoff in \stagethree\ prevents large targeted subgroups. PER also achieves 100\% feasibility since it does not consider social interactions. Among the baseline methods that consider social interactions, GRF incurs the lowest violation since only it partitions the user set into subgroups based on preference. However, it still has a low feasibility if not used with the pre-partition technique.

Figures \ref{fig:obj_M_Timik15} and \ref{fig:obj_M_Ep15} compares the total \scenario\ utility achieved by all methods \textit{with pre-partitioning} in Timik and Epinions, respectively, with $n=15$ and the subgroup size constraint varies from 3 to 15, where infeasible solutions achieves a total \scenario\ utility of 0. Note that all the baseline methods, except for PER, could still violate the subgroup size constraint even incorporated with the prepartitioning technique. This is because they from time to time display the same item to different pre-partitioned subgroups (or, in the case of GRF, subgroups of them) at the same display slot. \algo\ consistently outperforms all other methods except for the cases where $M$ is very small (3 in Epinions). GRF achieves high total \scenario\ utility in Epinions as it achieves 100\% feasibility and also selects proper distinct and preferred items for separated small subgroups of users in the sparse Epinions network. However, GRF has a surprisingly low feasiblity in Timik (around 20\%), leading to a low total \scenario\ utility. This is because GRF displays the commonly-preferred popular VR locations in Timik to almost all pre-partitioned subgroups at the first few display slots due to its greedy (and thus deterministic) algorithmic behavior in selecting display items for subgroups. Therefore, while its accumulative violation is the lowest (among all methods that possibly violate the constraint), it often finds slightly oversized subgroups at slots 1 and 2, and thereby has a low feasibility.
}

\begin{figure}[tp]
	\centering
	\subfigure[][Percentage of diff. $\lambda$.] {\
		\centering \includegraphics[width = 0.44 \columnwidth] {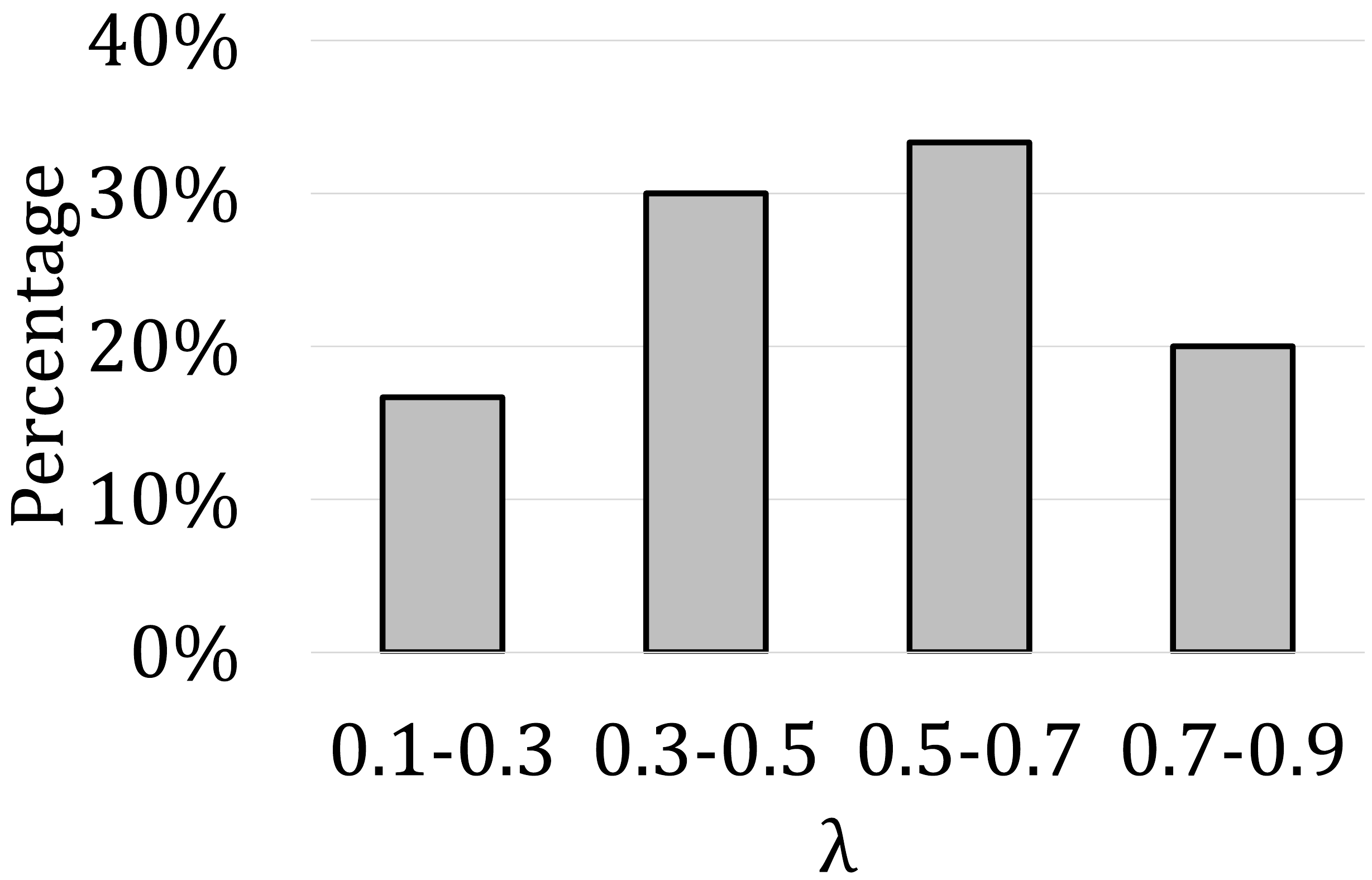}
		\label{fig:user_lambda} }
	\subfigure[][Total \scenario\ utility and user satisfaction.] {\
		\centering \includegraphics[width = 0.44 \columnwidth] {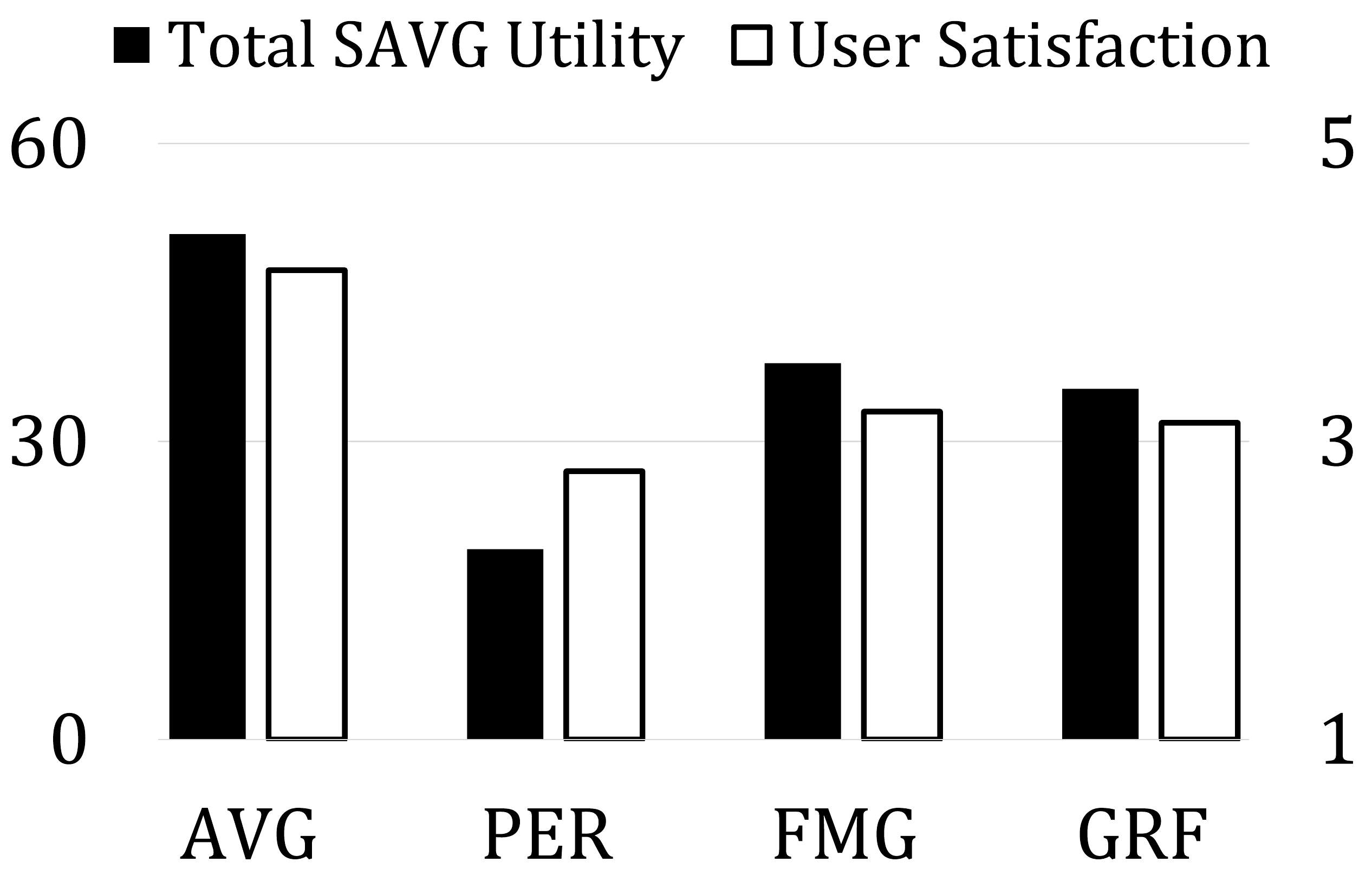}
		\label{fig:user_obj} }
	\subfigure[][Inter/Intra\% and subgroup density.] {\
		\centering \includegraphics[width = 0.44 \columnwidth] {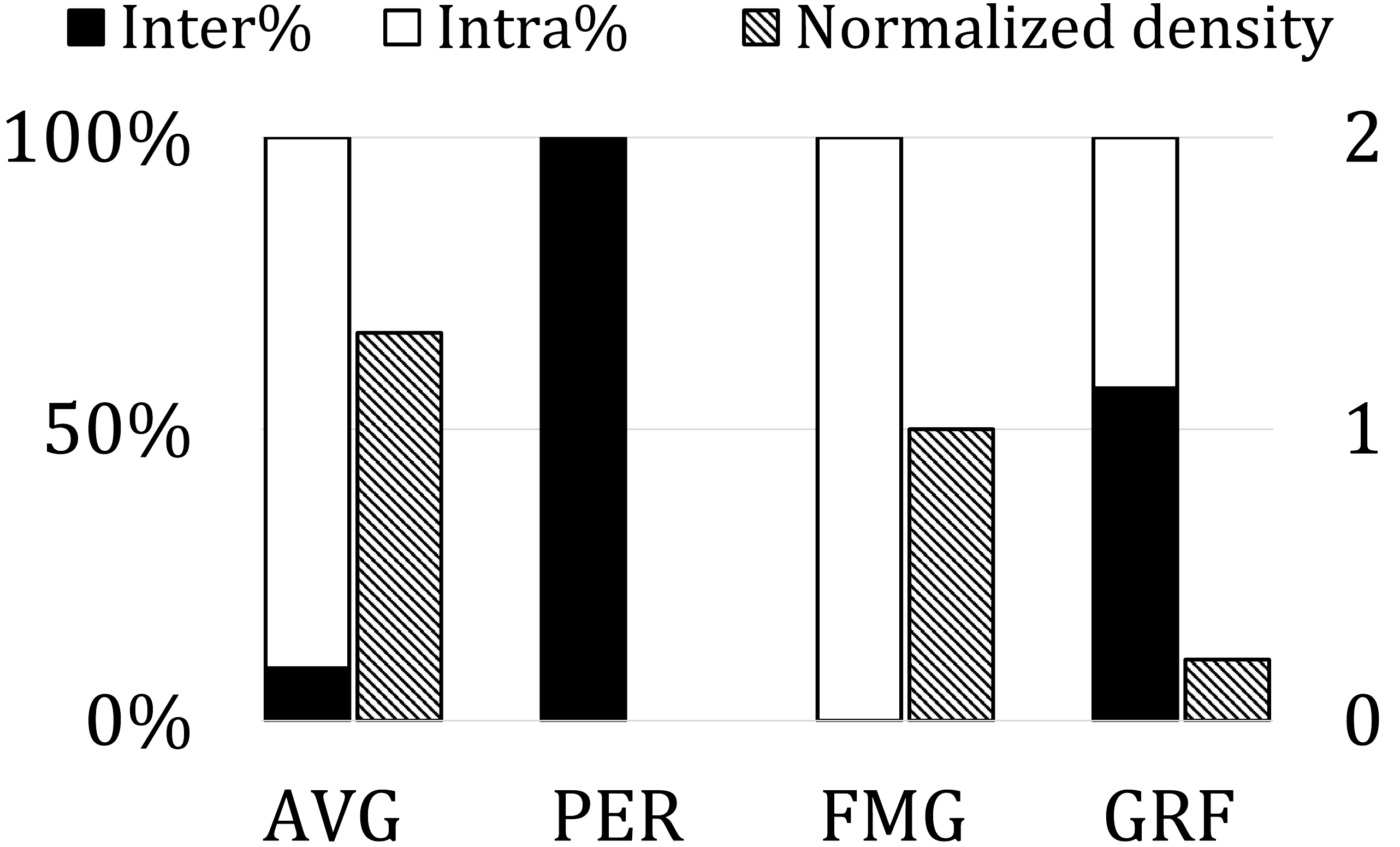}
		\label{fig:user_density} }
	\subfigure[][Co-display rate and alone rate.] {\
		\centering \includegraphics[width = 0.44 \columnwidth] {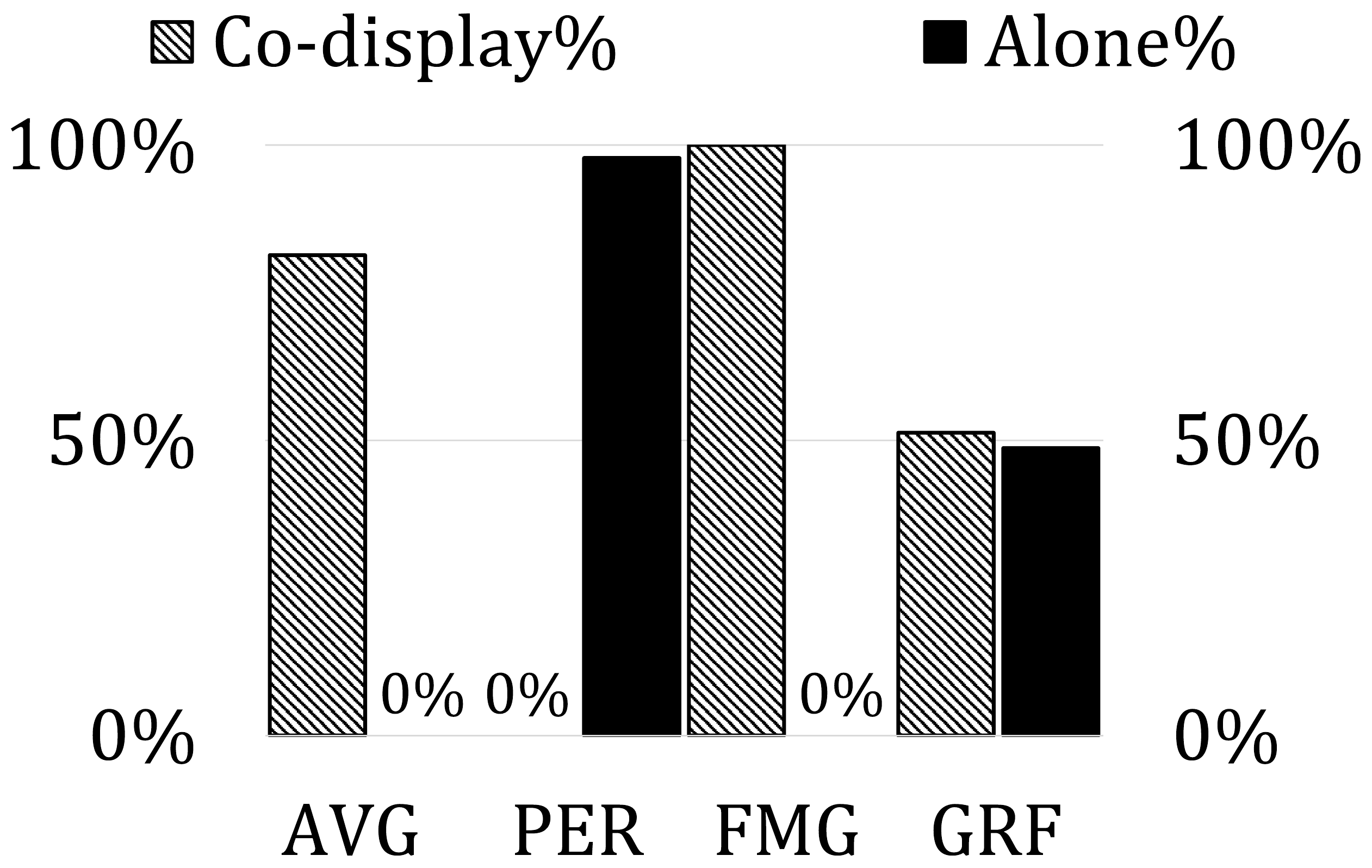}
		\label{fig:user_share_alone} }
	\caption{Comparisons on user study.}
	\label{exp:user_study}
\end{figure}

\subsection{User Study} \label{sec:user_study}

For the user study, we recruit 44 participants
to visit our VR store. Their social networks, preference utility, and $\lambda$ are pre-collected with questionnaires, which follows the setting of \cite{SBR15}, and the social utility is learned by PIERT \cite{YL18TOIS}. A Likert scale questionnaire \cite{likert} is used to find the preference utility of items, where users are allowed to discuss the products so that the social utility can be learned. 
Finally, they are asked to provide $\lambda$ in [0,1]. We investigate the following research question: 
After experiencing the VR store, are the participants satisfied with the \configuration s generated by \algo, PER, FMG, and GRF? User feedbacks are collected in Likert score \cite{likert} from 1 to 5 (very unsatisfactory, unsatisfactory, average, satisfactory, and very satisfactory). Each group of participants visits the VR stores twice via hTC VIVE with the items selected by each scheme in randomized order. 

Figure \ref{fig:user_lambda} reports that $\lambda$ values specified by the users range from 0.15 to 0.85 with the average as 0.53, indicating that both personal preferences and social interactions are essential in VR group shopping. 
Figure \ref{fig:user_obj} compares the total \scenario\ utility as well as the recorded user satisfaction of each method. \algo\ outperforms the baselines by at least 34.2\% and 29.6\% in terms of the average total \scenario\ utility and average user satisfaction, respectively. The difference of \algo\ is statistically significant (p-value $\leq$ 0.019 $<$ 0.05). 
It is worth noting that the correlation between the \scenario\ utility and user satisfaction is high (Spearman correlation 0.835; Pearson correlation 0.814), which manifest that the \scenario\ utility is a good estimation of user satisfaction.

Figures \ref{fig:user_density} and \ref{fig:user_share_alone} report the subgroup metrics in the user study datasets. GRF, which separates users into subgroups according to preference similarities, returns a low normalized density (0.21), i.e., users in the same subgroup tends to be strangers. Compared with the results in large-scale datasets (Figure \ref{exp:group_all}), GRF performs worse here since the normalized density is more sensitive when the user set is relatively small. In contrast, AVG flexibly assigns proper items to different subgroups of friends such that the normalized density is greater than 1 and the alone rate is 0\%.

\section{Conclusion} \label{sec:conclusion}

To the best of our knowledge, there exists no prior work tackling flexible configurations under the envisaged scenario of VR group shopping. In this paper, we formulate the \prob\ problem to retrieve the optimal \configuration\ that jointly maximizes the preference and the social utility, \revise{and prove \prob\ is $\mathsf{NP}$-hard to approximate within $\frac{32}{31} - \epsilon$}. We introduce an IP model and design a novel 4-approximation algorithm, \algo, and its deterministic version, \algod, by exploring the idea of \stagethreefull\ (\stagethree) that forms subgroups of friends to display them the same items. 
Experimental results on real VR datasets manifest that our algorithms outperform baseline approaches by at least 30.1\% in terms of solution quality.

\bibliographystyle{abbrv}
\bibliography{bibliography/ref.bib}
\balance

\end{document}